\documentclass[11pt]{article}
\usepackage{fullpage}
\usepackage{latexsym}
\usepackage{amsfonts}
\usepackage{amssymb}
\usepackage{amsmath,amsthm}
\usepackage{bm}
\usepackage{subcaption}
\usepackage{color}
\usepackage{graphicx, graphics}
\usepackage{hyperref}
\usepackage{xargs}
\usepackage{arydshln}
\usepackage{verbatim}

\DeclareMathOperator{\Res}{Res}

\newcommand{\ud}{\,\mathrm{d}}
\renewcommand\Re{\operatorname{Re}}
\renewcommand\Im{\operatorname{Im}}

\def\XXint#1#2#3{{\setbox0=\hbox{$#1{#2#3}{\int}$}
\vcenter{\hbox{$#2#3$}}\kern-.5\wd0}}

\definecolor{green}{RGB}{0,128,0}

\newcommand{\CC}{\mathbb{C}}

\newcommand{\QQ}{\mathbb{Q}}
\newcommand{\RR}{\mathbb{R}}
\newcommand{\ZZ}{\mathbb{Z}}

\newtheorem{thm}{Theorem}[section]

\theoremstyle{definition}

\theoremstyle{remark}
\newtheorem*{rmk}{Remark}
\newtheorem{eg}[thm]{Example}

\graphicspath{{gfx/}}

\newcommand{\BE}{\begin{equation}}                                 
\newcommand{\EE}{\end{equation}}                                   
\newcommand{\BES}{\begin{equation*}}                               
\newcommand{\EES}{\end{equation*}}                                 
\newcommand{\D}{\ensuremath{\,\mathrm{d}}}						   

\def\Eq#1$$#2$${\BE\label#1#2\EE}
\def\eq#1{\eqref{#1}}

\def\ostrut#1#2{\hbox{\vrule height #1pt depth #2pt width 0pt}}
 \def\ustrut#1{\ostrut{#1}0}

\newdimen\eqjot \eqjot = 1\jot
\def\openupeq{\openup \the\eqjot}
\def\addtab#1={#1\;&=}
\def\ezeq#1#2#3{{\def\\{\cr#1}\vcenter{\openupeq \halign{$\displaystyle 
  \hfil##$&$\displaystyle##\hfil$&&\hskip#2pt$\displaystyle##\hfil
    $\cr#1#3\cr}}}}
\def\eaeq{\ezeq\addtab}

\def\eeq{\eaeq{20}}  
\def\cthn#1#2{{}\cr&\hskip#1pt{}#2}

\def\qaeq#1#2{{\def\\{&}\vcenter{\openupeq\halign{$\displaystyle
  ##\hfil$&&\hskip#1pt$\displaystyle##\hfil$\cr #2\cr}}}}
 
\def\req{\qaeq{30}} \def\qxeq{\qaeq{30}} 
\def\qeq{\qaeq{20}} 
\def\xeq{\qaeq{10}}  

\def\pa#1{(#1)}
\def\Pa#1{\left (\,{#1}\,\right )}

\def\bbk#1{\bigl [\,{#1}\,\bigr ]}

\def\fracz#1#2{#1/#2} \def\fracp#1#2{#1/(#2)}

\def\:{\mskip2mu}
\def\pii{\:\pi}
\def\i{\:{\rm i}\:}
\def\ip#1#2{\langle \,{#1}\,\ipsep{#2}\,\rangle}
\def\ipsep{;}

\let\udsty\displaystyle
\def\rbox#1{\hbox{\rm #1}}
\def\roq#1{\qquad \rbox{#1}\qquad }

\def\ZD{Z_ \Delta } \def\ZDp{\ZD^+} \def\ZDm{\ZD^-}
\def\L{{\rm L}}

\def\kap{\kappa}  
\def\alp{\alpha}

\def\hexnumber#1{\ifcase#1 0\or1\or2\or3\or4\or5\or6\or7\or8\or9\or
 A\or B\or C\or D\or E\or F\fi}
\edef\msbhx{\hexnumber\symAMSb}   
\mathchardef\emptyset="0\msbhx3F

\begin{document}

\title{Revivals and Fractalisation\\ in the Linear Free Space Schr\"odinger Equation}

\date{\today}
\author{Peter J. Olver$^1$, Natalie E. Sheils$^1$, David A. Smith$^2$ \\
\footnotesize 1. School of Mathematics, University of Minnesota, Minneapolis, MN 55455\\
\footnotesize 2. Division of Science, Yale-NUS College, 16 College Avenue West, \#01-220 138527, Singapore
}
\maketitle

\begin{abstract}
We consider the one-dimensional linear free space Schr\"odinger equation on a bounded interval subject to homogeneous linear boundary conditions.
We prove that, in the case of pseudoperiodic boundary conditions, the solution of the initial-boundary value problem exhibits the phenomenon of revival at specific (``rational'') times, meaning that it is a linear combination of a certain number of copies of the initial datum. 
Equivalently, the fundamental solution at these times is a finite linear combination of delta functions. At other (``irrational'') times, for suitably rough initial data, e.g., a step or more general piecewise constant function, the solution exhibits a continuous but fractal-like profile. 
Further, we express the solution for general homogenous linear boundary conditions in terms of numerically computable eigenfunctions.
Alternative solution formulas are derived using the Uniform Transform Method (UTM), that can prove useful in more general situations.
{We then investigate the effects of general linear boundary conditions, including Robin, and find novel ``dissipative'' revivals in the case of energy decreasing conditions.} 
\end{abstract}


\section{Introduction}


{The one-dimensional linear free space Schr\"odinger equation
\Eq{se} 
$$\i u_t = u_{xx} $$
is arguably the simplest dispersive (complex-valued) partial differential equation, possessing a quadratic dispersion relation $\omega(\kappa)=\kappa^2$ that relates the wave number (spatial frequency) $\kappa $ to the temporal frequency $\omega $ \cite{Whitham}. 
It arises as the most basic equation in Schr\"odinger's mathematical formalism of quantum mechanics~\cite{Schrodinger}. 
The term ``free space'' refers to the fact that it models quantum mechanical effects in a one-dimensional flat, empty space, i.e., with zero potential. 
For simplicity, from hereon we will usually refer to \eq{se} as the ``linear Schr\"odinger equation.'' 
It also arises as the linearisation of a variety of important nonlinear equations, most notably the (generalised)} nonlinear Schr\"odinger (NLS) equation
\Eq{NLS} 
$$\i u_t(t,x)+u_{xx}(t,x) \pm |u(t,x)|^p u(t,x)=0,$$ 
in which the particular case $p=2$ is a well-known integrable partial differential equation supporting soliton solutions \cite{DrazinJohnson}. 
As such, the NLS equation is important in a large variety of application areas --- in particular, whenever the modulation of nonlinear wave trains is considered. 
Indeed, it has been derived in such diverse fields as waves in deep water~\cite{Zak}, plasma physics~\cite{Zak2}, nonlinear fibre optics~\cite{HasegawaTappert1, HasegawaTappert2}, magneto-static spin waves~\cite{ZvezdinPopkov}, and many other settings. 
The linear Schr\"odinger equation \eq{se} describes the behaviour of solutions of the NLS equation \eq{NLS} in the small amplitude limit, and hence a complete understanding of its dynamics is essential to the analysis of the more complicated nonlinear problem.

In the early 1990's, Michael Berry and collaborators~\cite{Berry,BerryKlein,BerryMarzoliSchleich} discovered the remarkable dynamics exhibited by solutions to the free space linear Schr\"odinger equation on circular domains, which they named after a 1835 optical experiment of the Victorian scientist William Henry Fox Talbot~\cite{Talbot}. 
The \emph{Talbot effect} is manifested in the quantum mechanical setting through the behaviour of rough solutions when subject to periodic boundary conditions. 
The evolution of a {piecewise constant} initial profile, e.g., a step function, ``fractalises'' into a continuous profile having a fractal form and dimension at irrational times (relative to the circumference of the circle) but ``quantises'' into other piecewise constant profiles at rational times; see also the introductory quantum mechanics text by Thaller, \cite{Thaller}. 
As shown in \cite{Olver}, similar phenomena appear in periodic solutions to linear evolution equations with ``integral'' polynomial dispersion relations, including the linearised Korteweg--deVries equation, as well as (integro-)differential equations with asymptotically polynomial dispersion relation, including the linearised Benjamin-Ono and integrable Boussinesq equations, although in the latter cases the quantised solutions are piecewise smooth but non-constant between cusps.
 The precise mathematical characterization of these cusped profiles remains mysterious. In \cite{ChenOlver}, these phenomena were shown to extend to nonlinear dispersive wave equations on periodic domains, both integrable and non-integrable, including the generalised NLS and Korteweg--deVries equations.

Another important manifestation of the Talbot effect is the phenomenon of \emph{revival}, which means that, at a rational time, the fundamental solution, i.e., that induced by a delta function initial condition, localises into a finite linear combination of delta functions.
This has the striking consequence that the solution to \emph{any} initial value problem at rational times is a finite linear combination of translates of the initial data and hence its value at any point on the periodic domain depends only upon finitely many of the initial values!
The term ``revival'' is based on the experimentally observed phenomenon of \emph{quantum revival}~\cite{BerryMarzoliSchleich, VrakkingVilleneuveStolow, YeazellStroud} in which an electron initially concentrated at a single location of its orbital shell is, at rational times, re-concentrated at a finite number of orbital locations. 

To date, these investigations have concentrated on periodic boundary conditions, and the main goal of the present paper is to move beyond the periodic cases and ascertain to what extent these phenomena depend upon the underlying boundary conditions.
We will fix our attention on the free space linear Schr\"odinger equation, leaving investigation of other models, including the linear Korteweg--deVries equation and nonlinear counterparts of both, to future research.

Note first that, for the second order Schr\"odinger equation,  Dirichlet, Neumann, or mixed boundary conditions can be recast as periodic boundary conditions by suitably doubling or quadrupling the length of the interval, and thus will exhibit similar revivals and fractalization. 
In such cases, the resulting periodic initial data is constructed from both shifts (translates) and reflections, under $x \mapsto -\:x$, of the original initial datum; hence the revival in such cases will also involve shifts and reflections. 
To move beyond such situations, we extend the analysis to more general types of boundary conditions. 
We are mostly concerned with pseudoperiodic boundary conditions, and show that the solutions to such problems exhibit the revival/fractalization  rational/irrational dichotomy, the rationality of the time depending on the length of the interval; see equation \eq{rattime}. 

The penultimate section will further extend the analysis to general homogeneous linear boundary conditions, which include Robin conditions.
Here we observe a range of interesting phenomena that depend upon the precise form of the boundary conditions. In particular, conservation of the energy ($L^2$ norm) appears to be required for revival and fractalization. 
Energy-dissipating boundary conditions lead to a novel form of revival, in which the translated and, possibly, reflected versions of the initial conditions are subjected to a dissipative-like decay in magnitude before being combined to form the solution at rational times. 
In particular, {by studying numerical experiments, it appears} that piecewise constant and piecewise linear initial data produces piecewise linear solutions at rational times of decreasing overall magnitude.
On the other hand, unstable and ill-posed systems produce exponentially growing modes that rapidly dominate the solution dynamics.

In the future, we plan to extend our analysis to the linear Schr\"odinger equation with a nonzero potential, which has not been investigated at all so far, even in the periodic case.
Furthermore, in accordance with the numerical results in \cite{ChenOlver}, we expect a similar range of phenomena to extend to the generalised NLS equation, both integrable and nonintegrable. 
Further investigations of the effect of boundary conditions on the solutions to the linearized and nonlinear Korteweg--deVries equation, and other dispersive equations with non-polynomial dispersion relations will also be of great interest.

\begin{rmk} This paper includes still shots of a variety of solutions at selected times. Julia code for generating the movies, which are even more enlightening, can be found on the web page:
\begin{center}\url{http://www.math.umn.edu/~olver/lseq.html}\end{center}
\end{rmk}

\section{The Linear Schr\"odinger Equation}

Most of the paper will be concerned with the linear Schr\"odinger equation subject to what we call \emph{pseudoperiodic boundary conditions} on an interval $0 \leq x \leq L$: 
\begin{equation}\label{eqn:ls_pp}
\eeq{  \i u_t+u_{xx}=0, &\ \,\beta_0 u(t,0)= u(t,L),
   \\
 u(0,x)=u_0(x), &
  \beta_1 u_x(t,0)=u_x(t,L),} 
  \qquad (t,x) \in[0,T]\times [0,L].
\end{equation}
Here $\beta_0,\beta_1$ are (possibly complex) constants. 
The previously studied case of periodic boundary conditions corresponds to $\beta_0=\beta_1=1$.

We apply generalised eigenfunction analysis to construct solution formulae and understand their resulting dynamical behaviour. 
As shown below, the eigenvalues $\lambda _j = -\i\kappa_j^2$
are determined by the vanishing of the \emph{discriminant} of the quasiperiodic problem, which will be seen to take the following form:
\Eq{discriminant}
$$\eeq{\Delta(\kappa )=e^{-\i\kappa L}(\beta_1+\beta_0)+e^{\i\kappa L}(\beta_1+\beta_0)-2(1+\beta_0\beta_1),\\
=2(\beta_0+\beta_1)\cos(\kappa L)-2(1+\beta_0\beta_1).}$$
Observe that the zeros of the discriminant are given by
\Eq{ZD} 
$$\ZD = \{\kappa\in\CC:\Delta(\kappa)=0\} = \{\kappa_j,-\kappa_j:j\in\ZZ\},$$
 where
\BE \label{eqn:Lambdaj.gen}
	\kappa_j = \kappa_0 + \frac{2j\pi}{L}, \roq{with} \kappa_0 = \frac{1}{L}\arccos\left( \frac{1+\beta_0\beta_1}{\beta_0+\beta_1} \right).
\EE
For simplicity of presentation, we will impose the further restriction that $\kappa_0 \neq j\pii/L$ for $j \in \ZZ$, which implies that the zeros of $\Delta$ are all simple and that $0\not \in \ZD$.
In this case, we decompose the discriminant locus into
\Eq{ZDpm}
$$\ZD = \ZDp \cup \ZDm $$
where
$$
\ZDp=\{\kappa\in\ZD:\Re(\kappa)>0\} ,\qquad \ZDm=\{\kappa\in\ZD:\Re(\kappa)<0\}.
$$

Given the discriminant roots \eq{eqn:Lambdaj.gen}, define
\Eq{gamma}
$$	\gamma  =e^{\i \kappa_jL} = e^{\i\kappa_0L}
=\frac{1+\beta_0\beta_1}{\beta_0+\beta_1} +\i\, \sqrt{1-\left(\frac{1+\beta_0\beta_1}{\beta_0+\beta_1}\right)^2} \ =\frac{1+\beta_0\beta_1+\i\, \sqrt{(\beta _1^2 -1) (1 - \beta_0^2)}}{\beta_0+\beta_1}\,,
$$
which is a complex number that is independent of $j\in\ZZ$.
Formula \eq{gamma} is valid  for all $\beta_0,\beta_1\in\CC$ such that 
$$\frac{1+\beta_0\beta_1}{\beta_0+\beta_1}\in\CC \setminus (-\infty,-1) \cup (1,\infty).$$
For values of $\beta_0,\beta_1$ on the branch cut, similar formulas for $\gamma$ are provided by (4.23.24) and (4.23.25) in~\cite{NIST}. 
Note that the definition of $\gamma$ implies
\Eq{gammaeq}
$$\frac{\gamma-\beta_0}{\beta_0-\gamma^{-1}}=\frac{\gamma-\beta_1}{\gamma^{-1}-\beta_1}$$ 
which we will use extensively below. 
Finally, for later use, we let
\Eq{deltatau}
$$\req{\delta =\frac{(\beta_0-\beta_1)}{\gamma((\beta_0+\beta_1)\gamma-2)},\\
\tau = \frac{(1+\beta_0\beta_1)(\gamma^2+1)-2\gamma(\beta_0+\beta_1)}{(\beta_0\gamma-1)(\beta_1\gamma-1)} .}$$


The initial-boundary value problem \eqref{eqn:ls_pp} is well-posed if and only if the zeros of $\Delta$ are all real, i.e., provided $\kappa_0 \in \RR$ is a real constant, and this will be assumed throughout.
(Indeed, if $\kappa_0$ is complex, then the system has an infinite number of unstable modes, and is hence ill-posed.)
If $\beta_0,\beta_1\in\RR$, the reality requirement reduces to $\lvert1+\beta_0\beta_1\rvert\leq \lvert\beta_0+\beta_1\rvert$.
In particular, the zeros are all real when the problem is self-adjoint with respect to the $\L^2$ inner product
\Eq{ip}
$$\ip fg = \frac{1}{L}\int_0^L f(x)\,g(x)\ud x,$$
which occurs if and only if
\Eq{ecbeta} 
$$\overline{\beta_0}\beta_1=1. $$
If this requirement is satisfied, the $L^2$ norm of the solution is constant, and we call the boundary conditions \emph{energy conserving}.


Given an ${\rm L}^2$ function $\phi \colon \RR \to \CC$, we denote its Fourier transform by
\Eq{Ftr}
$$	\widehat{\phi}(\kappa ) = \mathcal{F}[\phi](\kappa) = \int_{-\: \infty } ^\infty e^{-\i \kappa \: x}\phi(x)\ud x .$$
For a function $f \colon [0,L] \to \CC$ defined on a bounded interval, one extends it to vanish outside the interval, leading to its (bounded) Fourier transform
\Eq{FtrL}
$$	\widehat{f}(\kappa )=\mathcal{F}[f](\kappa) = \int_0^L e^{-\i \kappa \: x}f(x)\ud x .$$


\subsection{Revival at rational times}

For an initial-boundary value problem on an interval of length $L$, a time $t>0$  will be designated as \emph{rational} if 
\Eq{rattime}
$$t = \frac{L^2}{4\pii}\, r, \roq{where} r = \frac pq \in \QQ^+.$$
Note that one can also treat negative times by the simple discrete symmetry $ t \mapsto -t, \ u \mapsto \bar u$.
We always take the positive integers $p,q \in \ZZ^+$ to have no common factors. 
The extra factor of $4$ is introduced for later convenience, and is explained in detail in the following section.
Our main result, which follows from Theorem~\ref{thm:ShiftRep.uGenvGen} below, establishes the existence of \emph{revivals} for solutions $u(t,x)$ to the initial-boundary value problem \eqref{eqn:ls_pp} in which the solution profile at rational times is a finite linear combination of certain shifts (translations) and reflections (under $x \mapsto -\:x$) of the piecewise smooth initial datum $u_0(x) = u(0,x)$.

Given a rational time as in \eqref{rattime}, let
\Eq{alpha}
$$\alp = \exp(\pii\i /q)$$ 
be the primitive $2\:q$\textsuperscript{th} root of unity. 
Define the piecewise smooth functions $\phi, \psi \colon [0,L] \to \CC$ to be the following linear combinations of translates of the initial datum $u_0(x)$ on smaller subintervals. 
Given $x$ such that 
$$\Pa{1-\frac{\ell}{q}}L \leq x \leq \Pa{1-\frac{\ell-1}{q}}L, \roq{for some} \ell\in\{1,2,\ldots,q\},$$
we set 
\begin{subequations}\label{eqn:phipsi}
\begin{equation}
\begin{split}
		\phi(x)
			=\frac 1q\sum_{m=0}^{\ell-1} \gamma^{ -\frac{m}{q}} u_0\left(x+\frac{Lm}{q} \right) \sum_{n=0}^{q-1} \alpha^{-2nm-pn^2}
			+  \frac 1q\sum_{m=\ell}^{q-1} \gamma^{1-\frac{m}{q}} u_0\left(x+\frac{Lm}{q}-L\right) \sum_{n=0}^{q-1} \alpha^{-2nm-pn^2}
\end{split}
\end{equation}
\begin{equation}
\begin{split}
		\psi(x)
			=\frac 1q\sum_{m=0}^{\ell-1} \gamma^{\frac{m}{q} } u_0\left(x+\frac{Lm}{q} \right) \sum_{n=0}^{q-1} \alpha^{2nm-pn^2}+ \frac 1q\sum_{m=\ell}^{q-1} \gamma^{\frac{m}{q}-1} u_0\left(x+\frac{Lm}{q}-L\right) \sum_{n=0}^{q-1} \alpha^{2nm-pn^2}.
\end{split}
\end{equation}
\end{subequations}
When $\ell=q$, the above sums over $m$ from $\ell$ to $q-1$ are empty. 
Observe that the resulting functions $\phi(x), \psi (x)$ are well defined at each point $x \in [0,L]$, since their values at the common endpoints of these intervals match.
Let $\ustrut{12}\widehat{\phi}(\kappa ), \widehat{\psi}(\kappa )$, denote their bounded Fourier transforms, as in \eq{FtrL}.
With these definitions in hand, we are now able to state the main theorem characterizing revivals in the pseudoperiodic initial-boundary value problem for the linear Schr\"odinger equation.

\begin{thm} \label{thm:ShiftRep.uGenvGen}
	Given piecewise smooth initial data $u_0\colon [0,L]\to\CC$, the solution $u(t,x)$ to the initial-boundary value problem~\eqref{eqn:ls_pp} at rational time $t=L^2\:p/(4\pii\:q)$ is given by
\begin{equation}\label{eqn:F=q.withsc}
    u\left(\frac{L^2\: p}{4\pii\: q},x \right) = e^{\i {\kappa_0^2\:L^2\:p}/\pa{4\pii\: q}}F\left(x;-\,\frac{\kappa_0^2\:L^2\:p}{2\pii\: q}\right),
  \end{equation}
where
\BE \label{eqn:Analytic.F.looks.like.q.gen}
	F(x;s)= \frac{1}{L} \sum_{j\in\ZZ}e^{\i\kappa_js} \left[\frac{ \left((\beta_0+\beta_1)\gamma-2\right)e^{\i \kappa_j x} + (\beta_1-\beta_0)\gamma \,e^{-\i \kappa_j x}}{(\beta_0+\beta_1)(\gamma-\gamma^{-1})} \right]\left[ \widehat{\phi}(\kappa_j) +\delta \,\widehat{\psi}(-\kappa_j)\right],
\EE
with $\gamma, \delta, \phi$, and $\psi$ defined in \eq{gamma}, \eq{deltatau}, and \eqref{eqn:phipsi}, respectively.
\end{thm}
 
As we will see, the terms in \eq{eqn:F=q.withsc} involving $\widehat{\phi}(\kappa_j)$ will correspond to shifts of the initial data, while those involving $\widehat{\psi}(-\kappa_j)$ correspond to reflections thereof. 
Note if the linear Schr\"{o}dinger equation is supplied with $C^2$ initial datum which satisfies the pseudoperiodic boundary conditions, then it may be expected that the solution remains $C^2$ for all time.
  However, our characterisation in Theorem~\ref{thm:ShiftRep.uGenvGen} in terms of shifts of the initial datum may initially appear to contradict this.
  In fact, the precise shift operator employed, including using $\gamma$ to shift the phase at the interface, ensures that $\phi,\psi,F$ are all $C^2$.
Further interpretation of this result, both analytical and numerical, appears below.

\subsection{Solution representation --- generalised Fourier series} \label{sec:DeriveSeriesRep}

In this subsection, we solve our initial-boundary value problem~\eqref{eqn:ls_pp} by use of a generalised eigenfunction series~\cite{OlverPDEs}
\begin{equation}\label{eqn:series}
u(t,x)=\sum_{j\in\ZZ}c_je^{-\i\kappa_j^2t}X_j(x).
\end{equation}
Here, $-\i\kappa_j^2$ are the eigenvalues whose corresponding eigenfunctions $X_j(x)$ solve the pseudoperiodic boundary value problem
\Eq{Xjeq}
$$\req{X_j''(x)+\kappa_j^2X_j(x)=0,\\
\beta_0X_j(0)=X_j(L),\\
\beta_1X'_j(0)=X'_j(L).}$$
A straightforward calculation proves that \eq{Xjeq} admits a nonzero eigensolution of the form 
\Eq{Xj}
$$\eeq{X_j(x)=e^{\i\kappa_j x}+\frac{\gamma-\beta_0}{\beta_0-\gamma^{-1}}\,e^{-\i\kappa_jx}
=e^{\i\kappa_j x}+\frac{\gamma-\beta_1}{\gamma^{-1}-\beta_1}\,e^{-\i\kappa_jx},}$$
with $\kappa_j > 0$ if and only if $\kappa_j \in \ZDp$ belongs to the zero locus of the discriminant. 

Since the boundary value problem \eq{Xjeq} is not, in general, self-adjoint, the eigenfunctions $X_j$ on their own do not form an orthogonal system under the $\L^2$ inner product. 
However, when paired with the \emph{dual eigenfunctions} 
\Eq{Yk}
$$\eeq{Y_k(x)=e^{\i\overline{\kappa_k} x}+\frac{\gamma-1/\overline{\beta_1}}{1/\overline{\beta_1}-\gamma^{-1}}\,e^{-\i\overline{\kappa_k}x}=e^{\i\kappa_k x}+\frac{\gamma-1/\overline{\beta_0}}{\gamma^{-1}-1/\overline{\beta_0}}\,e^{-\i \kappa_kx}}$$
of the adjoint problem 
\Eq{Ykeq}
$${\req{Y_k''(x)+\overline{\kappa}_k^2Y_k(x)=0,\\
Y_k(0)=\overline{\beta_1}Y_k(L),\\
Y_k'(0)=\overline{\beta_0}Y_k'(L),}}$$
 they form a biorthogonal system.
 That is,
$$\langle X_j,Y_k\rangle=\left\{\begin{array}{lll}0,&&j\neq k,\\ 
\tau,&&j=k,\end{array}\right.$$
where $\tau $ is defined in \eq{deltatau}.

Evaluating the eigenfunction series for the solution~\eqref{eqn:series} at $t=0$ and taking the inner product of both sides with $Y_k$ we find the formulas for the Fourier coefficients:
\begin{align}
c_j={}&\frac{\langle u_0(x),Y_j(x)\rangle}{\langle X_j(x),Y_j(x)\rangle} = \frac{1}{\tau\:L}
\left[\>\widehat{u_0}(\kappa_j)+\frac{(\beta_1-\gamma)}{\gamma(1-\gamma\beta_1)}\,\widehat{u_0}(-\kappa_j)\>\right].
\end{align}
We conclude that the solution to the initial-boundary value problem is given by 
\Eq{eqn:seriessolnx}
$$u(t,x)=\frac{1}{\tau\:L}\sum_{j\in\ZZ}e^{-\i\kappa_j^2t}\left[\,e^{\i\kappa_jx}+\frac{(\gamma-\beta_0)e^{-\i\kappa_jx}}{\beta_0-\gamma^{-1}}\,\right]\left[\>\widehat{u_0}(\kappa_j)+\frac{(\beta_1-\gamma)}{\gamma(1-\gamma\beta_1)}\widehat{u_0}(-\kappa_j)\>\right].$$

\subsection{Solution representation --- Unified Transform Method} \label{sec:DeriveUTMRep}

A powerful alternative for deriving useful representations for solutions to a broad range of linear (and integrable nonlinear) initial-boundary value problems is the Unified Transform Method (often abbreviated UTM), which is also known as the Fokas Method~\cite{DeconinckTrogdonVasan, FokasBook, FokasPelloni4}.
In this subsection, we apply the UTM to derive alternative explicit solution formulae for the problem in hand.

Assuming existence of a solution, the first step in the UTM is to multiply the differential equation by an exponential solution to the adjoint equation and expressing the result as a divergence, leading to the so-called \emph{local relations}. 
In our case these take the form
\begin{equation}\label{ls_lr}
  \bbk{e^{-\i\kappa\:x+\omega\:t}u}_t-\bbk{e^{-\i\kappa\:x+\omega(\kappa)\:t}(\i u_x-\kappa \:u)}_x=0,
\end{equation}
with $\omega(\kappa )=\i\kappa ^2$.
Integrating the local relation~\eqref{ls_lr} over the strip $[0,L]\times [0,T]$ and applying Green's Theorem produces what are known as the \emph{global relations}:
\begin{equation}\label{ls_gr}
  0=\widehat{u_0}(\kappa )-e^{\omega T}\widehat{u}(\kappa ,T)+e^{-\i\kappa L}\bbk{\i h_1(T,\omega )-\kappa h_0(T,\omega )}-\bbk{\i g_1(T,\omega )-\kappa\: g_0(T,\omega )},
\end{equation}
where
$$\req{\widehat{u}(t,\kappa )=\int_0^L e^{-\i\kappa\:x}u(t,x)\ud x,\\
\widehat{u_0}(\kappa )=\int_0^L e^{-\i\kappa\:x}u_0(x)\ud x,}$$
are the Fourier transforms \eq{FtrL} of the solution and the initial data, respectively, while
\Eq{gh}
$$\xeq{g_j(t,\omega )=\int_0^t e^{\omega s}\,\partial_x^j u(s,x)\Big|_{x=0} \ud s,\\
h_j(t,\omega )=\int_0^t e^{\omega s}\,\partial_x^j u(s,x)\Big|_{x=L} \ud s,\\ t\in[0,T], \\ j=0,1,}$$
are time transforms of the solution evaluated at the boundary.
Formula~\eqref{ls_gr} is valid for all $\kappa\in\CC$; moreover, each term is entire in $\kap$.
Because the functions \eq{gh}
depend on $\kap$ only through the dispersion relation $\omega=\i\kappa^2$, we can create a second equation also valid for all $\kappa\in\CC$ using the {symmetry} transformation $\kappa\longmapsto-\kap$, namely,
\begin{equation}\label{ls_gr_neg}
  0=\widehat{u_0}(-\kappa )-e^{\omega T}\widehat{u}(-\kappa ,T)+e^{\i\kappa L}\bbk{\i h_1(T,\omega )+\kappa h_0(T,\omega )}-\bbk{\i g_1(T,\omega )+\kappa g_0(T,\omega )}.
\end{equation}
Rearranging, evaluating at $T=t$, and inverting the Fourier transform in~\eqref{ls_gr}, we find
\begin{equation}\label{ls_badsoln}
\eeq{  u(t,x)= \frac{1}{2\pii}\int_{-\infty}^\infty e^{\i\kappa\:x-\i\kappa ^2 t}\,\widehat{u_0}(\kappa )\ud \kappa +\frac{1}{2\pii}\int_{-\infty}^\infty e^{\i\kappa (x-L)-\i\kappa ^2 t}\,\bbk{\i h_1(t,\i\kappa ^2)-\kappa h_0(t,\i\kappa ^2)}\ud \kappa \cthn{150} -\frac{1}{2\pii}\int_{-\infty}^\infty e^{\i\kappa\:x-\i\kappa ^2 t}\,\bbk{\i g_1(t,\i\kappa ^2)-\kappa g_0(t,\i\kappa ^2)}\ud \kap.}
\end{equation}
Here, the sum of all three integrals should be understood in the sense of a single Cauchy principal value integral.
So far, equation \eqref{ls_badsoln} is not an effective solution formula because it involves all the boundary data, and hence is not well determined by the imposed boundary conditions. 

\begin{figure}[h!]
  \centering
  \def\svgwidth{.45\textwidth}
  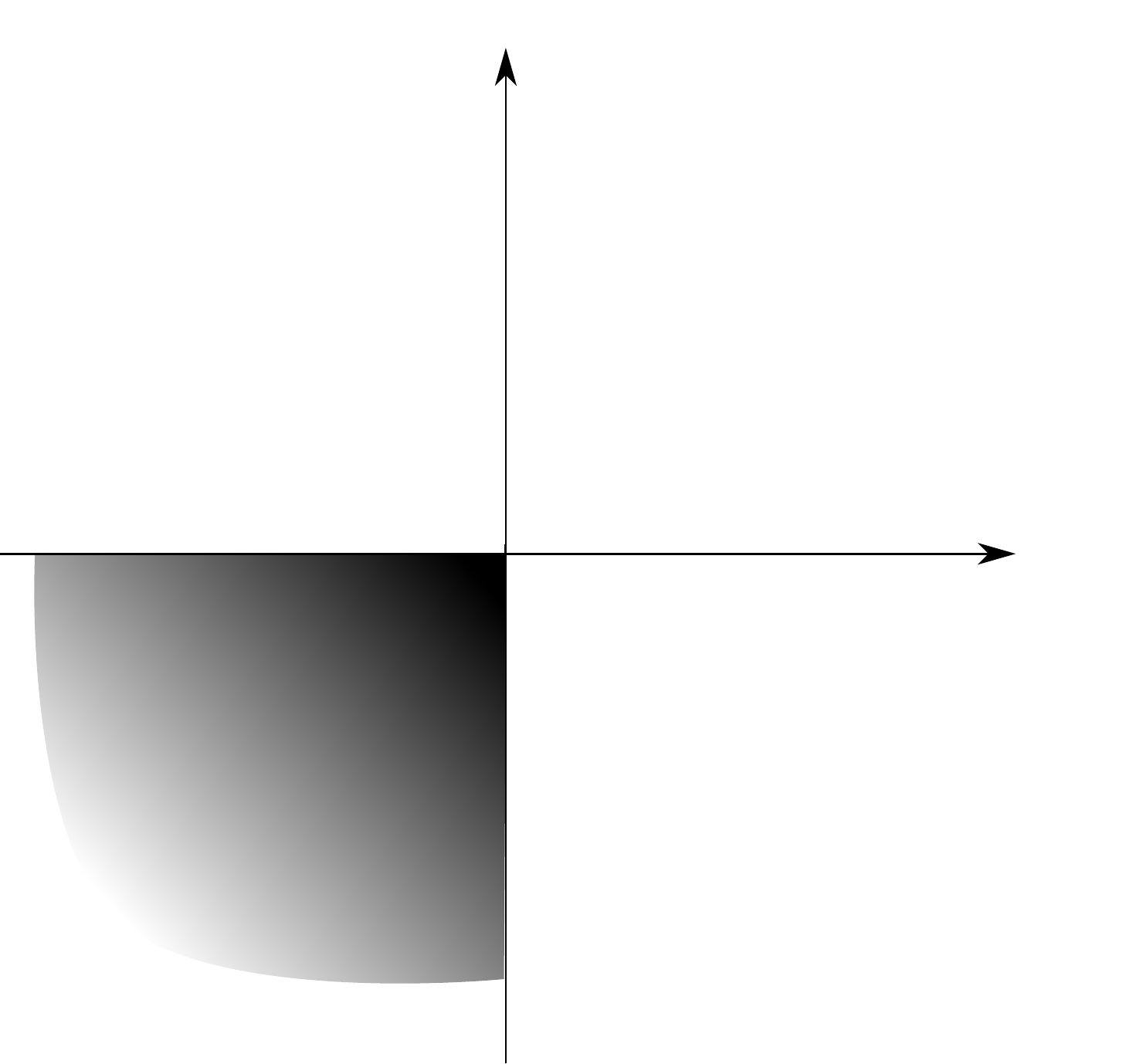
  \hfill
  \def\svgwidth{.45\textwidth}
  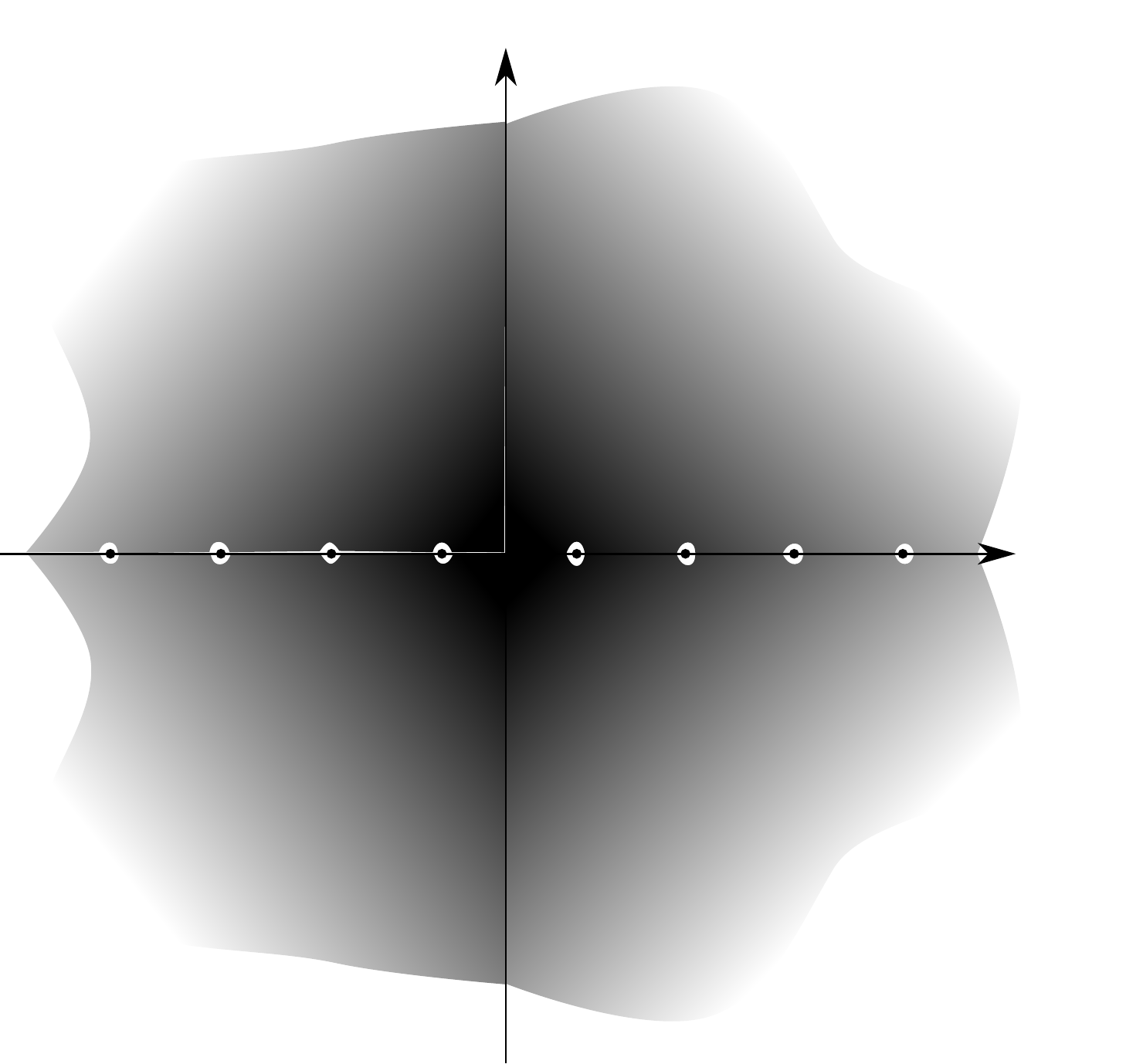
  \caption{The regions $D^+$ (solid red line) and $D^-$ (dashed green line) on the left and $\widetilde{D}^+$ (solid red line), $\widetilde{D}^-$ (dashed green line), $\widetilde{E}^+$ (solid blue line), and $\widetilde{E}^-$ (dashed purple line) on the right for the linear Schr\"odinger equation. 
  \label{fig:LS_D}}
\end{figure}

With some foresight, we will deform these integrals into the complex plane. Define the domains 
\Eq{Dpm}
$$\req{D^+=\{\kappa \in\CC:\Re(\i\kappa ^2)\leq0, \ \Im \kappa \geq 0\},\\D^-=\{\kappa \in\CC^\pm:\Re(\i\kappa ^2)\leq0, \ \Im \kappa \leq 0\},}$$
as plotted in Figure~\ref{fig:LS_D}.
Further, to avoid integrating through the zeros of the discriminant \eq{discriminant}, we introduce the subdomains
\Eq{Dtpm}
$$\widetilde{D}^\pm=D^\pm\,\setminus \bigcup_{\kappa \in \ZD} \text{a small neighbourhood of $\kappa $.}$$
We will integrate around the boundaries with the orientation of the integration contour is such that the interior of the domain lies on the left, as in Figure~\ref{fig:LS_D}, and thus find
\Eq{ls_badsoln_deps}
$$\eeq{  u(t,x)=\frac{1}{2\pii}\int_{-\infty}^\infty e^{\i\kappa\:x-\i\kappa ^2 t}\,\widehat{u_0}(\kappa )\ud \kappa -\frac{1}{2\pii}\int_{\partial D^-} e^{\i\kappa (x-L)-\i\kappa ^2 t}\,\bbk{\i h_1(t,\i\kappa ^2)-\kappa h_0(t,\i\kappa ^2)}\ud \kappa \cthn{150}-\frac{1}{2\pii}\int_{\partial D^+} e^{\i\kappa\:x-\i\kappa ^2 t}\,\bbk{\i g_1(t,\i\kappa ^2)-\kappa g_0(t,\i\kappa ^2)}\ud \kappa \\
=\frac{1}{2\pii}\int_{-\infty}^\infty e^{\i\kappa\:x-\i\kappa ^2 t}\,\widehat{u_0}(\kappa )\ud \kappa -\frac{1}{2\pii}\int_{\partial \widetilde{D}^-} e^{\i\kappa (x-L)-\i\kappa ^2 t}\,\bbk{\i h_1(t,\i\kappa ^2)-\kappa h_0(t,\i\kappa ^2)}\ud \kappa \cthn{150}-\frac{1}{2\pii}\int_{\partial \widetilde{D}^+} e^{\i\kappa\:x-\i\kappa ^2 t}\,\bbk{\i g_1(t,\i\kappa ^2)-\kappa g_0(t,\i\kappa ^2)}\ud \kappa .}$$



Using~\eqref{ls_gr},~\eqref{ls_gr_neg}, with $T$ replaced by $t$, and the pseudoperiodic boundary conditions~\eqref{eqn:ls_pp}, one can solve for $h_1,h_0,g_1$, and $g_0$ in terms of $\widehat{u_0}$ and $\widehat{u}(T,\cdot)$.
Substituting these into formula~\eqref{ls_badsoln_deps}, and using Jordan's Lemma and Cauchy's theorem,~\cite{AblowitzFokas}, as is typical in the UTM,~\cite{DeconinckTrogdonVasan, FokasBook, FokasPelloni4},
 we see that the terms involving $\widehat{u}(T,\cdot)$ evaluate to 0, and our solution reduces to 
\begin{equation}\label{ls_soln_deps2}
\begin{split}
u(t,x)={}&\frac{1}{2\pii}\left(\int_{-\infty}^\infty \!\!e^{\i\kappa\:x-\i\kappa ^2 t}\,\widehat{u_0}(\kappa )\ud \kappa -\!\!\int_{\partial \widetilde{D}^-} \!\!\frac{e^{\i\kappa (x-L)-\i\kappa ^2 t}\,\zeta^-(\kap,t)}{\Delta(\kappa)} \ud \kappa -\!\!\int_{\partial \widetilde{D}^+}\!\! \frac{e^{\i\kappa\:x-\i\kappa ^2 t}\,\zeta^+(t,\kappa)}{\Delta(\kappa)}\ud \kappa\right),
\end{split}
\end{equation}
where the notation
\begin{subequations}
\begin{align}
	\zeta^+(t,\kappa ) ={}& \left( \left(\beta_0+\beta_1\right)e^{\i \kappa L}-2 \right) \widehat{u_0}(\kappa) + \left(\beta_0-\beta_1\right) e^{-\i\kappa L} \,\widehat{u_0}(-\kappa),\\
	\zeta^-(t,\kappa ) ={}& \left(2\beta_0\beta_1e^{\i\kappa L}-\beta_0-\beta_1\right) \widehat{u_0}(\kappa) + \left(\beta_0-\beta_1\right) \widehat{u_0}(-\kappa),
\end{align}
\end{subequations}
follows~\cite{KesiciPelloniPryerSmith, SmithThesis, Smith2012}.
Note that 
$$\zeta^+(t,\kappa )-e^{-\i\kappa L}\,\zeta^-(t,\kappa )=\widehat{u_0}(\kappa )\Delta(\kappa ).$$ 

Further, set
$$E^+\cup E^-=\{\kappa \in\CC^\pm:\Re(\i\kappa ^2)\geq0\} \roq{and}
\widetilde{E}^\pm=E^\pm\setminus \bigcup_{\kappa \in\ZD} \text{a small neighborhood of $\kappa $,}$$ 
as in Figure~\ref{fig:LS_D}.
Let $S(\mu,r)$ be the positively oriented boundary of a disc centred at $\mu$ with small radius $r$. 
Then equation~\eqref{ls_soln_deps2} becomes
$$\eeq{	u(t,x)=\frac{1}{2\pii}\left(\int_{\partial \widetilde{E}^-}e^{\i\kappa\:x-\i\kappa ^2t}\, \frac{e^{-\i\kappa L}\,\zeta^-(t,\kappa )}{\Delta(\kappa )}\ud \kappa +\int_{\partial \widetilde{E}^+}e^{\i\kappa\:x-\i\kappa ^2t} \,\frac{\zeta^+(t,\kappa )}{\Delta(\kappa )}\ud \kappa\right)
\cthn{15}+\frac{1}{2\pii}\left(\sum_{\mu \in \ZDp}\int_{S(\mu,r)} e^{\i\kappa\:x-\i\kappa ^2t}\,\frac{\zeta^+(\kap,t)}{\Delta(\kappa)}\ud \kap+ \sum_{\mu \in \ZDm}\int_{S(\mu,r)} e^{\i\kappa\:x-\i\kappa ^2t}\,\frac{e^{-\i\kappa L}\,\zeta^-(\kap,t)}{\Delta(\kappa)}\ud\kappa \right)\\
	=\frac{1}{2\pii}\left(\int_{\partial \widetilde{E}^-}e^{\i\kappa\:x-\i\kappa ^2t} \,\frac{e^{-\i\kappa L}\,\zeta^-(t,\kappa )}{\Delta(\kappa )}\ud \kappa +\int_{\partial \widetilde{E}^+}e^{\i\kappa\:x-\i\kappa ^2t} \,\frac{\zeta^+(t,\kappa )}{\Delta(\kappa )}\ud \kappa\right)
	\cthn{15}+\frac{1}{2\pii}\left(\sum_{\mu \in \ZDp}\int_{S(\mu,r)} e^{\i\kappa\:x-\i\kappa ^2t}\,\frac{\zeta^+(\kap,t)}{\Delta(\kappa)}\ud\kap+ \sum_{\mu \in \ZDm}\int_{S(\mu,r)} e^{\i\kappa\:x-\i\kappa ^2t}\,\left(\frac{\zeta^+(\kap,t)}{\Delta(\kappa)}-\widehat{u_0}(\kappa)\right)\ud\kappa \right).}
	$$
The first two integrals are zero by an application of Jordan's Lemma and Cauchy's theorem.
Similarly,
$$\int_{S(\mu,r)}e^{\i\kappa\:x-\i\kappa ^2t}\, \widehat{u_0}(\kappa)\ud\kap=0,$$
and so
\Eq{eqn:q.series}
$$\eeq{u(t,x)=\frac{1}{2\pii} \sum_{\mu \in \ZD}\int_{S(\mu,r)} e^{\i\kappa\:x-\i\kappa ^2t}\>\frac{\zeta^+(\kap,t)}{\Delta(\kappa)}\ud\kappa\\
= \i \sum_{j\in\ZZ} \Res_{\kap=\kappa_j} \left[ e^{\i \kappa\:x-\i \kappa^2t} \>\frac{\zeta^+(\kap,t)}{\Delta(\kappa)} \right] +\i \sum_{j\in\ZZ} \Res_{\kap=-\kappa_j} \left[ e^{\i \kappa\:x-\i \kappa^2t} \>\frac{\zeta^+(\kap,t)}{\Delta(\kappa)} \right] \\ 
	=\i \sum_{j\in\ZZ} e^{-\i \kappa_j^2t} \> \frac{e^{\i \kappa_j x}\,\zeta^+(\kappa_j,t)-e^{-\i \kappa_j x}\,\zeta^+(-\kappa_j,t)}{\Delta'(\kappa_j)}, }$$
where the last equality holds provided all zeros of $\Delta$ are simple, as we assumed from the outset.
Note that, due to the sense in which equation~\eqref{ls_badsoln} is understood, the above series should be understood in the sense of principal values, that is, the limit $n\to\infty$ of partial sums over $-n\leq j \leq n$.

Finally, using the definition \eq{gamma} of $\gamma$, equation~\eq{eqn:q.series} becomes
\Eq{eqn:seriessoln}
$$\eeq{u(t,x)=\frac{1}{L}\sum_{j\in\ZZ} e^{-\i\kappa_j^2t}\left[ e^{\i\kappa_j x} \left(\frac{\left[-2+ \left(\beta_0+\beta_1\right)\gamma \right] \widehat{u_0}(\kappa_j)+ \left(\beta_0-\beta_1\right) \gamma^{-1} \widehat{u_0}(-\kappa_j)}{(\beta_0+\beta_1)(\gamma-\gamma^{-1})}\right) \right.
\cthn{100}\left. +e^{-\i\kappa_j x}\left(\frac{\left[2- \left(\beta_0+\beta_1\right)\gamma^{-1} \right] \widehat{u_0}(-\kappa_j) + \left(\beta_1-\beta_0\right) \gamma \widehat{u_0}(\kappa_j)}{(\beta_0+\beta_1)(\gamma-\gamma^{-1})}\right)\right]\\
= \frac{1}{L} \sum_{j\in\ZZ} e^{-\i \kappa_j^2t} \left[\frac{ \left((\beta_0+\beta_1)\gamma-2\right)e^{\i \kappa_j x} + (\beta_1-\beta_0)\gamma e^{-\i \kappa_j x}}{(\beta_0+\beta_1)(\gamma-\gamma^{-1})} \right]\left[\,\widehat{u_0}(\kappa_j)+ \delta\,\widehat{u_0}(-\kappa_j) \,\right],}
$$
where $\delta $ is given in \eq{deltatau} above. 
Again using~\eq{gamma}, equation~\eq{eqn:seriessoln} is equivalent to \eqref{eqn:seriessolnx}.

\section{Numerical experimentation}\label{sec:NumericalResults}

Applying the generalised Fourier series representation of $u(t,x)$ given by~\eqref{eqn:seriessolnx}, or, equivalently~\eqref{eqn:seriessoln}, it is straightforward to obtain reasonable numerical approximations to the solution profiles at specified times simply by summing over a sufficient number of terms. 
In the accompanying figures, we sum over 20,001 terms, given that reasonable changes in the total number of terms does not produce significantly different results.
In conformity with the results obtained  in~\cite{Olver} for the periodic problem, we expect that problems with energy-conserving pseudoperiodic boundary conditions will exhibit fractalization and revival.
Results confirming this are displayed in Figures~\ref{fig:PseudoConservativeIrrational} and~\ref{fig:PseudoConservativeRational}.
More surprising is that we have found cases where even linear Schr\"{o}dinger problems with nonconservative 
pseudoperiodic boundary conditions still exhibit revivals, as illustrated in Figures~\ref{fig:PseudoNonconservativeIrrational} and~\ref{fig:PseudoNonconservativeRational}.
\begin{figure}[h!]
  \centering
  \begin{minipage}[b]{.32\linewidth}
    \centering
    \includegraphics[width=\linewidth]{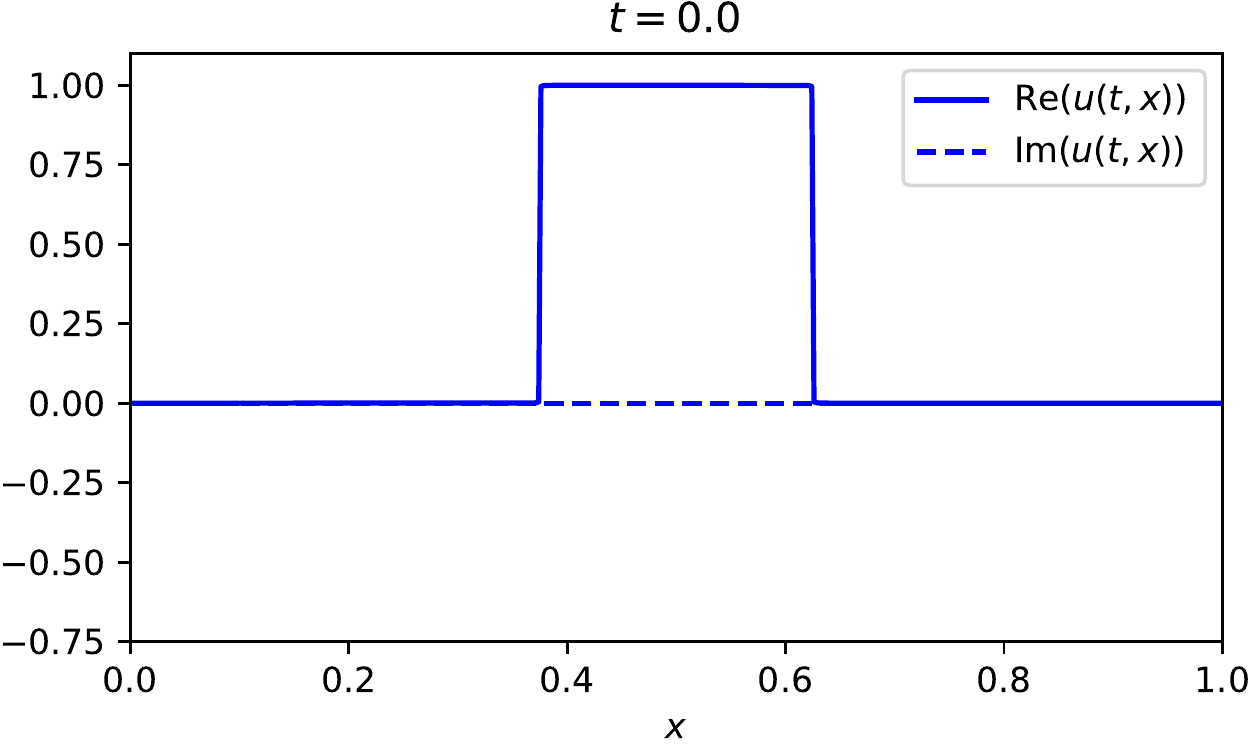}
    \subcaption{$t=0$}  \label{fig:PseudoConservativeIrrational.1}
  \end{minipage}
  \hfill
  \begin{minipage}[b]{.32\linewidth}
    \centering
    \includegraphics[width=\linewidth]{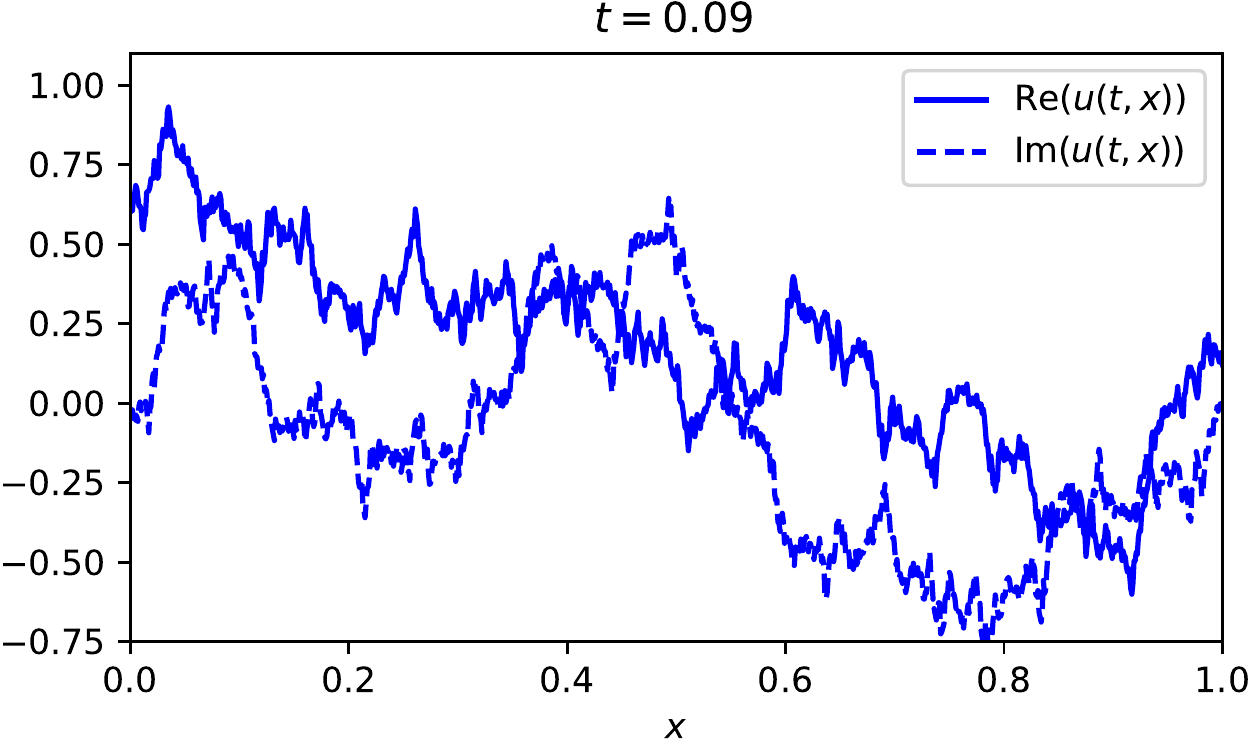}
    \subcaption{$t=0.09$}\label{fig:PseudoConservativeIrrational.2}
  \end{minipage}
  \hfill
  \begin{minipage}[b]{.32\linewidth}
    \centering
    \includegraphics[width=\linewidth]{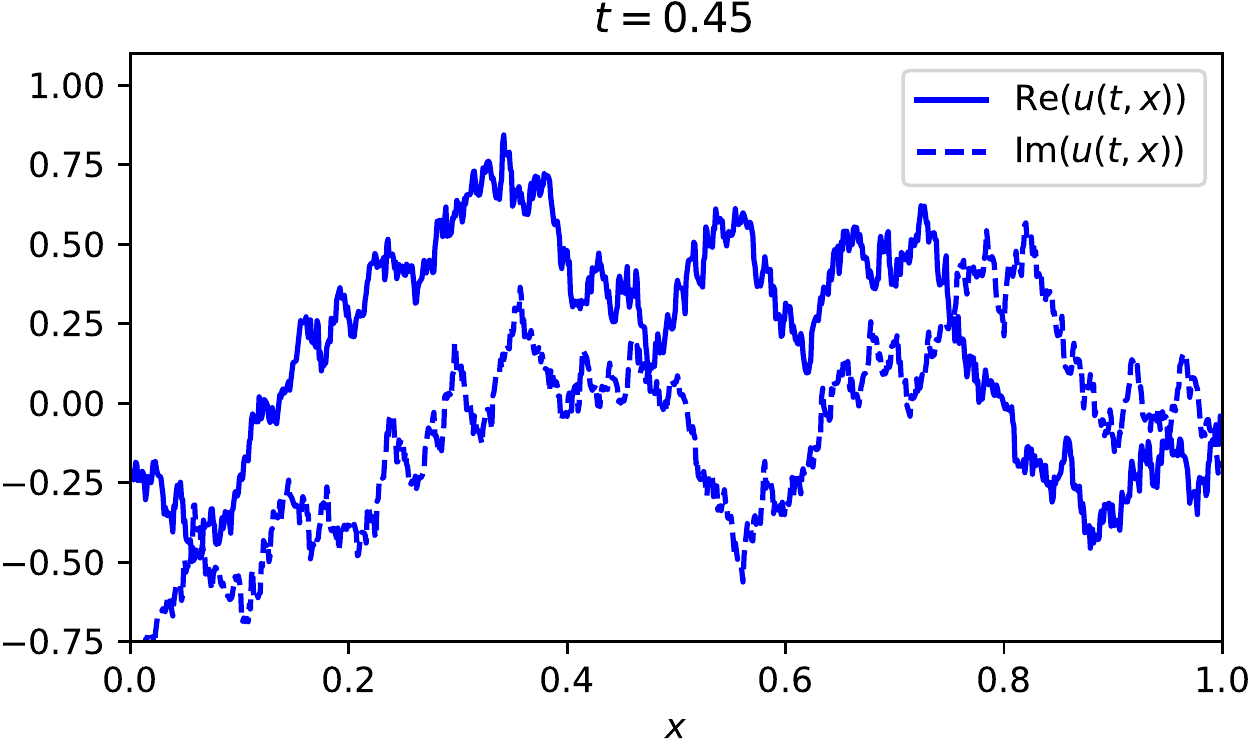}
    \subcaption{$t=0.45$}\label{fig:PseudoConservativeIrrational.6}
  \end{minipage}
  \\
  \vspace{2ex}
  \begin{minipage}[b]{.32\linewidth}
    \centering
    \includegraphics[width=\linewidth]{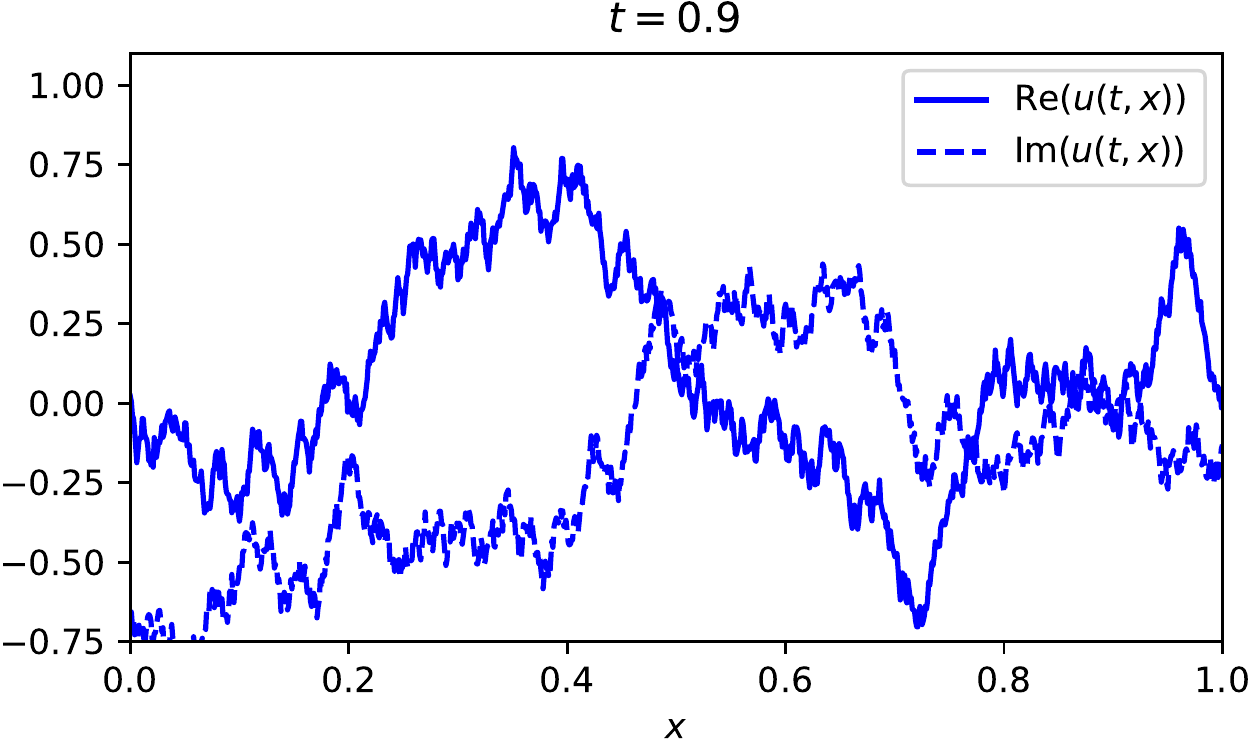}
    \subcaption{$t=0.9$} \label{fig:PseudoConservativeIrrational.11}
  \end{minipage}
  \hfill
  \begin{minipage}[b]{.32\linewidth}
    \centering
    \includegraphics[width=\linewidth]{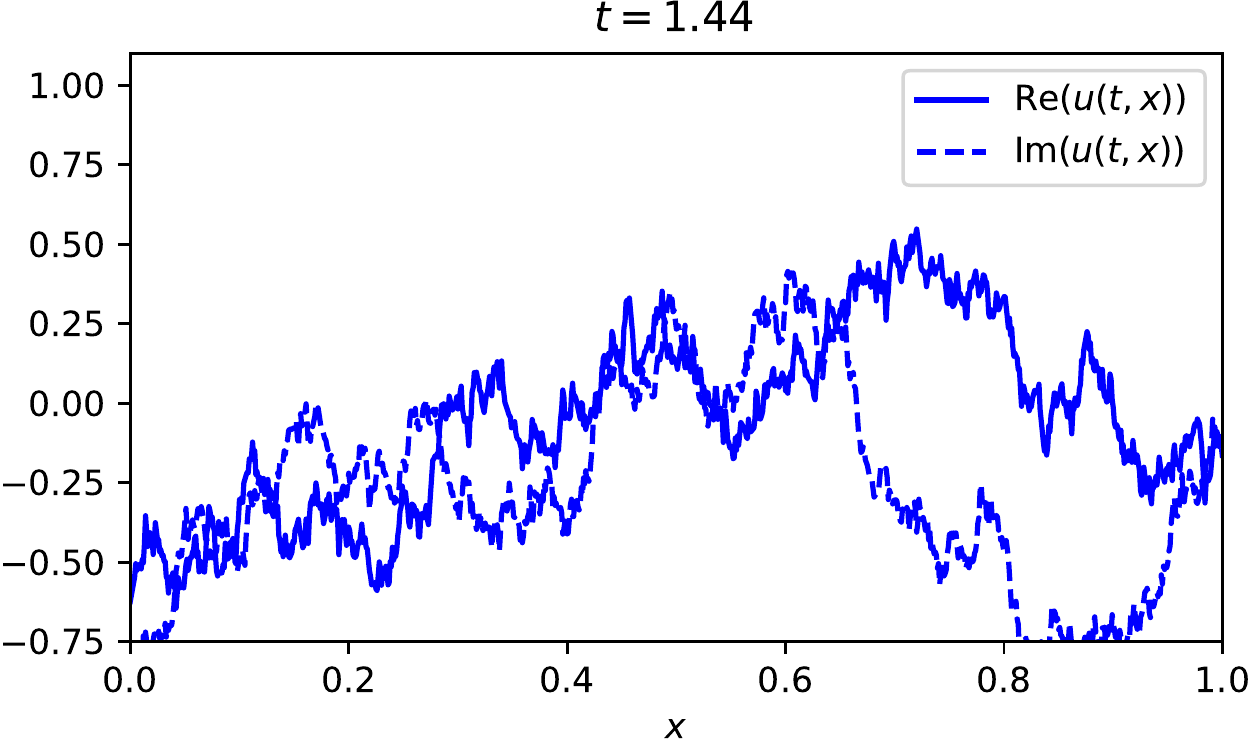}
    \subcaption{$t=1.44$}\label{fig:PseudoConservativeIrrational.17}
  \end{minipage}
  \hfill
  \begin{minipage}[b]{.32\linewidth}
    \centering
    \includegraphics[width=\linewidth]{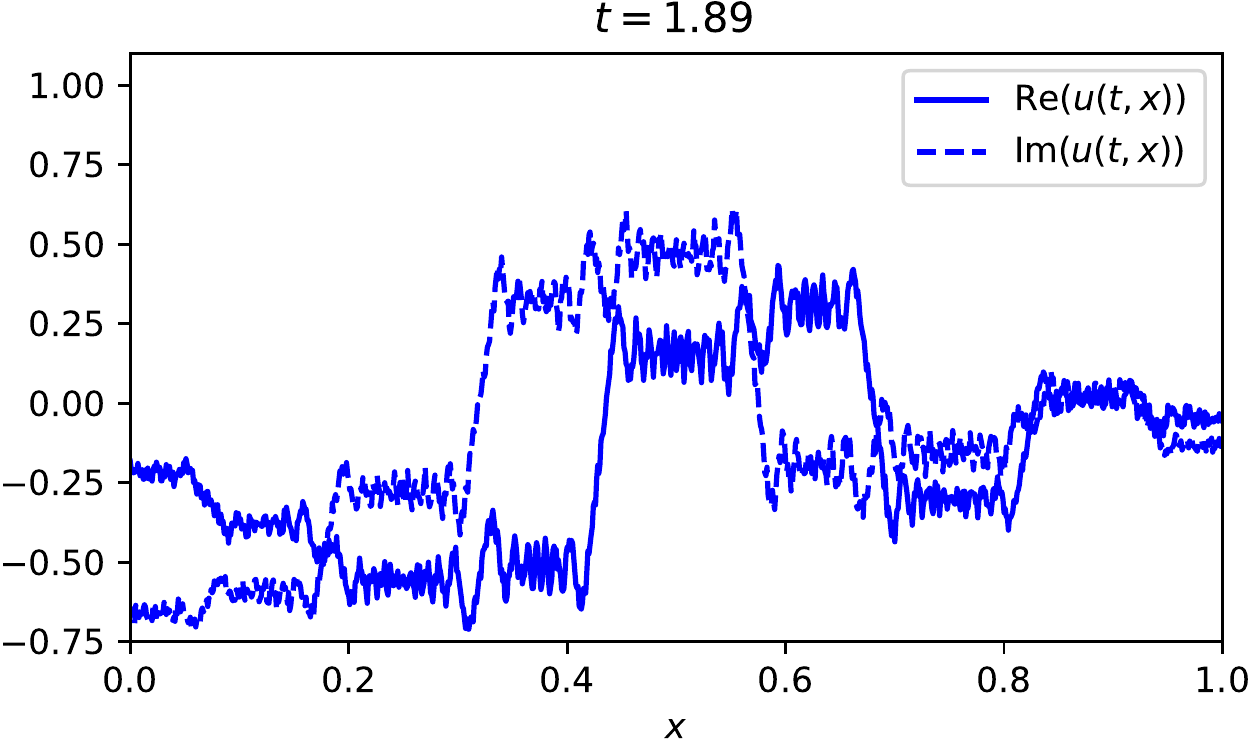}
    \subcaption{$t=1.89$}\label{fig:PseudoConservativeIrrational.22}
  \end{minipage}
  \caption{
    The solution of the linear Schr\"{o}dinger equation with energy conserving pseudoperiodic boundary conditions $\beta_0=1/5$, $\beta_1=5$ on $[0,1]$ and box initial datum, evaluated at certain ``irrational" times which are not commensurate with $L^2/(4\pii)$.
  }
  \label{fig:PseudoConservativeIrrational}
\end{figure}
\begin{figure}[h!]
  \centering
  \begin{minipage}[b]{.32\linewidth}
    \centering
    \includegraphics[width=\linewidth]{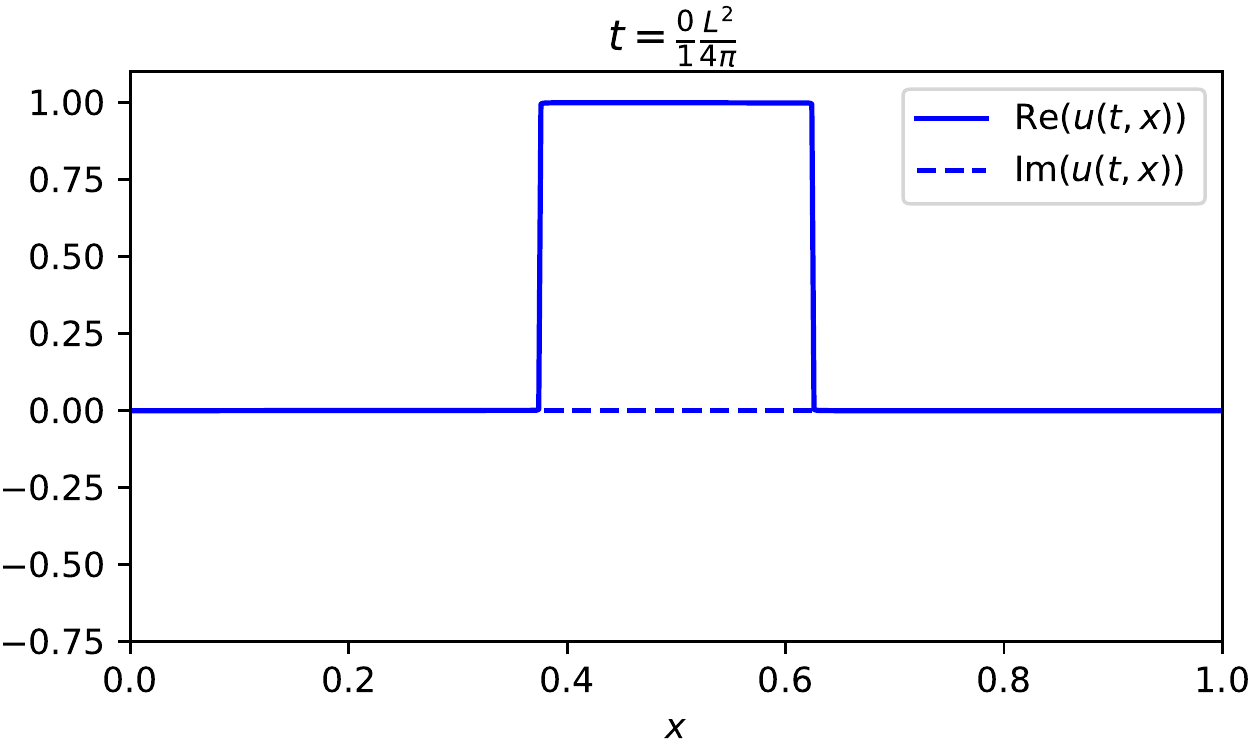}
    \subcaption{$t=0$}\label{fig:PseudoConservativeRational.1}
  \end{minipage}
  \hfill
  \begin{minipage}[b]{.32\linewidth}
    \centering
    \includegraphics[width=\linewidth]{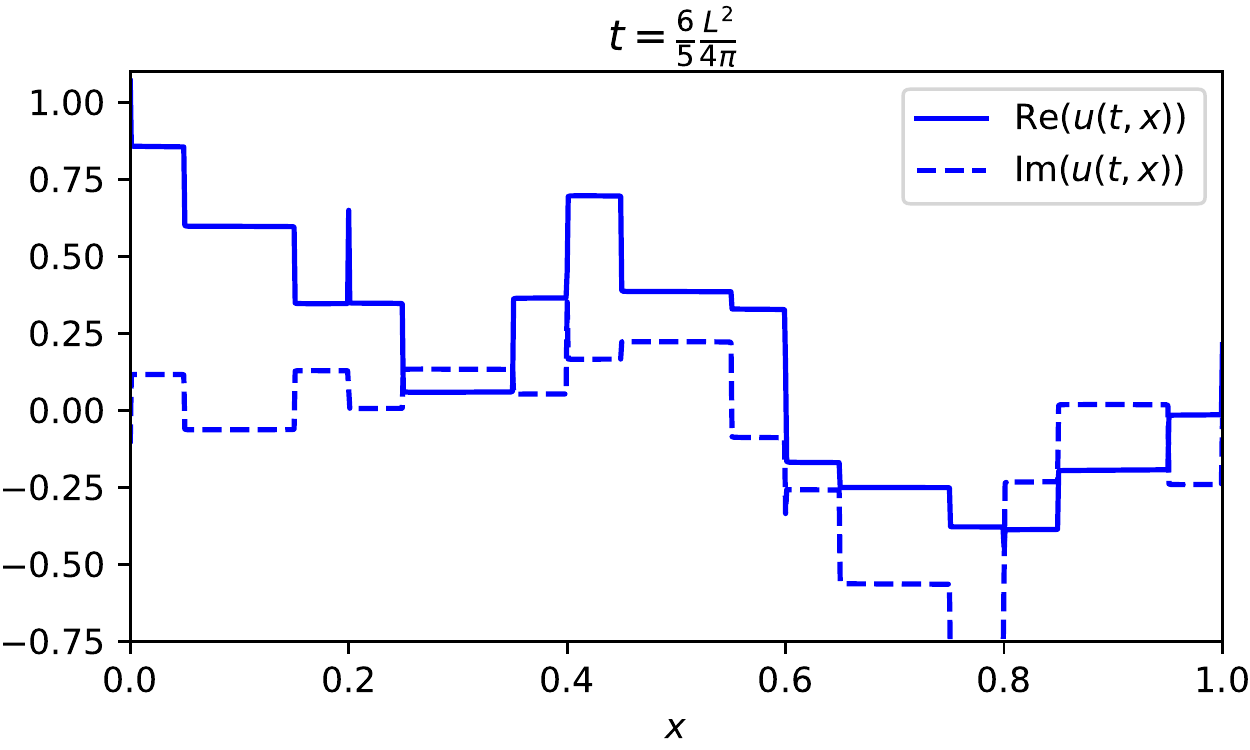}
    \subcaption{$t=\frac{6}{5}\frac{L^2}{4\pii}\approx0.09$}\label{fig:PseudoConservativeRational.2}
  \end{minipage}
  \hfill
  \begin{minipage}[b]{.32\linewidth}
    \centering
    \includegraphics[width=\linewidth]{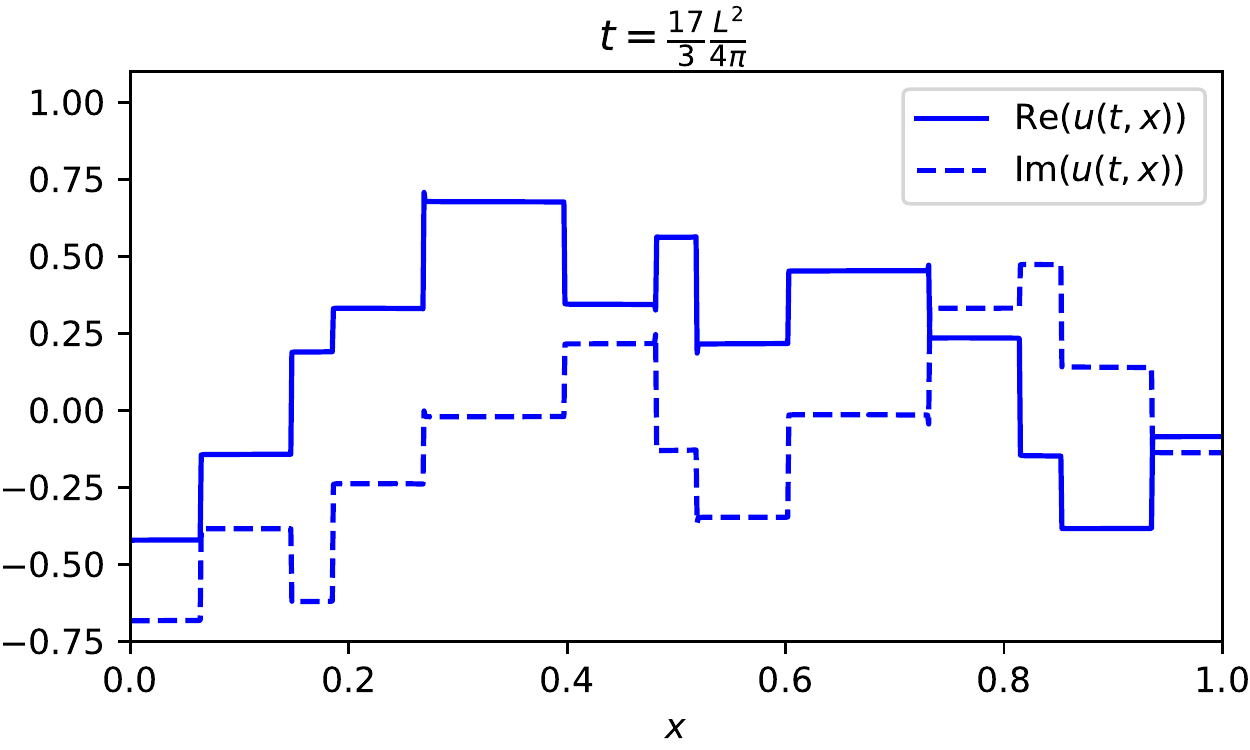}
    \subcaption{$t=\frac{17}{3}\frac{L^2}{4\pii}\approx0.45$}\label{fig:PseudoConservativeRational.6}
  \end{minipage}
  \\
  \vspace{2ex}
  \begin{minipage}[b]{.32\linewidth}
    \centering
    \includegraphics[width=\linewidth]{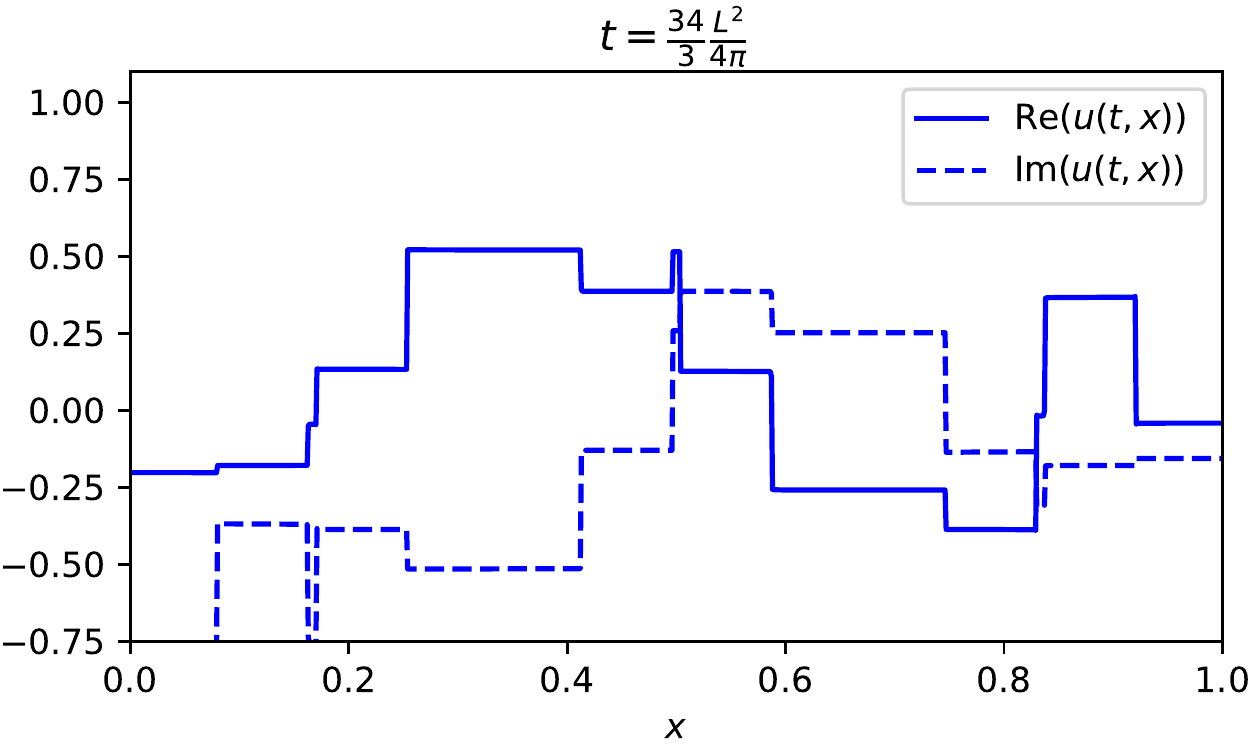}
    \subcaption{$t=\frac{34}{3}\frac{L^2}{4\pii}\approx0.9$}\label{fig:PseudoConservativeRational.11}
  \end{minipage}
  \hfill
  \begin{minipage}[b]{.32\linewidth}
    \centering
    \includegraphics[width=\linewidth]{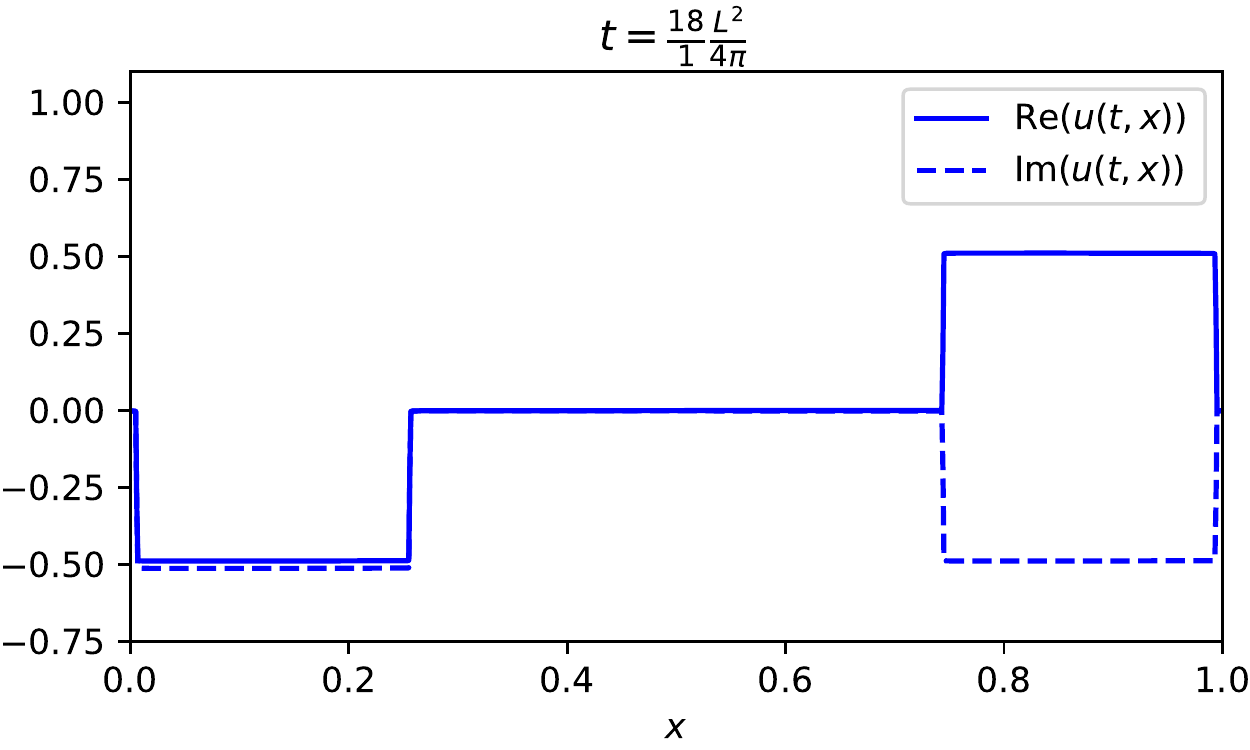}
    \subcaption{$t=\frac{18}{2}\frac{L^2}{4\pii}\approx1.44$}\label{fig:PseudoConservativeRational.17}
  \end{minipage}
  \hfill
  \begin{minipage}[b]{.32\linewidth}
    \centering
    \includegraphics[width=\linewidth]{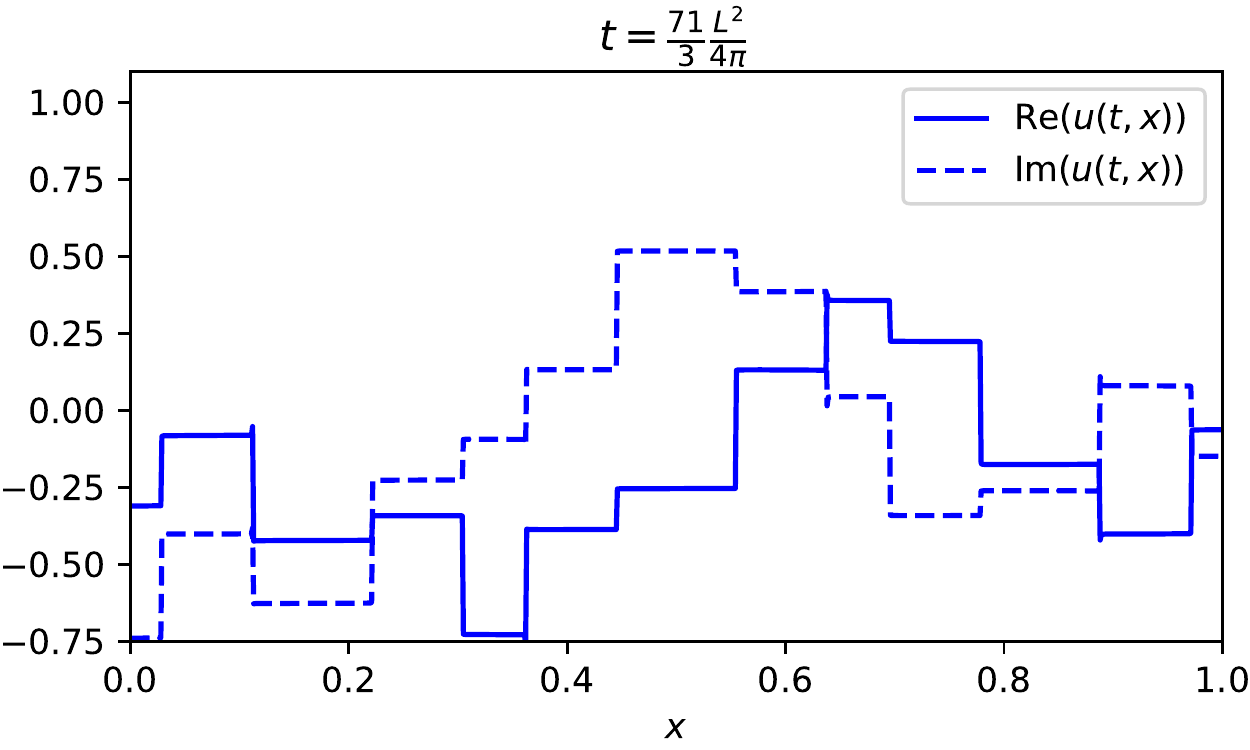}
    \subcaption{$t=\frac{71}{3}\frac{L^2}{4\pii}\approx1.89$}\label{fig:PseudoConservativeRational.22}
  \end{minipage}
  \caption{
    The solution of the linear Schr\"{o}dinger equation with energy conserving pseudoperiodic boundary conditions $\beta_0=1/5$, $\beta_1=5$ on $[0,1]$ and box initial datum, evaluated at certain ``rational" times which are commensurate with $L^2/(4\pii)$.
  }
  \label{fig:PseudoConservativeRational}
\end{figure}
\begin{figure}[h!]
  \centering
  \begin{minipage}[b]{.32\linewidth}
    \centering
    \includegraphics[width=\linewidth]{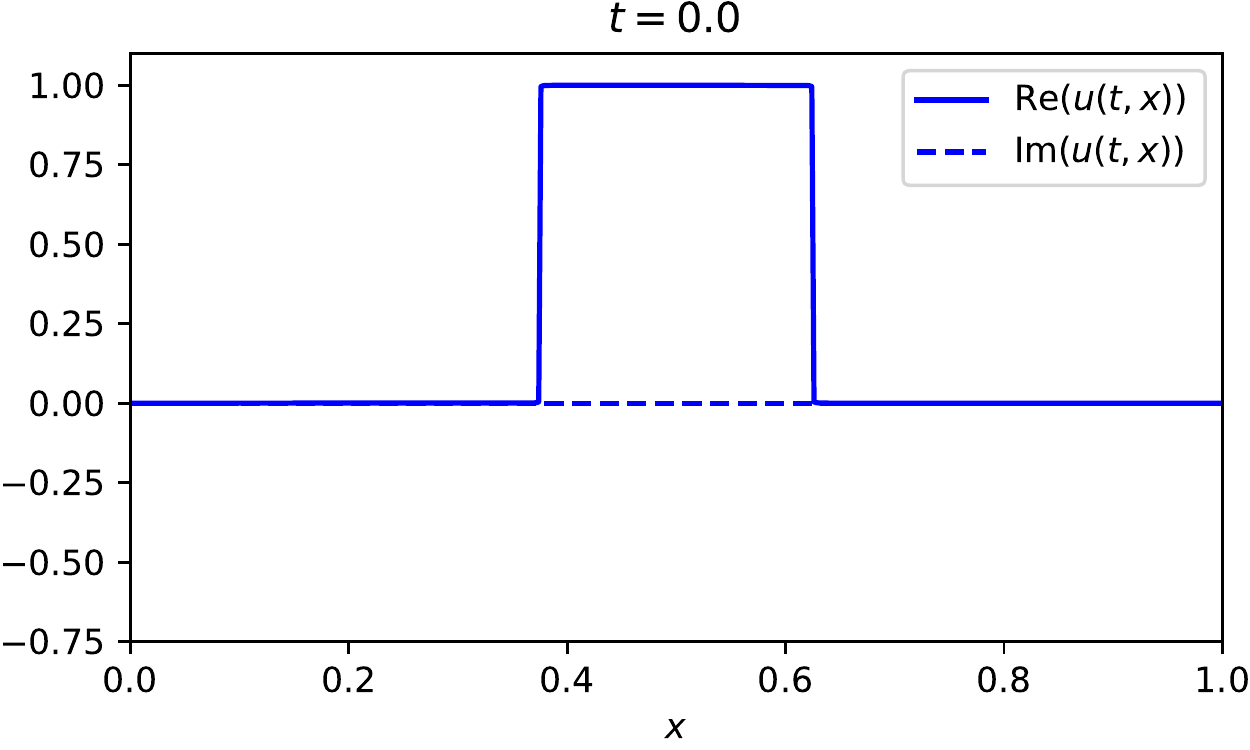}
    \subcaption{$t=0$}\label{fig:PseudoNonconservativeIrrational.1}
  \end{minipage}
  \hfill
  \begin{minipage}[b]{.32\linewidth}
    \centering
    \includegraphics[width=\linewidth]{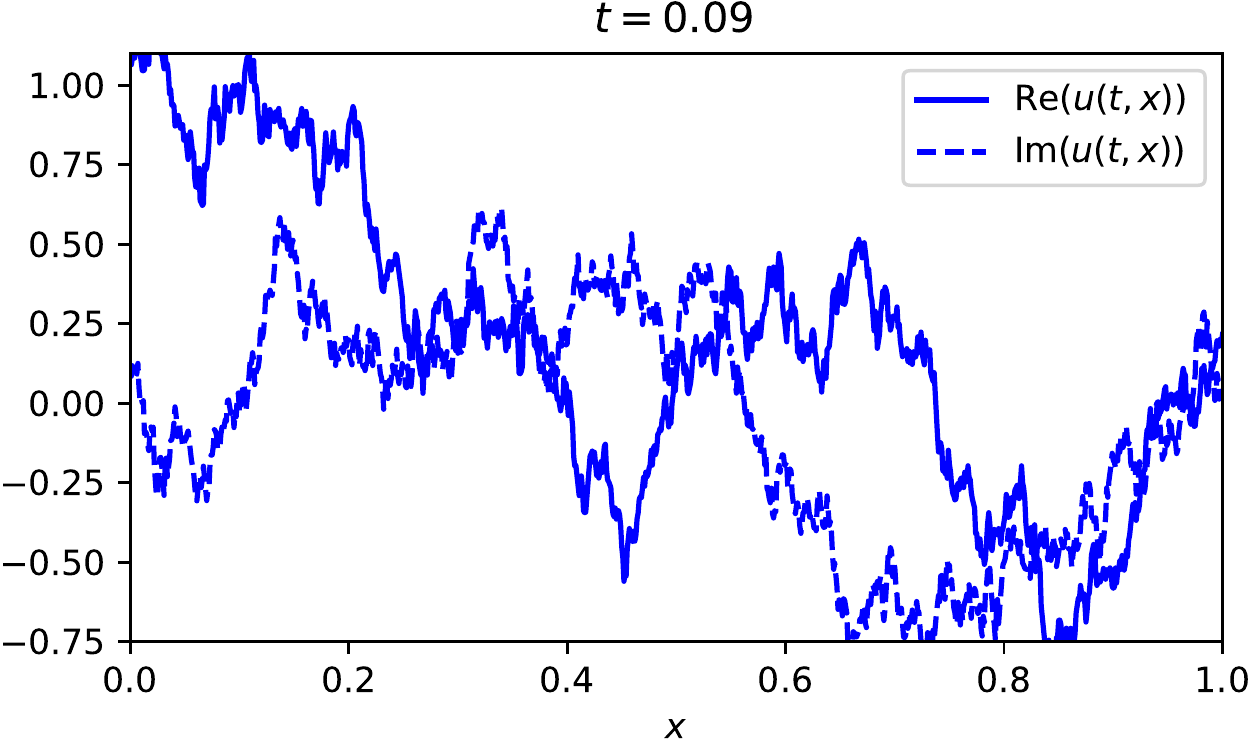}
    \subcaption{$t=0.09$}\label{fig:PseudoNonconservativeIrrational.2}
  \end{minipage}
  \hfill
  \begin{minipage}[b]{.32\linewidth}
    \centering
    \includegraphics[width=\linewidth]{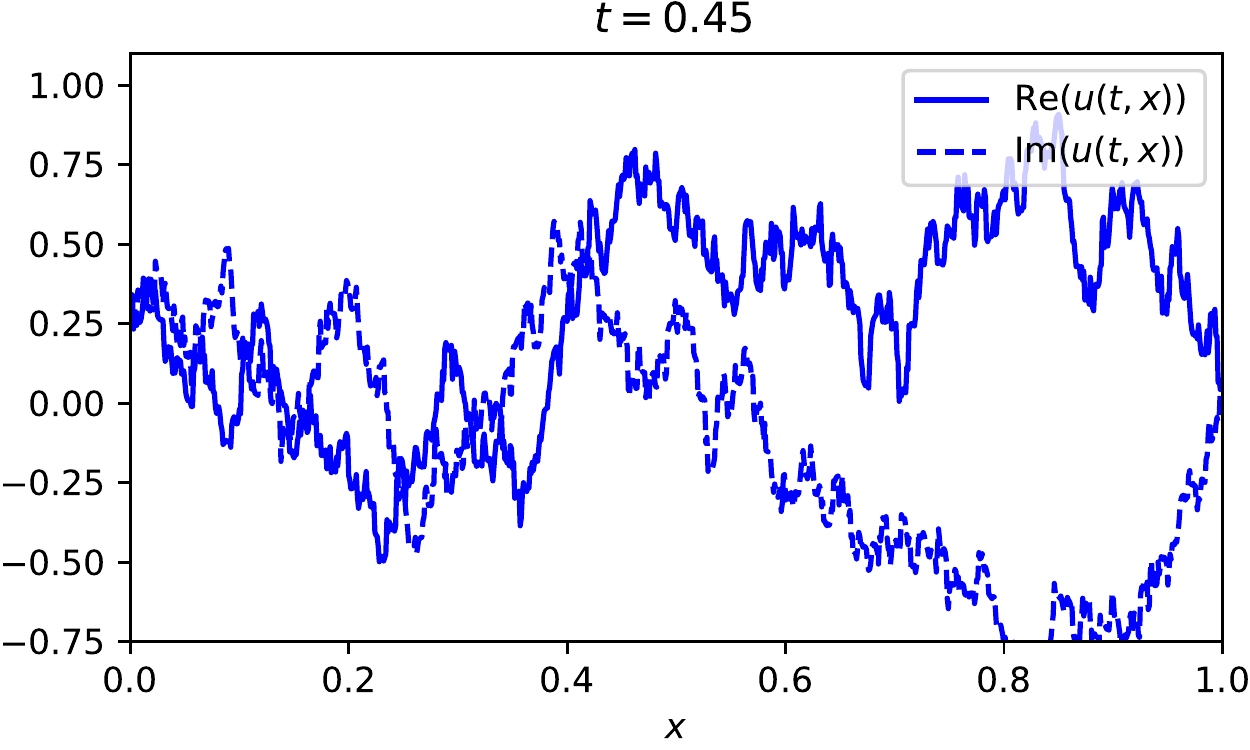}
    \subcaption{$t=0.45$}\label{fig:PseudoNonconservativeIrrational.6}
  \end{minipage}
  \\
  \vspace{2ex}
  \begin{minipage}[b]{.32\linewidth}
    \centering
    \includegraphics[width=\linewidth]{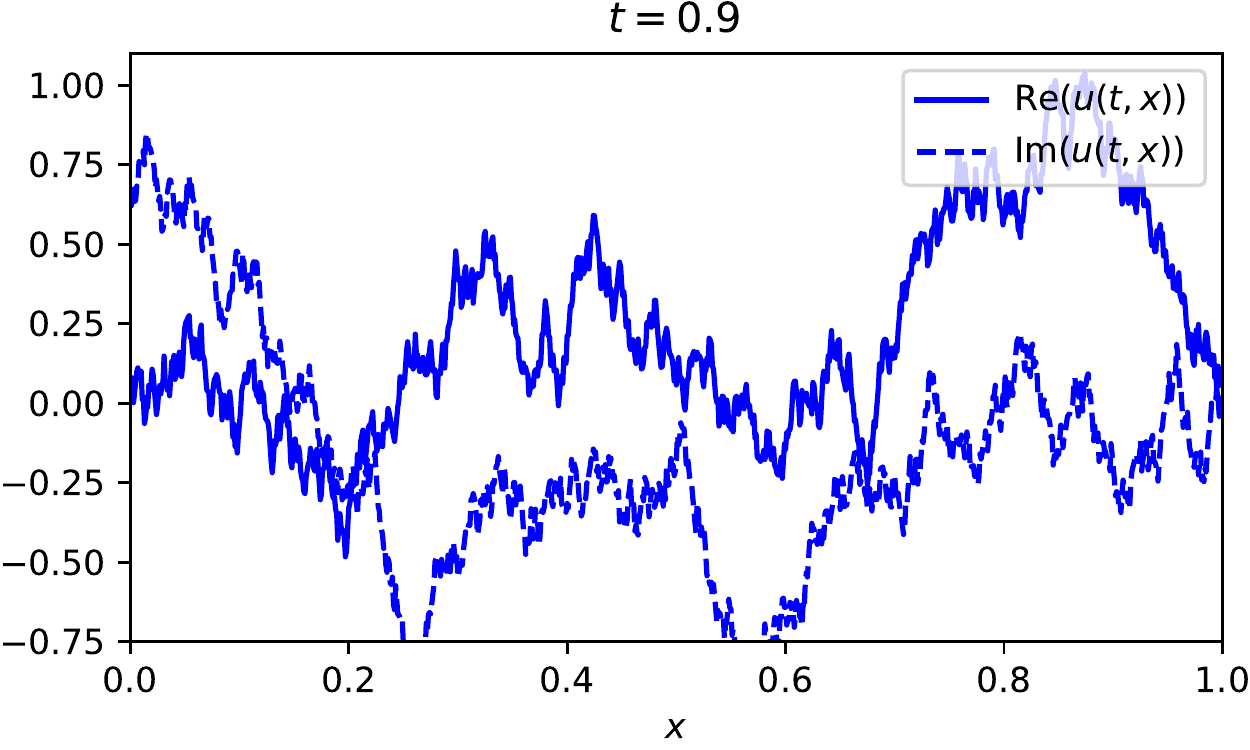}
    \subcaption{$t=0.9$}\label{fig:PseudoNonconservativeIrrational.11}
  \end{minipage}
  \hfill
  \begin{minipage}[b]{.32\linewidth}
    \centering
    \includegraphics[width=\linewidth]{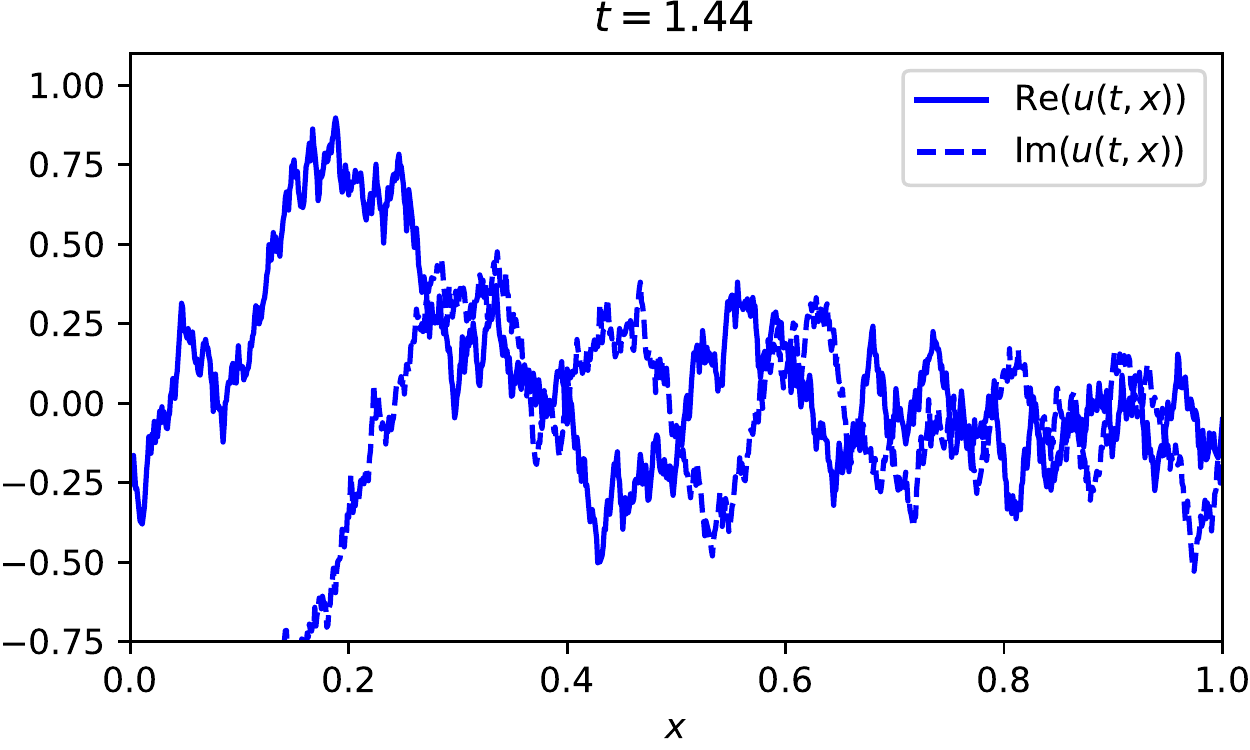}
    \subcaption{$t=1.44$}\label{fig:PseudoNonconservativeIrrational.17}
  \end{minipage}
  \hfill
  \begin{minipage}[b]{.32\linewidth}
    \centering
    \includegraphics[width=\linewidth]{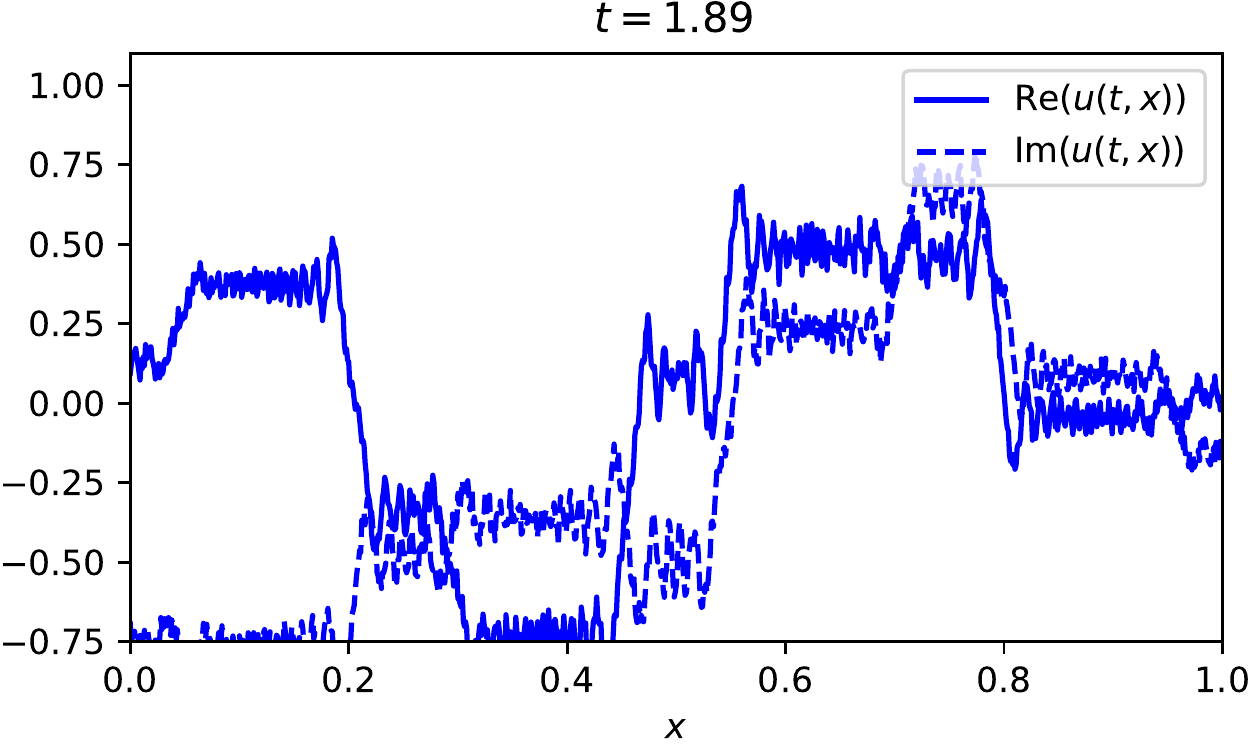}
    \subcaption{$t=1.89$}\label{fig:PseudoNonconservativeIrrational.22}
  \end{minipage}
  \caption{
    The solution of the linear Schr\"{o}dinger equation with energy non conserving pseudoperiodic boundary conditions $\beta_0=1/5$, $\beta_1=2$ on $[0,1]$ and box initial datum, evaluated at certain ``irrational" times which are not commensurate with $L^2/(4\pii)$.
  }
  \label{fig:PseudoNonconservativeIrrational}
\end{figure}
\begin{figure}[h!]
  \centering
  \begin{minipage}[b]{.32\linewidth}
    \centering
    \includegraphics[width=\linewidth]{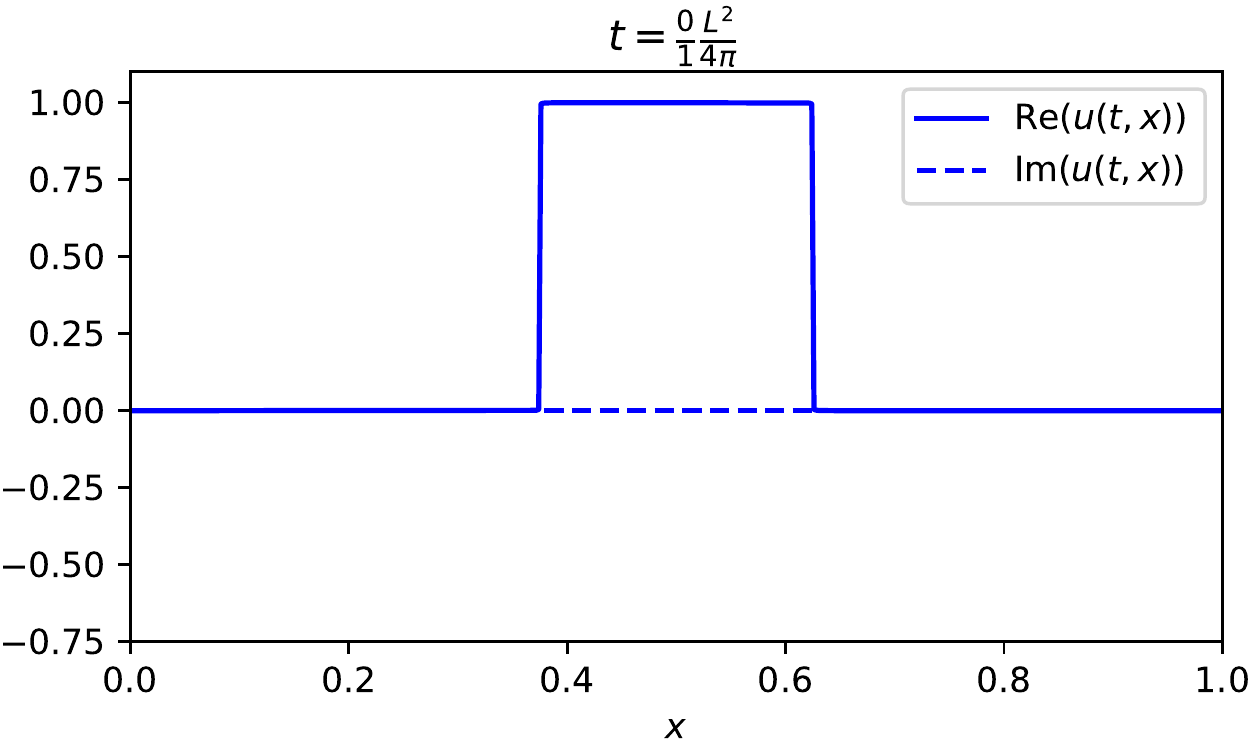}
    \subcaption{$t=0$}\label{fig:PseudoNonconservativeRational.1}
  \end{minipage}
  \hfill
  \begin{minipage}[b]{.32\linewidth}
    \centering
    \includegraphics[width=\linewidth]{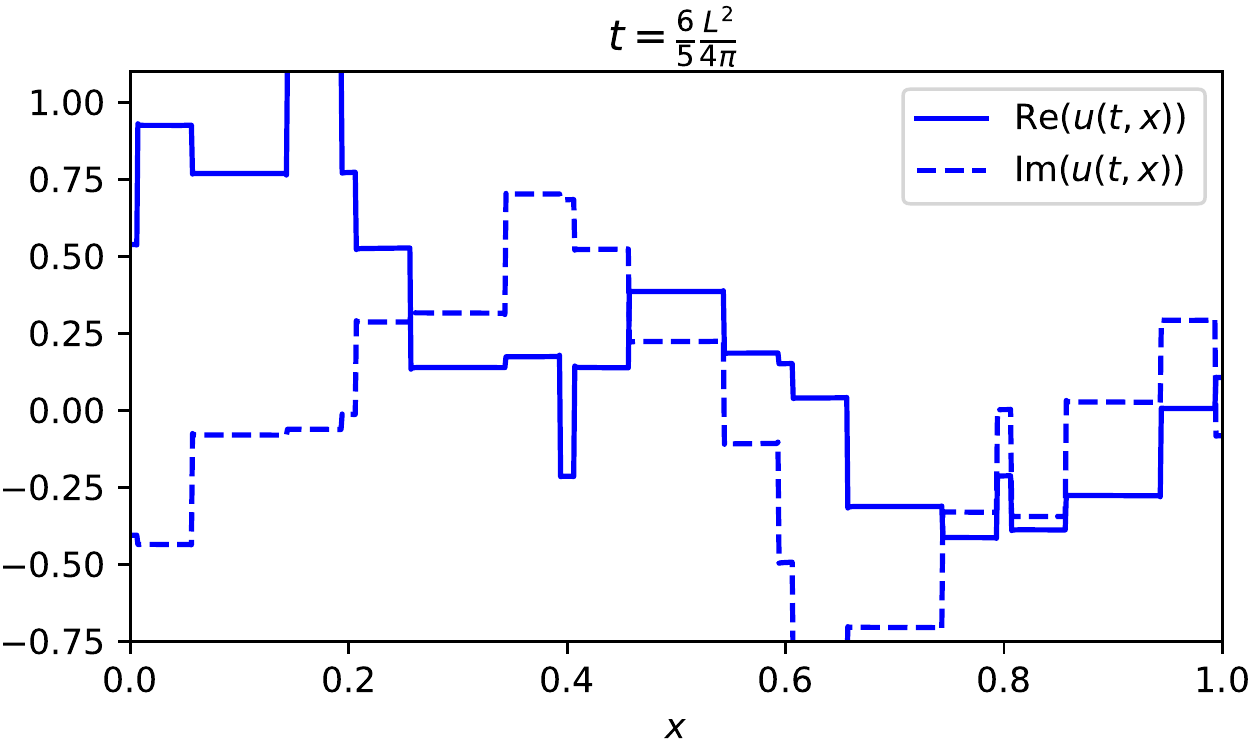}
    \subcaption{$t=\frac{6}{5}\frac{L^2}{4\pii}\approx0.09$}\label{fig:PseudoNonconservativeRational.2}
  \end{minipage}
  \hfill
  \begin{minipage}[b]{.32\linewidth}
    \centering
    \includegraphics[width=\linewidth]{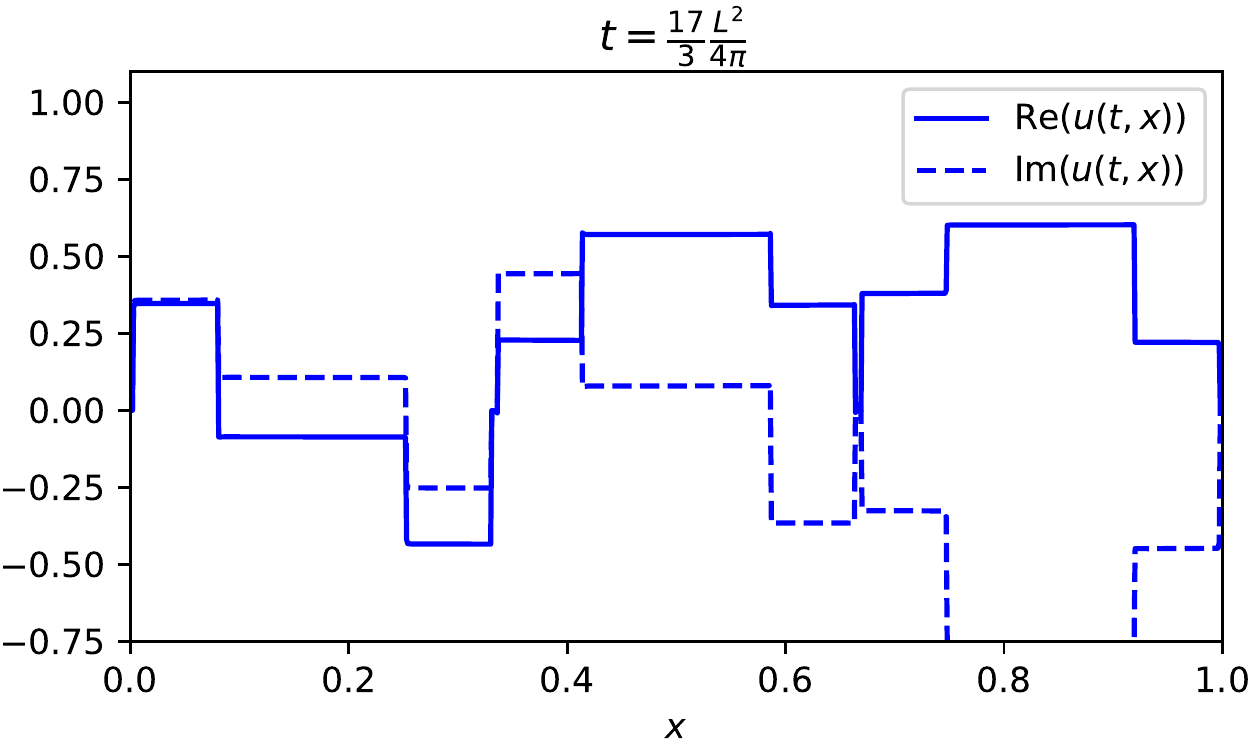}
    \subcaption{$t=\frac{17}{3}\frac{L^2}{4\pii}\approx0.45$}\label{fig:PseudoNonconservativeRational.6}
  \end{minipage}
  \\
  \vspace{2ex}
  \begin{minipage}[b]{.32\linewidth}
    \centering
    \includegraphics[width=\linewidth]{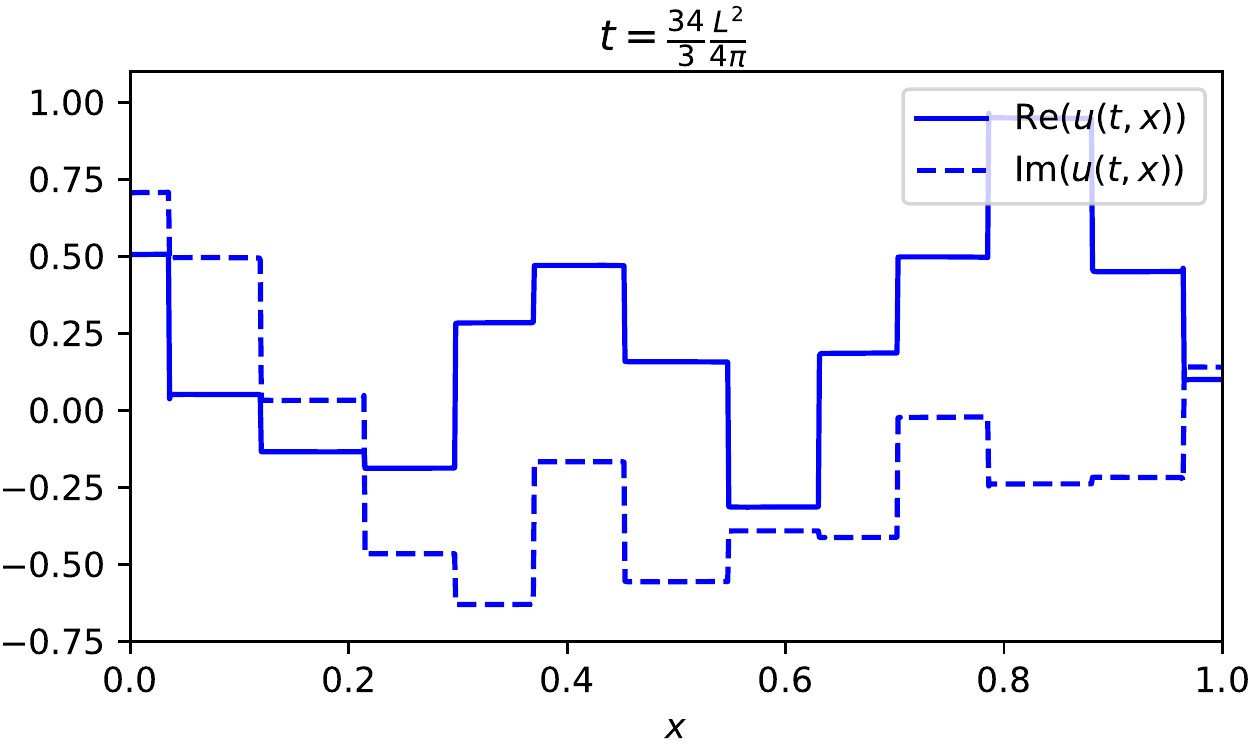}
    \subcaption{$t=\frac{34}{3}\frac{L^2}{4\pii}\approx0.9$}\label{fig:PseudoNonconservativeRational.11}
  \end{minipage}
  \hfill
  \begin{minipage}[b]{.32\linewidth}
    \centering
    \includegraphics[width=\linewidth]{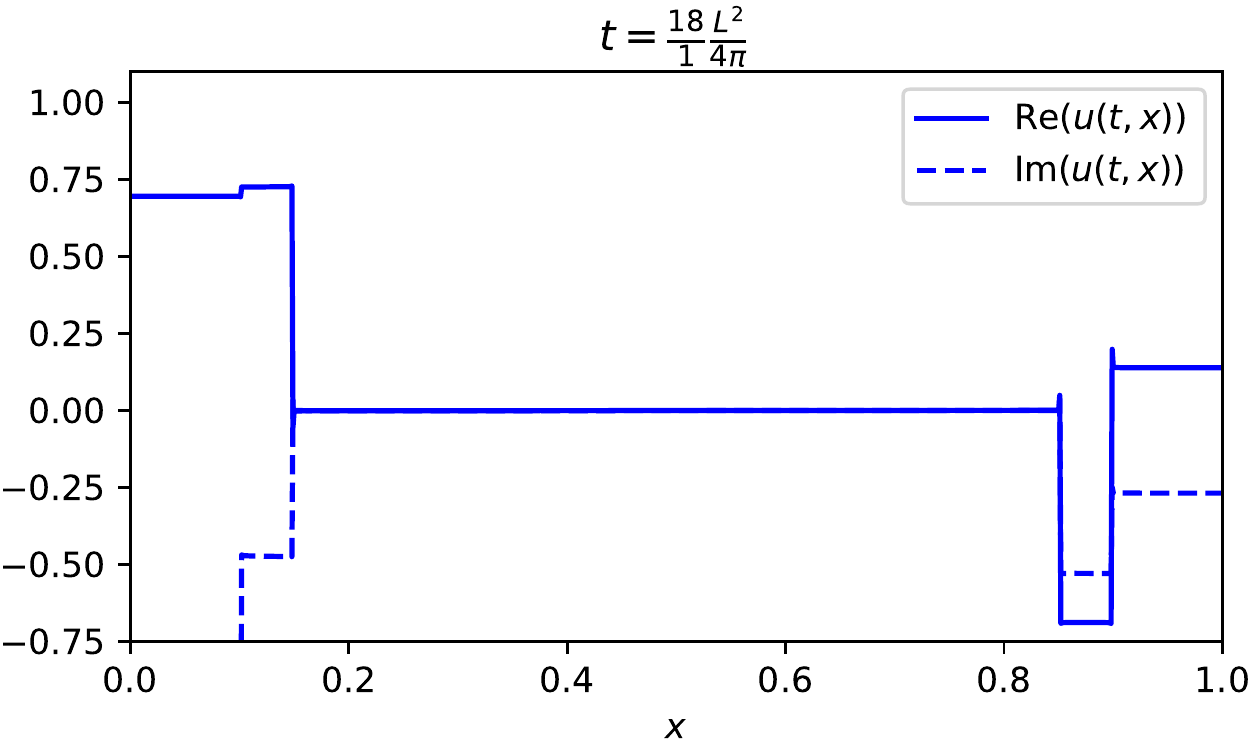}
    \subcaption{$t=\frac{18}{1}\frac{L^2}{4\pii}\approx1.44$}\label{fig:PseudoNonconservativeRational.17}
  \end{minipage}
  \hfill
  \begin{minipage}[b]{.32\linewidth}
    \centering
    \includegraphics[width=\linewidth]{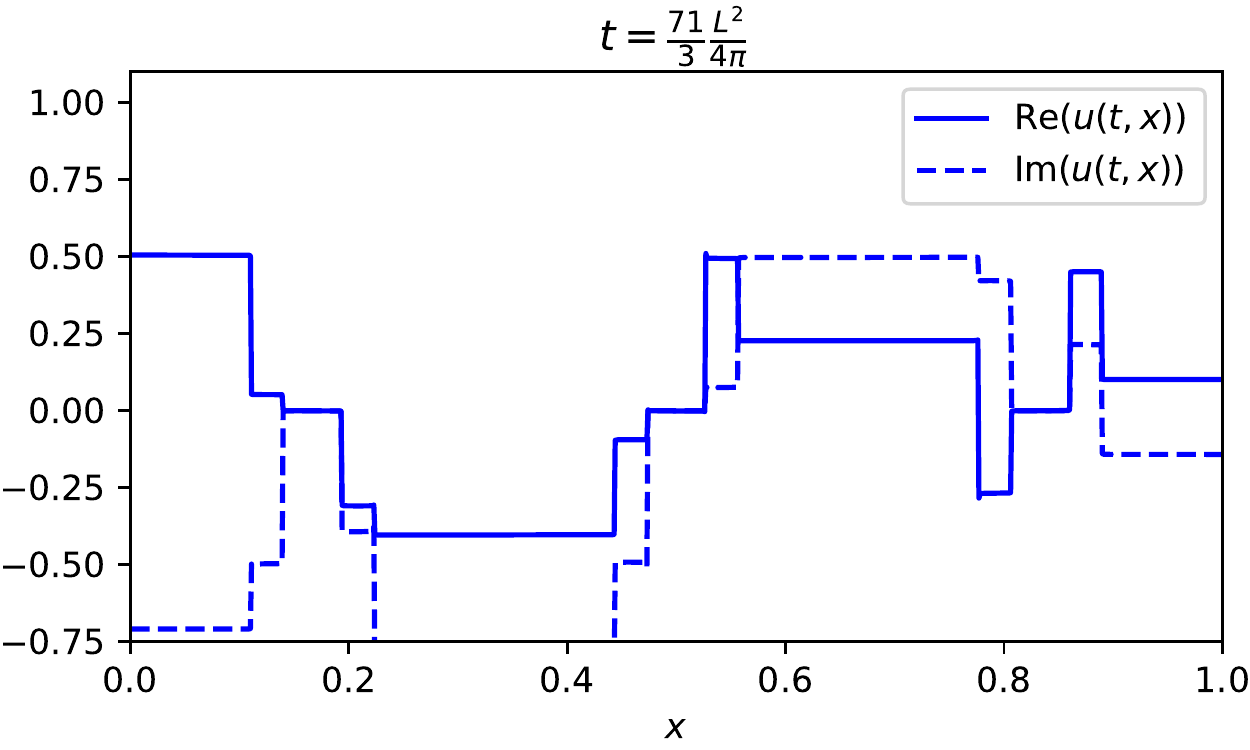}
    \subcaption{$t=\frac{71}{3}\frac{L^2}{4\pii}\approx1.89$}\label{fig:PseudoNonconservativeRational.22}
  \end{minipage}
  \caption{
    The solution of the linear Schr\"{o}dinger equation with energy non conserving pseudoperiodic boundary conditions $\beta_0=1/5$, $\beta_1=2$ on $[0,1]$ and box initial datum, evaluated at certain ``rational" times which are commensurate with $L^2/(4\pii)$.
  }
  \label{fig:PseudoNonconservativeRational}
\end{figure}

In order to better understand the effect, and guide the development of an analytic theory, it is instructive to perform further numerical experiments.
Referring to Figure~\ref{fig:PiecewiseConstantUnnecessary}, it appears that the solution at rational times --- commensurate with $L^2/(4\pii)$ --- is piecewise linear  for a piecewise linear initial datum.
The manner in which the slope of each piece of the solution appears to depend upon the slope of the initial datum also suggests that the solution at rational times may depend upon the initial datum in a simple way.
Note that, although Figure~\ref{fig:PiecewiseConstantUnnecessary} demonstrates that at appropriate times $u(t,x)$ is piecewise linear, its modulus $\lvert u(t,x) \rvert$ is not.
Hence our preference throughout Section~\ref{sec:NumericalResults} to display plots of real and imaginary parts of $u$, despite the greater relevance of $\lvert u(t,x) \rvert$ in applications.

\begin{figure}[h!]
  \centering
  \begin{minipage}[b]{.49\linewidth}
    \centering
    \includegraphics[width=\linewidth]{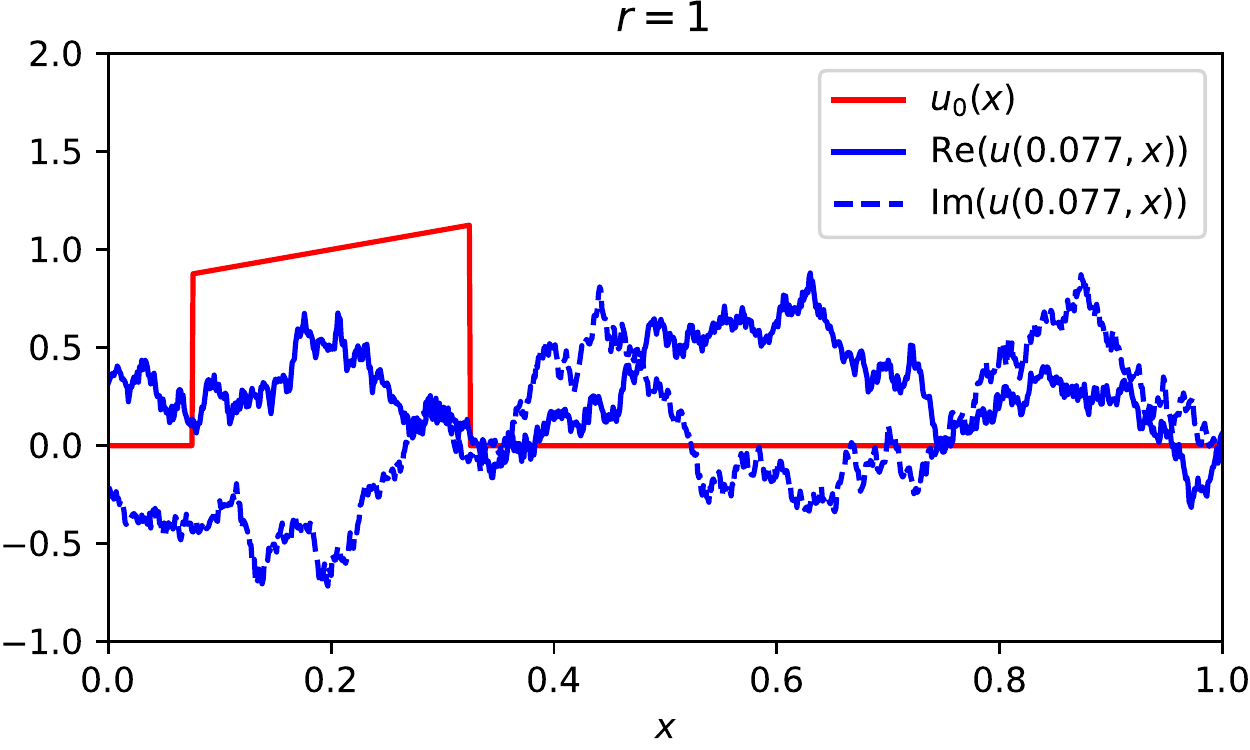}
    \subcaption{$t=0.077$, $r=1$}\label{fig:PiecewiseConstantUnnecessary.0}
  \end{minipage}
  \hfill
  \begin{minipage}[b]{.49\linewidth}
    \centering
    \includegraphics[width=\linewidth]{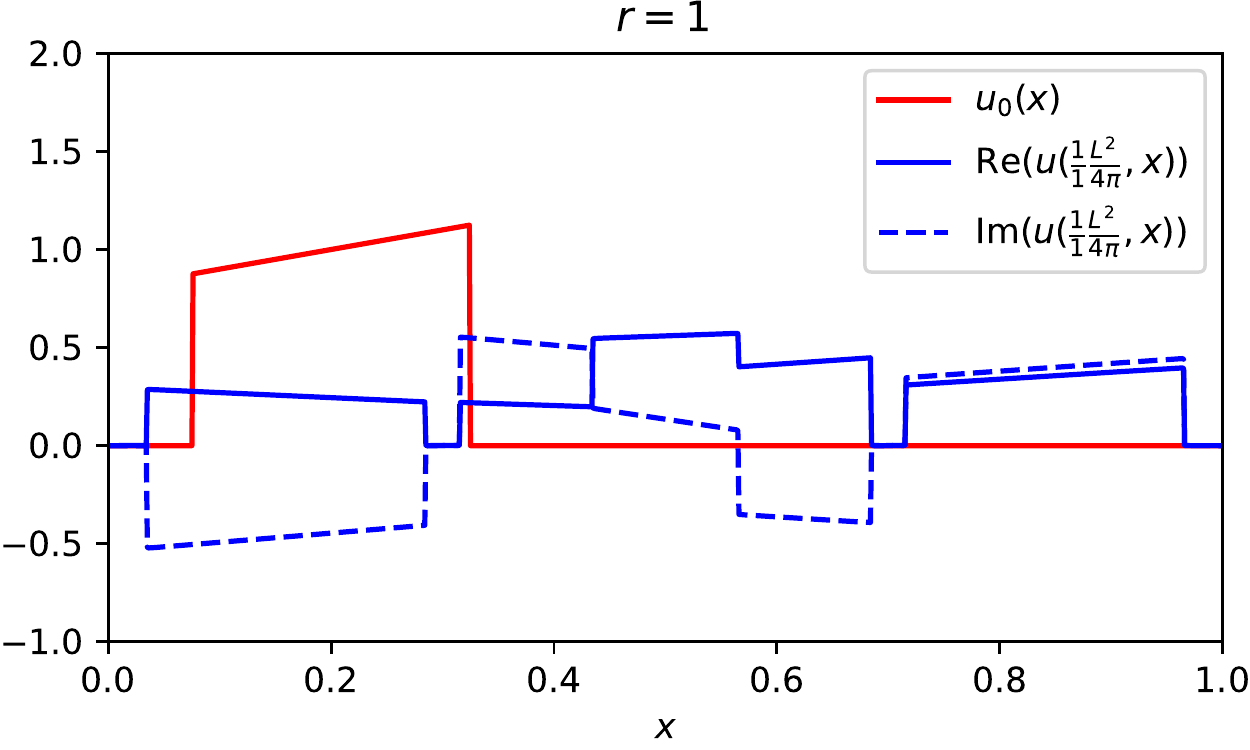}
    \subcaption{$t=\frac{1}{1}\frac{L^2}{4\pii}$, $r=1$}\label{fig:PiecewiseConstantUnnecessary.1}
  \end{minipage}
  \\
  \vspace{2ex}
  \begin{minipage}[b]{.49\linewidth}
    \centering
    \includegraphics[width=\linewidth]{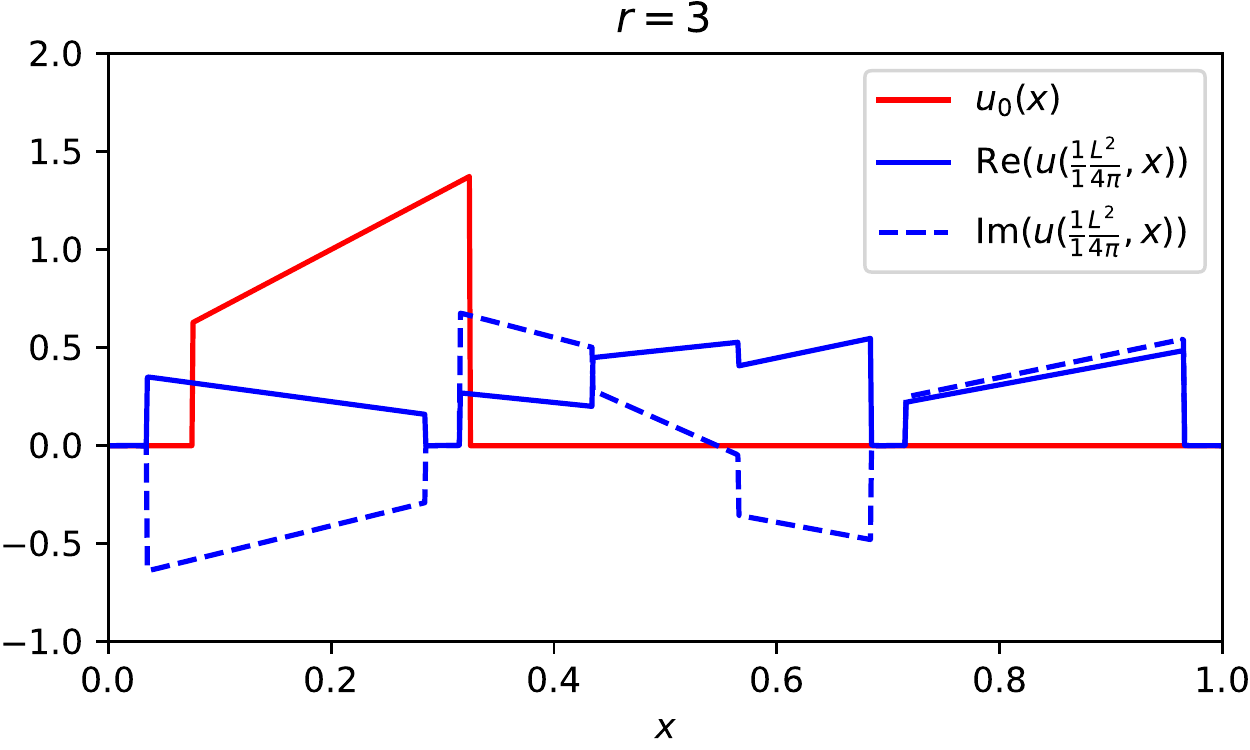}
    \subcaption{$t=\frac{1}{1}\frac{L^2}{4\pii}$, $r=3$}\label{fig:PiecewiseConstantUnnecessary.10}
  \end{minipage}
  \hfill
  \begin{minipage}[b]{.49\linewidth}
    \centering
    \includegraphics[width=\linewidth]{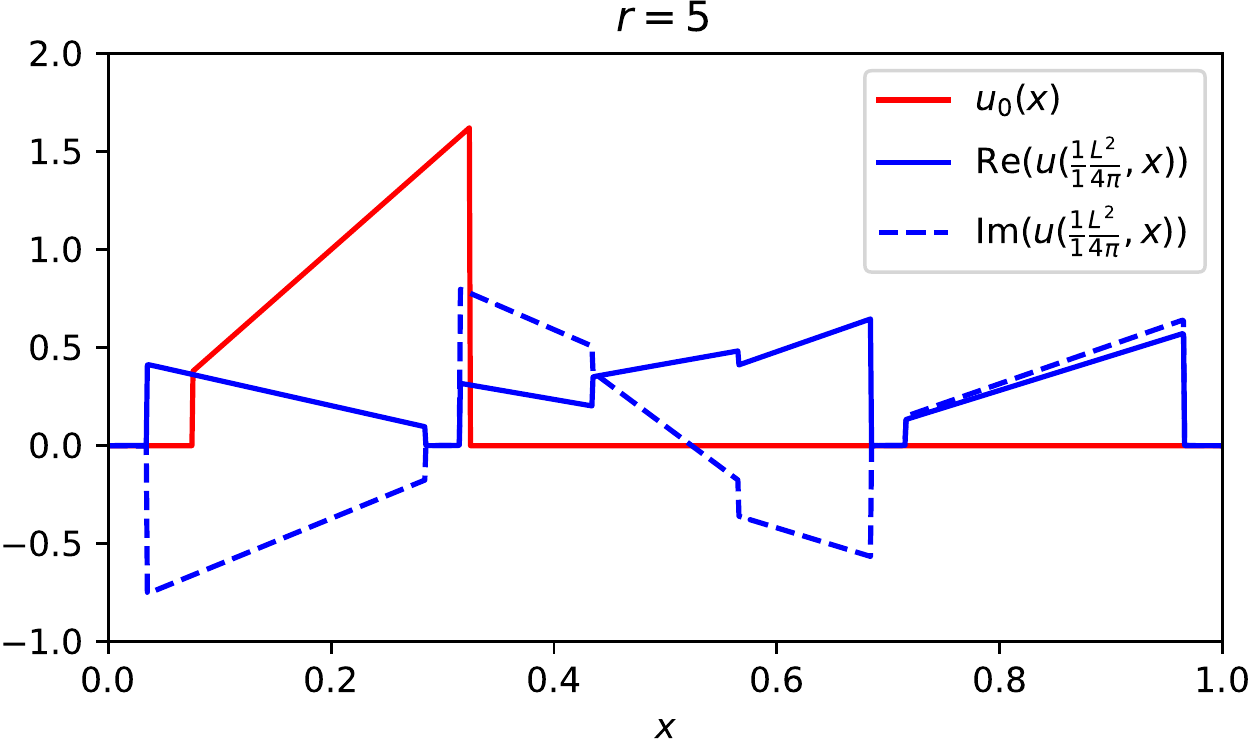}
    \subcaption{$t=\frac{1}{1}\frac{L^2}{4\pii}$, $r=5$}\label{fig:PiecewiseConstantUnnecessary.16}
  \end{minipage}
  \caption{
    The solution of the linear Schr\"{o}dinger equation with pseudoperiodic boundary conditions $\beta_0=1/5$, $\beta_1=2$ on $[0,1]$ and piecewise linear initial datum, evaluated at certain times.
  }
  \label{fig:PiecewiseConstantUnnecessary}
\end{figure}

To better study the precise dependence of the solution on the initial datum, we study the problem whose initial value is supported on a small subinterval of $[0,L]$.
Figure~\ref{fig:SmallSupport} suggests that the solution at rational times
is simply a certain linear combination of copies of the initial datum and its reflection, after certain shifts have been applied.
\begin{figure}[h!]
  \centering
  \begin{minipage}[b]{.32\linewidth}
    \centering
    \includegraphics[width=\linewidth]{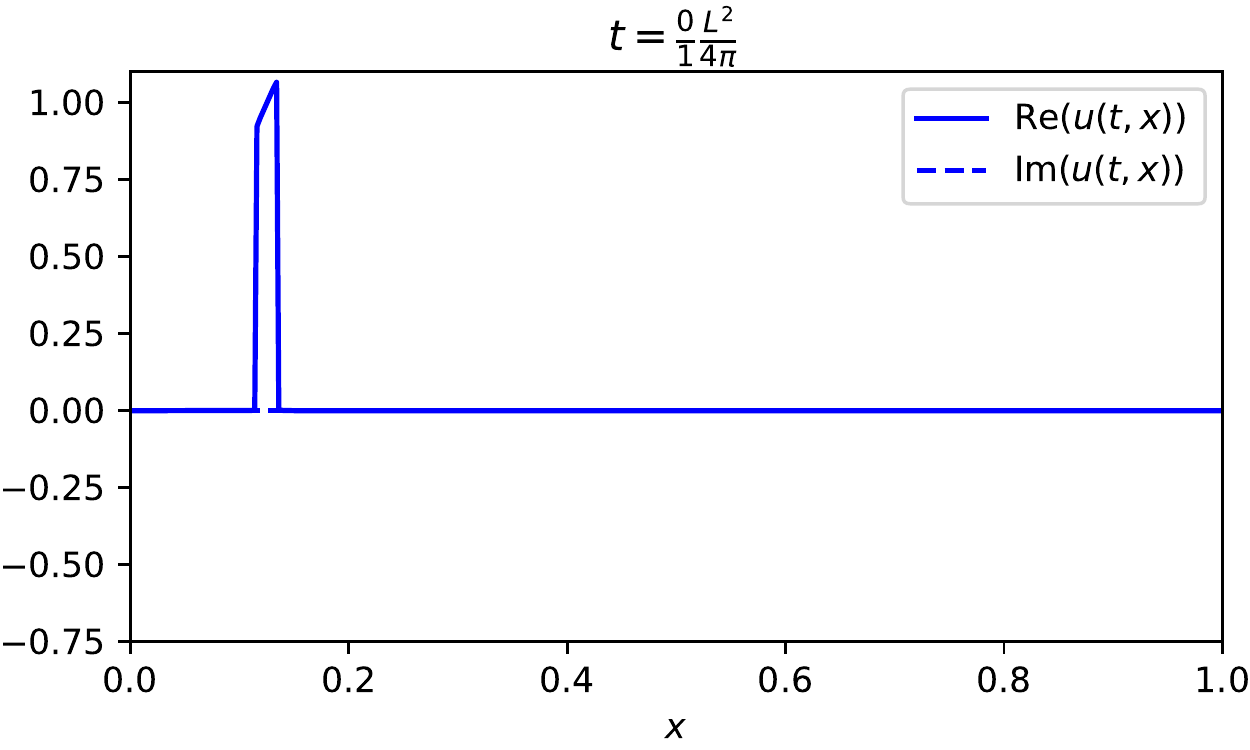}
    \subcaption{$t=0$}\label{fig:SmallSupport.1}
  \end{minipage}
  \hfill
  \begin{minipage}[b]{.32\linewidth}
    \centering
    \includegraphics[width=\linewidth]{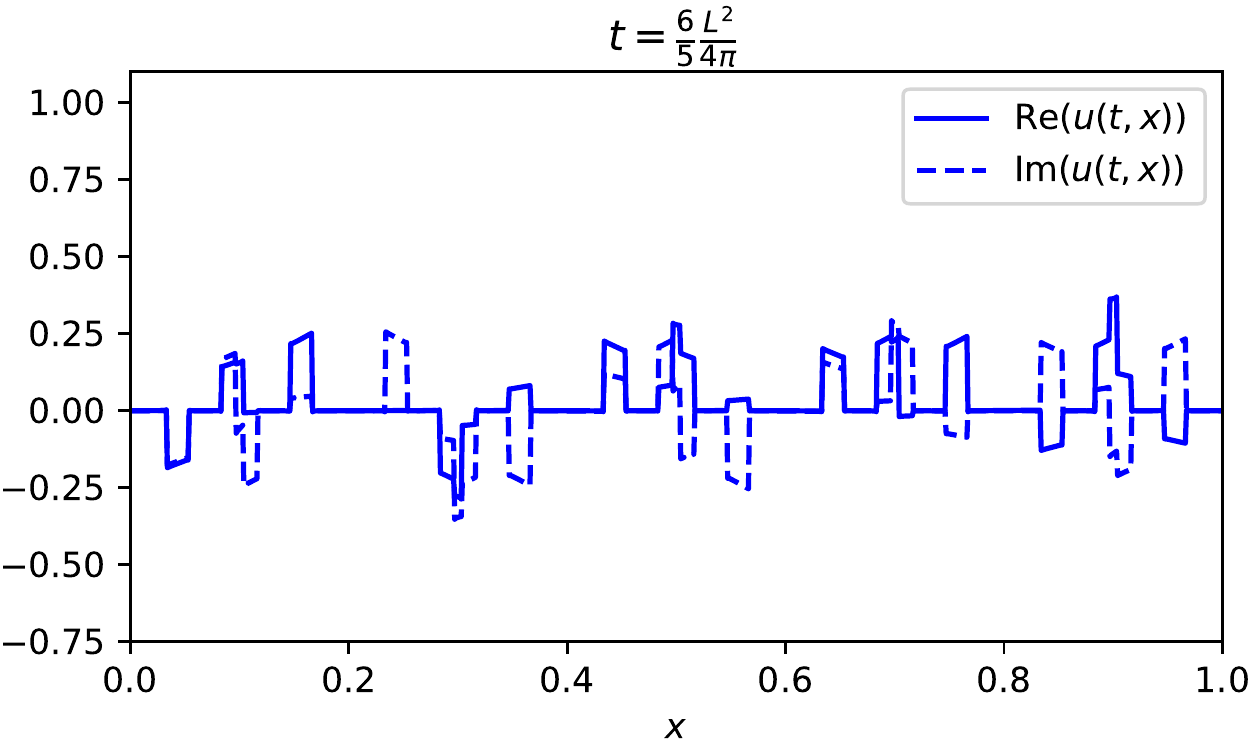}
    \subcaption{$\udsty t=\frac{6}{5}\frac{L^2}{4\pii}\approx0.09$}\label{fig:SmallSupport.2}
  \end{minipage}
  \hfill
  \begin{minipage}[b]{.32\linewidth}
    \centering
    \includegraphics[width=\linewidth]{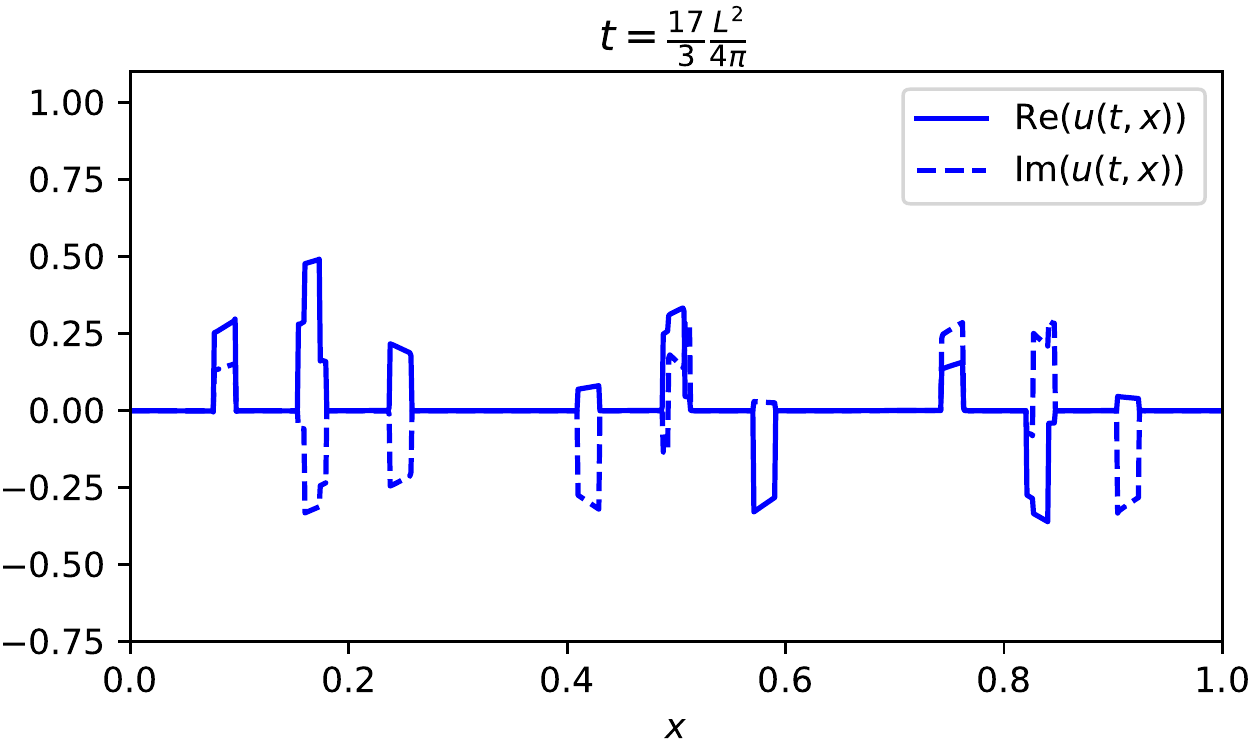}
    \subcaption{$\udsty t=\frac{17}{3}\frac{L^2}{4\pii}\frac{L^2}{4\pii}\approx0.45$}\label{fig:SmallSupport.6}
  \end{minipage}
  \\
  \vspace{2ex}
  \begin{minipage}[b]{.32\linewidth}
    \centering
    \includegraphics[width=\linewidth]{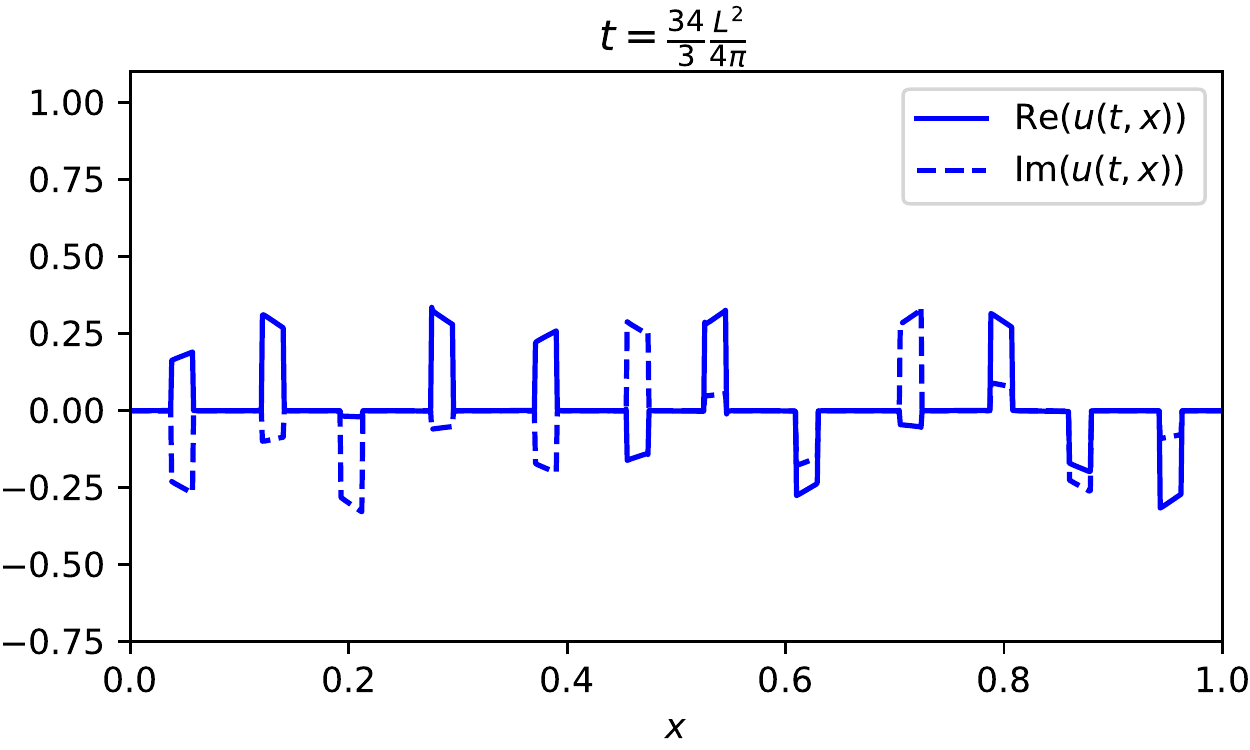}
    \subcaption{$\udsty t=\frac{34}{3}\frac{L^2}{4\pii}\approx0.9$}\label{fig:SmallSupport.11}
  \end{minipage}
  \hfill
  \begin{minipage}[b]{.32\linewidth}
    \centering
    \includegraphics[width=\linewidth]{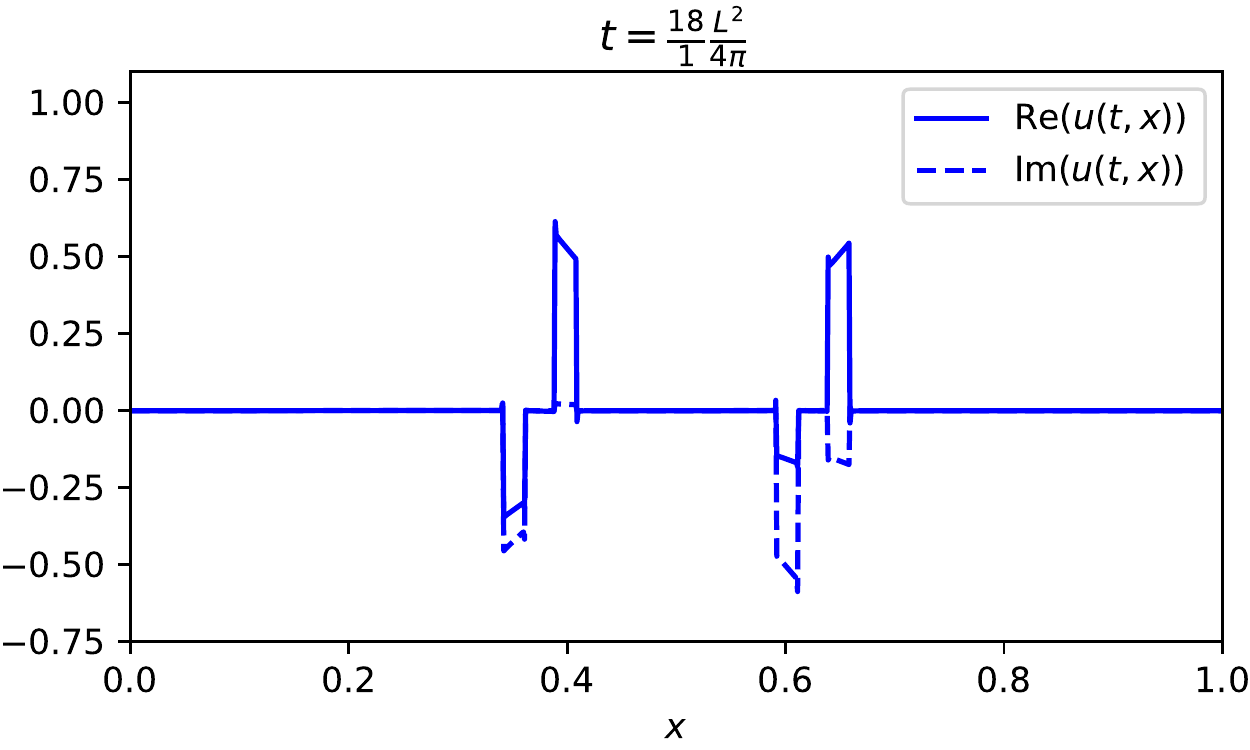}
    \subcaption{$\udsty t=\frac{18}{1}\frac{L^2}{4\pii}\approx1.44$}\label{fig:SmallSupport.17}
  \end{minipage}
  \hfill
  \begin{minipage}[b]{.32\linewidth}
    \centering
    \includegraphics[width=\linewidth]{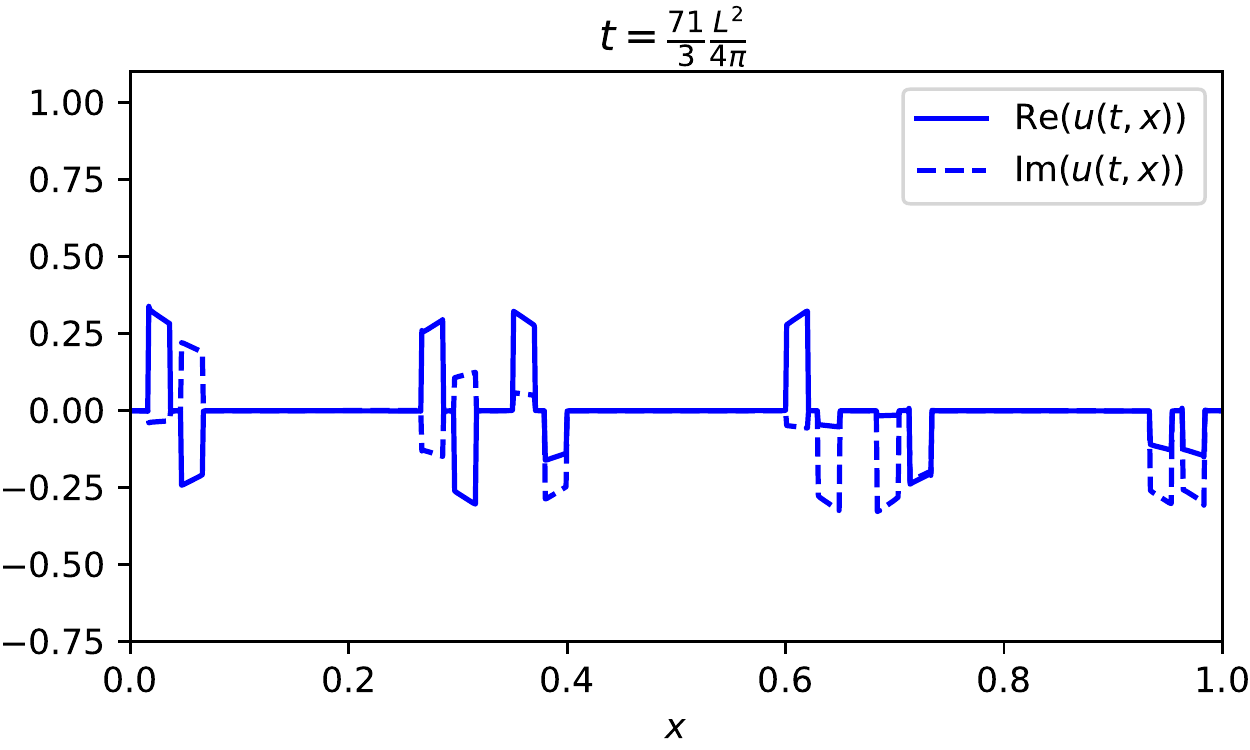}
    \subcaption{$\udsty t=\frac{71}{3}\frac{L^2}{4\pii}\approx1.89$}\label{fig:SmallSupport.22}
  \end{minipage}
  \caption{
    The solution of the linear Schr\"{o}dinger equation with pseudoperiodic boundary conditions $\beta_0=1/5$, $\beta_1=2$ on $[0,1]$ evaluated at ``rational" times commensurate with $L^2/(4\pii)$.
    The initial datum is $u_0(x) = 8x$ on $1/8-1/50 < x < 1/8+1/50$, and $x=0$ otherwise.
  }
  \label{fig:SmallSupport}
\end{figure}
Indeed the plots in Figure~\ref{fig:SmallSupport} correspond to plots in Figure~\ref{fig:PseudoNonconservativeRational}, but for an initial datum that is piecewise linear on its small support. 


The more careful study of the rational time dependence summarised in Figure~\ref{fig:FourVCopies} suggests that, at the rational time
\BE \label{eqn:TRelationUV}
  t = \frac{L^2}{4\pii} \times \frac{p}{q},
\EE
for $p$, $q$ coprime positive integers, there are $4\:q$ shifted and reflected copies of the initial datum.
\begin{figure}[h!]
  \centering
  \begin{minipage}[b]{.32\linewidth}
    \centering
    \includegraphics[width=\linewidth]{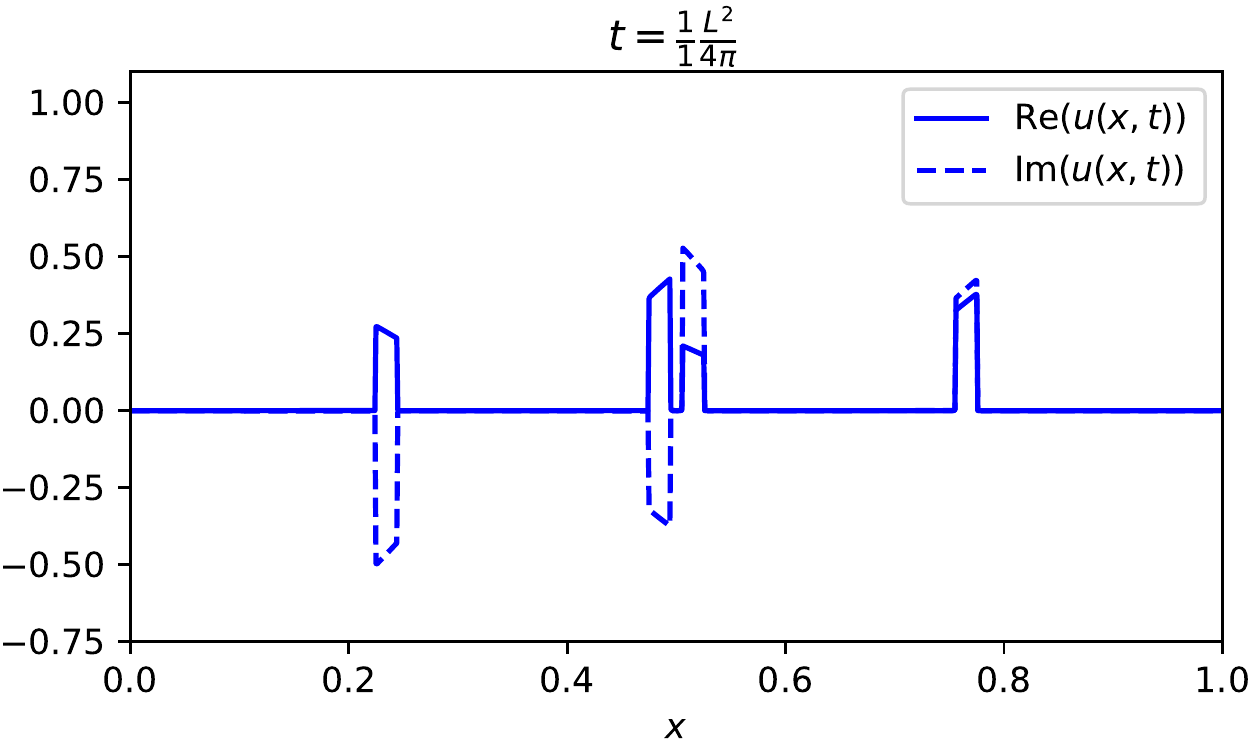}
    \subcaption{$\udsty \frac{p}{q}=\frac{1}{1}$}\label{fig:FourVCopies.u1_v1}
  \end{minipage}
  \hfill
  \begin{minipage}[b]{.32\linewidth}
    \centering
    \includegraphics[width=\linewidth]{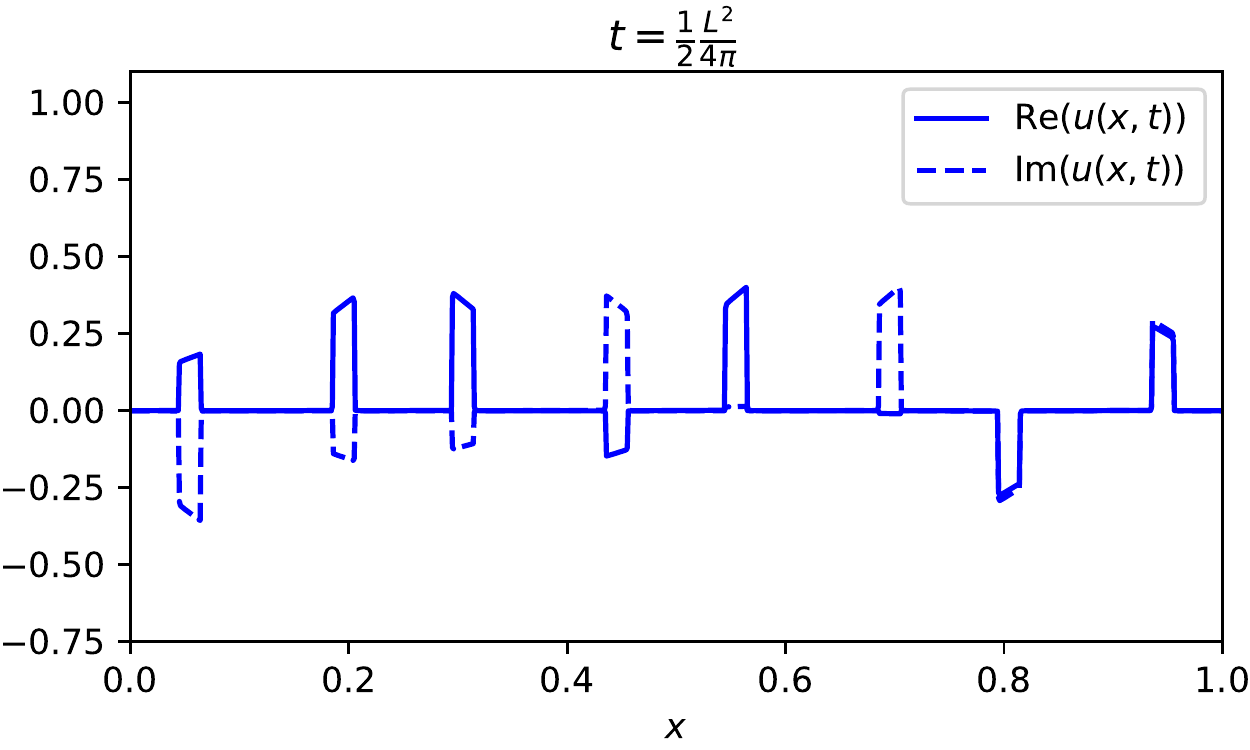}
    \subcaption{$\udsty \frac{p}{q}=\frac{1}{2}$}\label{fig:FourVCopies.u1_v2}
  \end{minipage}
  \hfill
  \begin{minipage}[b]{.32\linewidth}
    \centering
    \includegraphics[width=\linewidth]{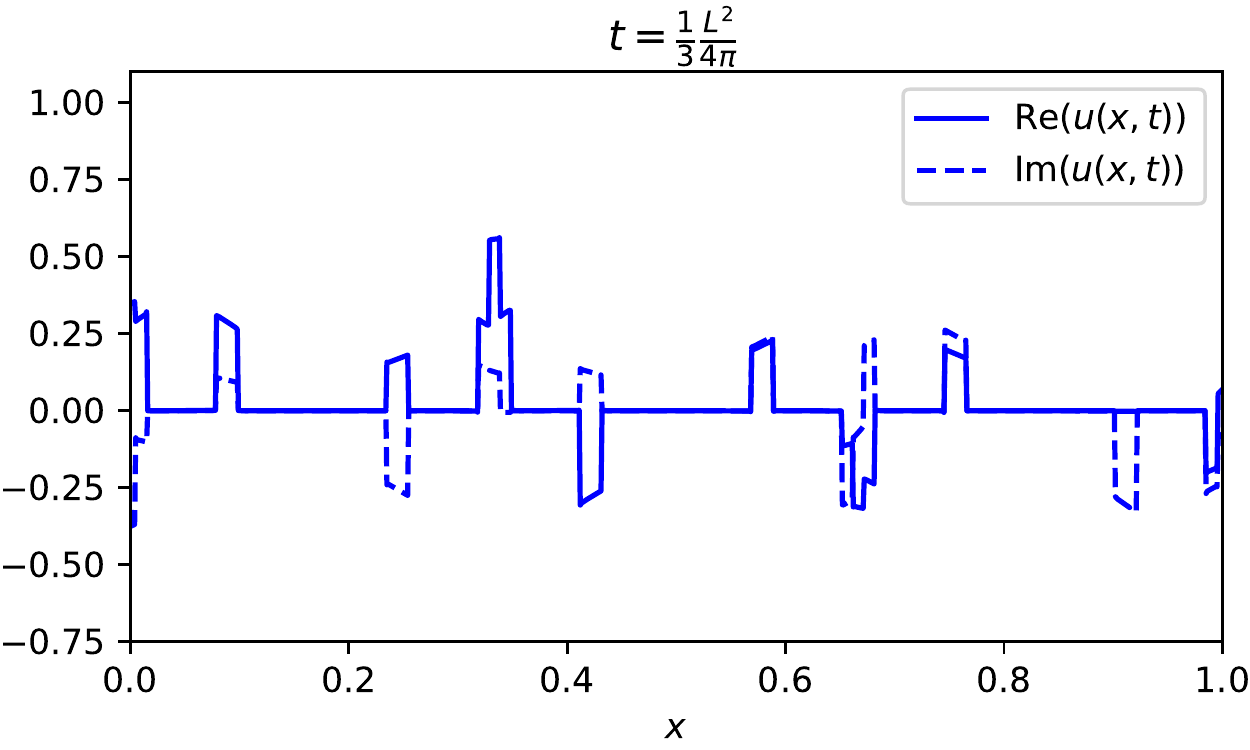}
    \subcaption{$\udsty \frac{p}{q}=\frac{1}{3}$}\label{fig:FourVCopies.u1_v3}
  \end{minipage}
  \\
  \vspace{2ex}
  \begin{minipage}[b]{.32\linewidth}
    \centering
    \includegraphics[width=\linewidth]{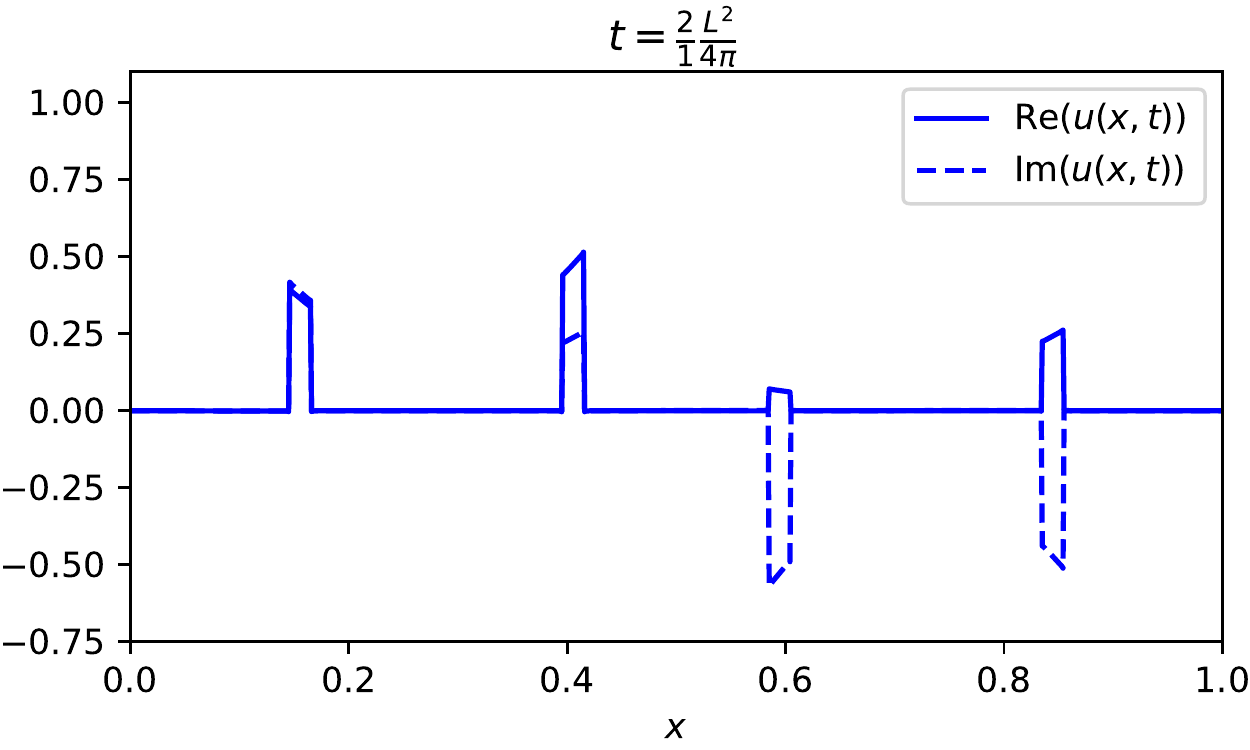}
    \subcaption{$\udsty \frac{p}{q}=\frac{2}{1}$}\label{fig:FourVCopies.u2_v1}
  \end{minipage}
  \hfill
  \begin{minipage}[b]{.32\linewidth}
    \centering
    \includegraphics[width=\linewidth]{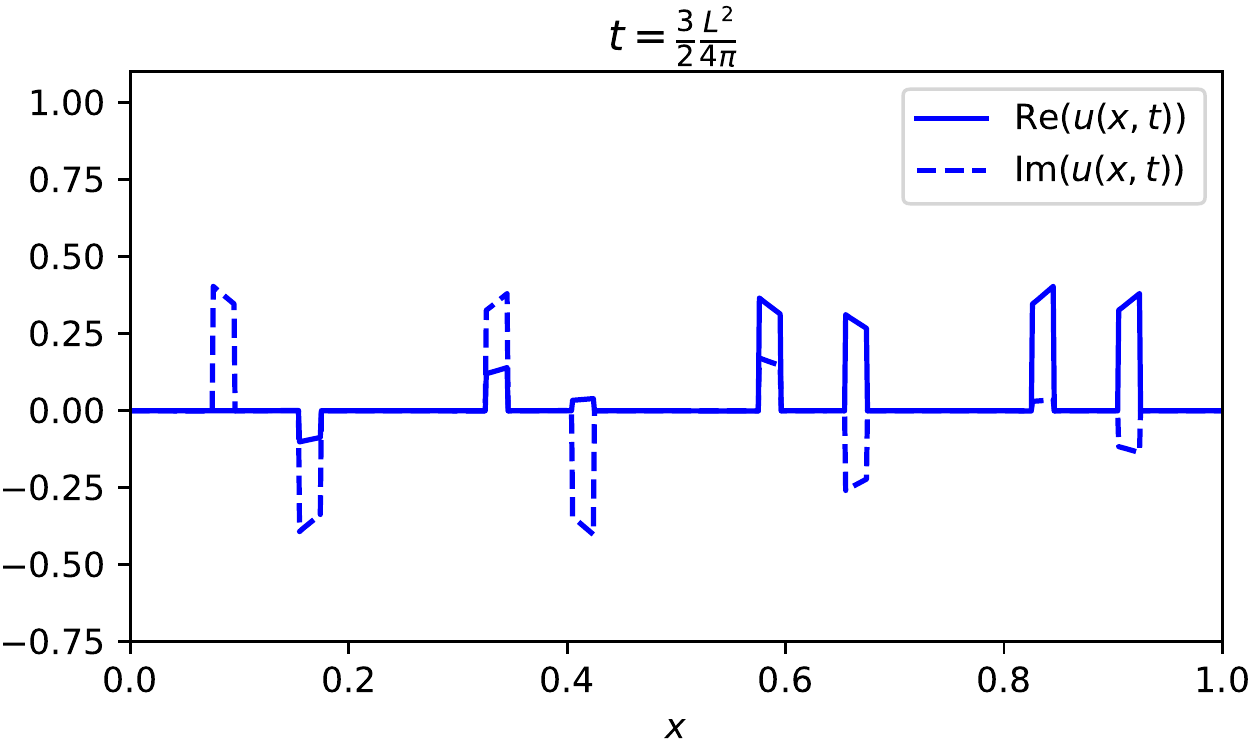}
    \subcaption{$\udsty \frac{p}{q}=\frac{3}{2}$}\label{fig:FourVCopies.u3_v2}
  \end{minipage}
  \hfill
  \begin{minipage}[b]{.32\linewidth}
    \centering
    \includegraphics[width=\linewidth]{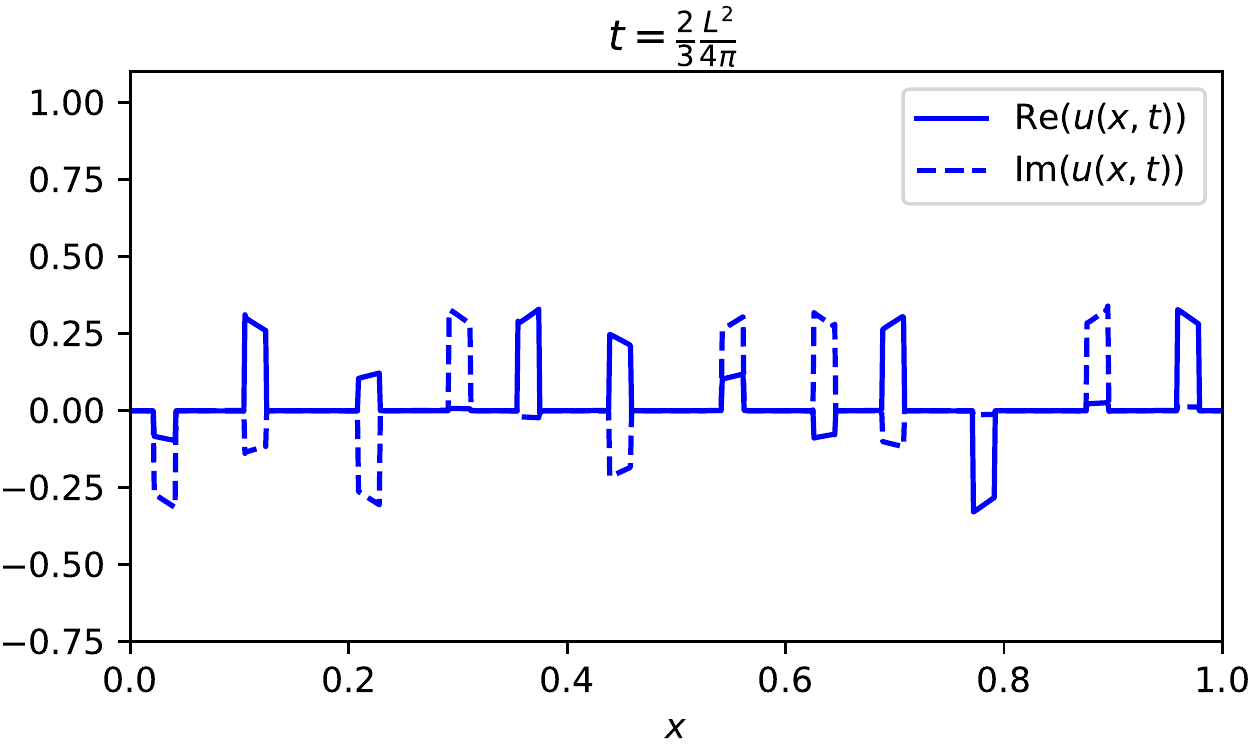}
    \subcaption{$\udsty \frac{p}{q}=\frac{2}{3}$}\label{fig:FourVCopies.u2_v3}
  \end{minipage}
  \caption{
    The solution of the linear Schr\"{o}dinger equation with pseudoperiodic boundary conditions $\beta_0=1/5$, $\beta_1=2$ on $[0,1]$ and initial datum of small support, evaluated at ``rational" times commensurate with $L^2/(4\pii)$ and $u_0(x) = 8x$ on $L(1/8-1/50) < x < L(1/8+1/50)$, and $x=0$ otherwise.
  }
  \label{fig:FourVCopies}
\end{figure}
Although it may at first appear unnatural to factor $4$ out of $q$ in equation~\eqref{eqn:TRelationUV}, comparing figures~\ref{fig:FourVCopies.u1_v1} and~\ref{fig:FourVCopies.u2_v1} suggests this is the proper approach.
It is not quite true to claim that, at such times, the support of $u(t,x)$ is $4\:q$ shifted copies of the support of $u_0(x)$ because, as demonstrated in Figure~\ref{fig:FourVCopies.u1_v3}, some of the connected intervals of support may overlap.
Nevertheless, when the support of $u_0(x)$ is much smaller than $L$, it is relatively easy to observe the linear interactions of the superimposed shifts of $u_0(x)$.
Thus, we see strong numerical evidence for a revival in the solution profiles at rational times.

At a fixed time, we can study how translating the initial datum affects the solution.
Indeed, Figure~\ref{fig:ShiftIndependentID} shows that translations of initial datum lead to translations within the solution. 
Specifically, shifts of the initial datum all translate to the right, while the shifted reflections all translate to the left. 
Moreover, by considering different initial data, one deduces that the shifts applied to each copy and reflection of the initial datum are independent thereof.
%
\begin{figure}[h!]
  \centering
  \begin{minipage}[b]{.49\linewidth}
    \centering
    \includegraphics[width=\linewidth]{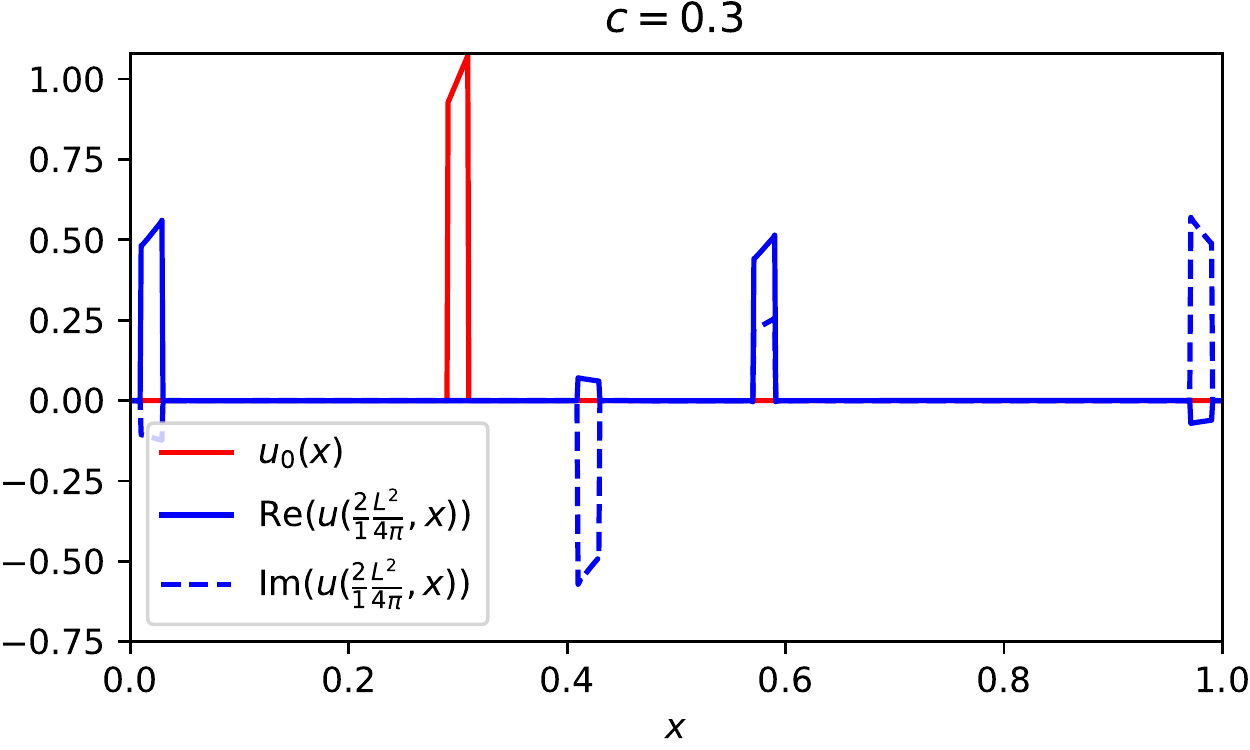}
    \subcaption{$c=0.3$}\label{fig:ShiftIndependentID.43}
  \end{minipage}
  \hfill
  \begin{minipage}[b]{.49\linewidth}
    \centering
    \includegraphics[width=\linewidth]{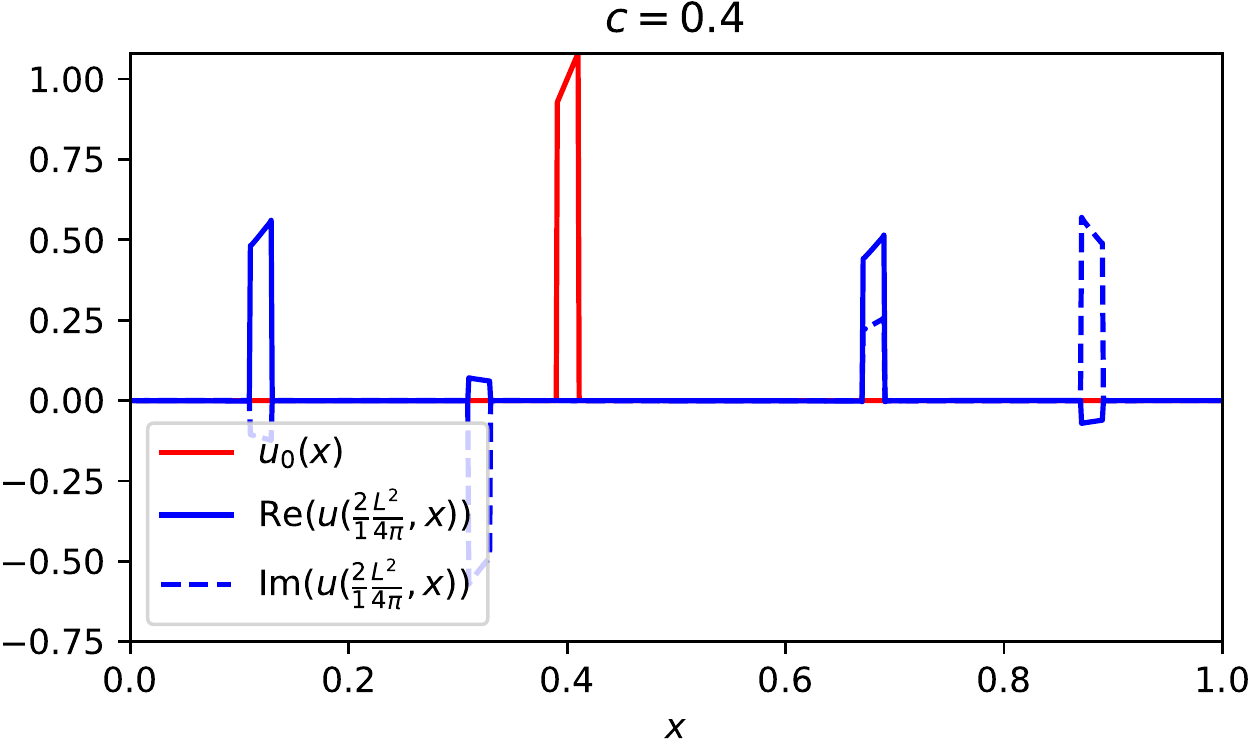}
    \subcaption{$c=0.4$}\label{fig:ShiftIndependentID.20}
  \end{minipage}
  \\
  \vspace{2ex}
  \begin{minipage}[b]{.49\linewidth}
    \centering
    \includegraphics[width=\linewidth]{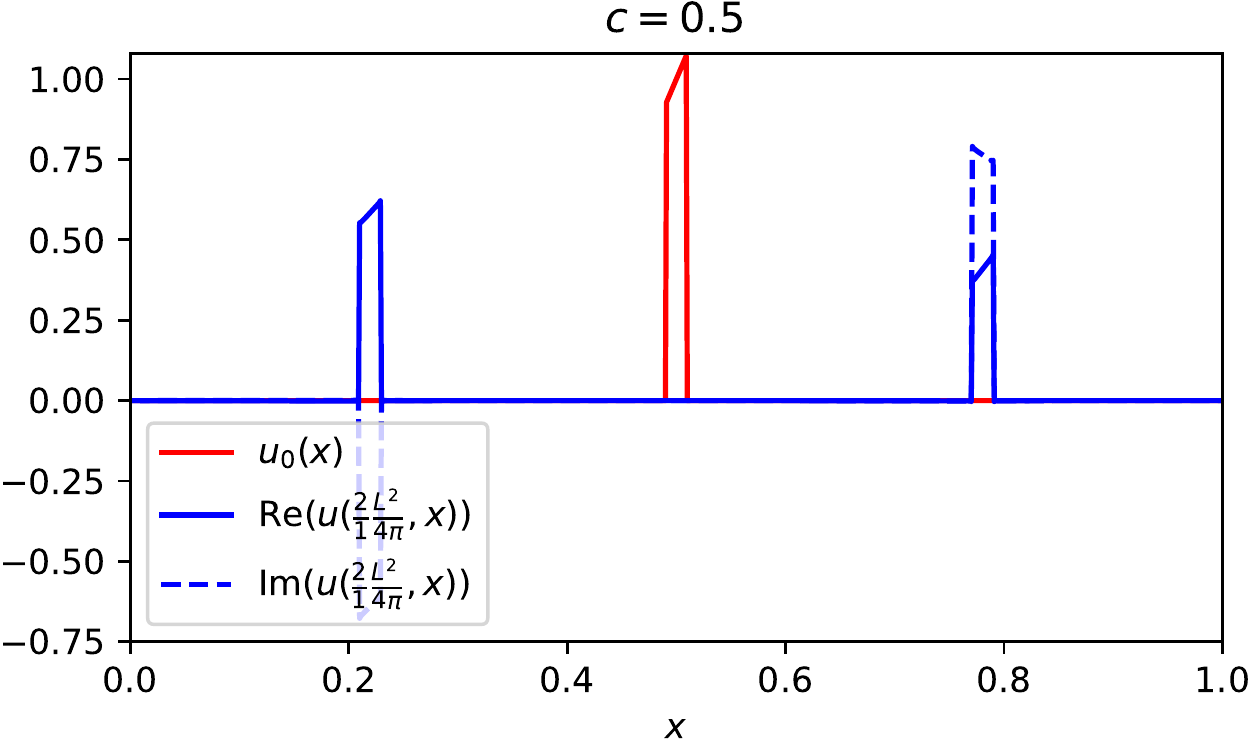}
    \subcaption{$c=0.5$}\label{fig:ShiftIndependentID.25}
  \end{minipage}
  \hfill
  \begin{minipage}[b]{.49\linewidth}
    \centering
    \includegraphics[width=\linewidth]{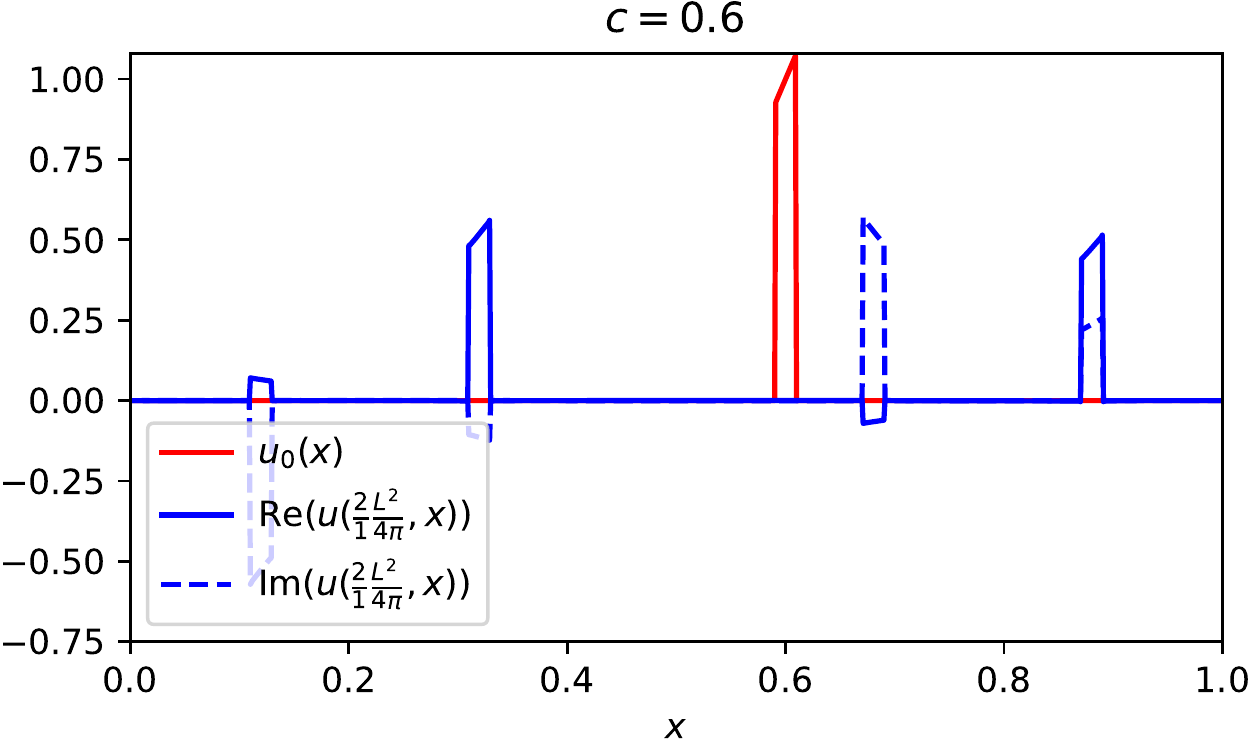}
    \subcaption{$c=0.6$}\label{fig:ShiftIndependentID.30}
  \end{minipage}
  \caption{
    The solution of the linear Schr\"{o}dinger equation with pseudoperiodic boundary conditions $\beta_0=1/5$, $\beta_1=2$ on $[0,1]$ and initial datum of small support, centred at different positions, evaluated at time $L^2/(4\pii)$.
    The initial data are $u_0(x) = 8(x-c)+1$ on $L(c-1/50) < x < L(c+1/50)$, and $x=0$ otherwise, for various $c$.
  }
  \label{fig:ShiftIndependentID}
\end{figure}

\section{Analytic characterization of revivals} \label{sec:analysis}

The numerical observations presented in Section~\ref{sec:NumericalResults} suggest the existence of revivals. 
Namely, at the rational time $\udsty t=\fracp{L^2p}{4\pii\:q}$ for positive coprime integers $p$ and $q$, the solution of the initial-boundary value problem is a linear combination of $4\:q$ copies of the initial datum, subject to certain shifts and reflections that are independent of the initial datum.


\subsection{{Shifts and Reflections}}

In order to analytically establish the existence of revivals we first need to understand the effect of the (bounded) Fourier transform on shifts (translations) and reflections of functions. 
Recall that application of the Fourier transform \eq{Ftr}
to a shift of a function produces an exponential factor multiplying the Fourier transform of the original function:
\BE \label{eqn:FourierTransformShifts}
  \mathcal{F}\bbk{\phi(\cdot-s)}(\kappa) = e^{\i \kappa s} \mathcal{F}[\phi](\kappa).
\EE
In the periodic case, one formulates a similar result by periodically extending the function $\phi $ that is initially defined on the interval $[0,L]$.
In the present situation, one needs to introduce a more general pseudoperiodic extension of the initial data and its reflection.

Explicitly, let $\phi(x)$ be a function defined for $0 \leq x \leq L$. 
Let $\phi(L-x)$ denote its reflection through the midpoint of the interval.
Recalling the definition \eq{gamma} of the quantity $\gamma $, which, by the following formulae, represents the phase shift applied to successive intervals in the full line scaled-periodic full line extensions of $\phi(x)$ and its reflection $\phi(L-x)$:
\begin{subequations} \label{eqn:DefnShifts}
\begin{align}
	\phi_0^\sharp(x;\gamma) ={}& \gamma^{m} \phi(x-mL), \quad\mbox{where } m\in\ZZ \mbox{ is such that } mL \leq x < (m+1)L, \\
	\phi_0^\flat (x;\gamma) ={}& \gamma^{m} \phi(mL-x), \quad\mbox{where } m\in\ZZ \mbox{ is such that } (m-1)L \leq x < mL.
\end{align}
We further define the shifts of these extensions
\begin{align}
	\phi_s^\sharp(x;\gamma) ={}& \phi_0^\sharp(x+s;\gamma), \qquad \qquad 
	\phi_s^\flat (x;\gamma) = \phi_0^\flat (x-s;\gamma).
\end{align}
\end{subequations}
Thus, the sharped functions represent shifts of the pseudoperiodic extension of $\phi $ while the flatted functions represent reflections thereof.

We next calculate their bounded Fourier transforms on the interval $[0,L]$.
Given $s\in\RR$ and $n\in\ZZ$ such that $nL \leq s < (n+1)L$,
\begin{align*}
	\widehat{\phi_s^\sharp}(\kap;\gamma) ={}& \int_0^L e^{-\i \kappa\:x}\, \phi_s^\sharp(x;\gamma) \ud x = \int_0^L e^{-\i \kappa\:x}\, \phi_0^\sharp(x+s;\gamma) \ud x \\
  	={}& \int_{0}^{(n+1)L-s} e^{-\i \kappa\:x}\, \gamma^{n} \phi(x+s-nL) \ud x + \int_{(n+1)L-s}^{L} e^{-\i \kappa\:x}\, \gamma^{n+1} \phi(x+s-(n+1)L) \ud x \\
  	={}& \gamma^{n} \int_{s-nL}^L e^{-\i \kap(y-s+nL)}\, \phi(y) \ud y + \gamma^{n+1} \int_0^{s-nL} e^{-\i \kap(y-s+(n+1)L)}\, \phi(y) \ud y.\\
  \widehat{\phi_s^\flat}(\kap;\gamma) ={}& \int_0^L e^{-\i \kappa\:x} \,\phi_s^\flat(x;\gamma) \ud x = \int_0^L e^{-\i \kappa\:x} \,\phi_0^\flat(x-s;\gamma) \ud x \\
  	={}& \int_{0}^{s-nL} e^{-\i \kappa\:x}\, \gamma^{-n} \phi(-nL-x+s) \ud x + \int_{s-nL}^{L} e^{-\i \kappa\:x}\, \gamma^{-(n-1)} \phi(-(n-1)L-x+s) \ud x \\
  	={}& \gamma^{-n} \int_{0}^{s-nL} e^{\i \kap(y+nL-s)} \,\phi(y) \ud y +\gamma^{-(n-1)} \int_{s-nL}^{L} e^{\i \kap(y+(n-1)L-s)}\, \phi(y) \ud y.
\end{align*}
Similarly, for $nL \leq s < (n+1)L$,
\begin{align*}
	\widehat{\phi_{s}^\sharp}(\kap;\gamma ) ={}& e^{\i \kappa s} \left[
		\gamma^{n+1}e^{-\i L(n+1)\kappa} \int_0^{s-nL} e^{-\i \kappa y} \phi(y) \ud y +
		\gamma^ne^{-\i L n\kappa } \int_{s-nL}^L e^{-\i \kappa y} \,\phi(y) \ud y \right], \\
	\widehat{\phi_{-s}^\flat} (\kap;\gamma) ={}& e^{ i\kappa s} \left[
		\gamma^{n+1}e^{-\i L(n+1)\kappa} \int_0^{(n+1)L-s} e^{\i \kappa y}\, \phi(y) \ud y +
		\gamma^{n+2}e^{-\i L(n+2) \kappa} \int_{(n+1)L-s}^L e^{\i \kappa y}\, \phi(y) \ud y \right],\\
	\widehat{\phi_{s}^\flat} (-\kap;\gamma^{-1} ) ={}& e^{\i \kappa s} \left[
		\gamma^n e^{-\i L n \kappa  } \int_0^{s-nL} e^{-\i \kappa y}\, \phi(y) \ud y +
		\gamma^{n-1}e^{-\i L(n-1)\kappa} \int_{s-nL}^L e^{-\i \kappa y}\, \phi(y) \ud y \right], \\
	\widehat{\phi_{-s}^\sharp}(-\kap;\gamma^{-1}) ={}& e^{\i \kappa s} \left[
		\gamma^ne^{-\i L n\kappa  } \int_0^{(n+1)L-s} e^{\i \kappa y} \,\phi(y) \ud y +
		\gamma^{n+1}e^{-\i L(n+1)\kappa)} \int_{(n+1)L-s}^L e^{\i \kappa y} \,\phi(y) \ud y \right]. 
	\end{align*}
Evaluating some of the above at the zeros of the discriminant $\kap=\pm\kappa_j$, we have
\Eq{eqn:sharp_flat_q0}
$$\eeq{	\widehat{\phi_{s}^\sharp}(\kappa_j;\gamma) = e^{\i \kappa_j s} \,\widehat{\phi}(\kappa_j), \\
	\widehat{\phi_{-s}^\flat} (\kappa_j;\gamma) = e^{ \i \kappa_j s} \,\widehat{\phi}(-\kappa_j),\\
	\widehat{\phi_{s}^\flat} (-\kappa_j;\gamma^{-1} ) = e^{\i \kappa_j s}\, \widehat{\phi}(\kappa_j), \\
	\widehat{\phi_{-s}^\sharp}(-\kappa_j;\gamma^{-1}) = e^{\i \kappa_j s} \,\widehat{\phi}(-\kappa_j) , \\
	\widehat{\phi_{s}^\sharp}(-\kappa_j;\gamma) = e^{-\i \kappa_j s} \left[
		\gamma^{2(n+1)} \int_0^{s-nL} e^{\i \kappa_j y} \,\phi(y) \ud y +
		\gamma^{2n} \int_{s-nL}^L e^{\i \kappa_j y} \,\phi(y) \ud y \right], \\
	\widehat{\phi_{-s}^\flat} (-\kappa_j;\gamma) = e^{-\i \kappa_j s} \left[
		\gamma^{2(n+1)}\!\! \int_0^{(n+1)L-s} \!\!\! e^{-\i \kappa_j y}\, \phi(y) \ud y +
		\gamma^{2(n+2)}\!\! \int_{(n+1)L-s}^L \!\!\!e^{-\i \kappa_j y} \,\phi(y) \ud y \right],\\
	\widehat{\phi_{s}^\flat} (\kappa_j;\gamma^{-1}) = e^{-\i \kappa_j s} \left[
		\gamma^{2n} \int_0^{s-nL} e^{\i \kappa_j y}\, \phi(y) \ud y +
		\gamma^{2(n-1)} \int_{s-nL}^L e^{\i \kappa_j y} \,\phi(y) \ud y \right], \\
	\widehat{\phi_{-s}^\sharp}(\kappa_j;\gamma^{-1}) = e^{-\i \kappa_j s} \left[
		\gamma^{2n} \int_0^{(n+1)L-s} e^{-\i \kappa_j y}\, \phi(y) \ud y +
		\gamma^{2(n+1)} \int_{(n+1)L-s}^L e^{-\i \kappa_j y}\, \phi(y) \ud y \right].
}$$
Evaluating our series solution~\eqref{eqn:seriessoln} at $t=0$, we find
\begin{equation}
	u_0(x) = \frac{1}{L} \sum_{j\in\ZZ} \left[\frac{ \left((\beta_0+\beta_1)\gamma-2\right)e^{\i \kappa_j x} + (\beta_1-\beta_0)\gamma e^{-\i \kappa_j x}}{(\beta_0+\beta_1)(\gamma-\gamma^{-1})} \right]\left[\,\widehat{u_0}(\kappa_j)+ \delta \,\widehat{u_0}(-\kappa_j) \,\right].
\end{equation}
This generalised Fourier series representation holds not only for the initial datum $u_0(x)$, but for any piecewise smooth function defined on $[0,L]$. 
Hence, it is valid for any $F(\cdot;s)$ that is a linear combination of the sharp and flat shifted functions defined in~\eqref{eqn:sharp_flat_q0}, and hence
\begin{equation}\label{eqn:Fq0.gen}
	F(x;s)= \frac{1}{L} \sum_{j\in\ZZ} \left[\frac{ \left((\beta_0+\beta_1)\gamma-2\right)e^{\i \kappa_j x} + (\beta_1-\beta_0)\gamma e^{-\i \kappa_j x}}{(\beta_0+\beta_1)(\gamma-\gamma^{-1})} \right]\left(\widehat{F}(\kappa_j;s)+ \delta \,\widehat{F}(-\kappa_j;s) \right).
\end{equation}
In particular, we set
\BE \label{eqn:Asymptotics.defn.F.gen}
	F(x;s)
	= c_1\phi_{ s}^\sharp(x;\gamma)+c_2\phi_{ s}^\flat (x;\gamma^{-1})+c_3\psi_{-s}^\sharp(x;\gamma^{-1})+c_4\psi_{-s}^\flat (x;\gamma ),
\EE
where
$$
 \req{ c_1=c_4=\frac{1}{1-\delta^2\gamma^2},\\c_2=c_3=\frac{-\delta\gamma^2}{1-\delta^2\gamma^2},}
$$
for piecewise smooth $\phi,\psi:[0,L]\to\CC$ to be specified later.
Using~\eqref{eqn:sharp_flat_q0}, we see that the parenthetical term on the right of~\eqref{eqn:Fq0.gen} can be written in terms of their bounded Fourier transforms: 
$$  \widehat{F}(\kappa_j;s)+ \delta \,\widehat{F}(-\kappa_j;s) = e^{\i\kappa_js} \left[ \widehat{\phi}(\kappa_j) +\delta \,\widehat{\psi}(-\kappa_j)\right].$$

Finally, the identity \eqref{eqn:Fq0.gen} reduces to equation~\eqref{eqn:Analytic.F.looks.like.q.gen} in our statement of Theorem~\ref{thm:ShiftRep.uGenvGen}.
The similarity between the formulae on the right sides of equations~\eqref{eqn:seriessoln} and~\eqref{eqn:Analytic.F.looks.like.q.gen} is clear, despite their very different origins.
The former appears in an equation for the solution of the initial-boundary value problem evaluated at time $t$, and the latter is a representation of the compound function $F$ defined as a linear combination of the $\sharp$ shifts and $\flat$ reflections applied to arbitrary piecewise smooth functions $\phi$ and $\psi$. 
The formulae differ only in that $\widehat{u_0}(\pm\kappa_j)$ has been replaced by $\widehat{\phi}(\kappa_j)$ and $\widehat{\psi}(-\kappa_j)$, and the initial exponential factor is different.


\subsection{The case $q=1$}

We begin by studying the simplest case of ``integral time,'' in which the denominator in \eq{eqn:TRelationUV} is $q=1$.
Here, in accordance with the general observations, we expect to obtain a formula for $u(t,x)$ as a linear combination of precisely four shifted and reflected copies of $u_0(x)$.

\begin{eg}\label{ex:uEven.v1}
  Given that $q=1$, assume that $p$ is a positive even integer. We aim to find $s\in\RR$, $c\in\CC$, and appropriate functions $\phi(x)$, $\psi(x)$ such that
\Eq{eic}
$$e^{\i c}F(x;s)=u\left(\frac{L^2p}{4\pii},x\right).$$
  Let $\phi(x)=\psi(x)=u_0(x)$. 
  Then, by equations~\eqref{eqn:seriessoln} and~\eqref{eqn:Analytic.F.looks.like.q.gen} at $t=\fracp{L^2p}{4\pii}$, finding $c$ and $s$ is equivalent to solving the equations
  $$
    \exp(\i c+\i \kappa_js) = \exp\left( \frac{-\i \kappa_j^2\:p\:L^2}{4\pii} \right) \qquad \hbox{for all $j\in\ZZ$.}
  $$
  Taking logarithms and applying equation~\eqref{eqn:Lambdaj.gen}, this requires
  \BE \label{eqn:cs.ExponentCriterion}
    c+s\kappa_0 + s\frac{2j\pi}{L} + 2\:n\:\pi = - p\:j^2\pi - p\:j\:L\:\kappa_0 - \frac{p\:L^2\kappa_0^2}{4\pii},
  \EE
  for all $j\in\ZZ$ and some $n\in\ZZ$.
  Equating powers of $j$ and using $n$ to balance the terms of order $j^2$, we find
  \begin{equation}\label{eqn:scn.solution.uEven.v1}
    s = -\,\frac{p\:L^2\:\kappa_0}{2\pii}\in\RR, \qquad c = \frac{p\:L^2\:\kappa_0^2}{4\pii}\in\CC, \qquad n=-\,\frac{p\:j^2}{2}\in\ZZ.
  \end{equation}
\end{eg}
In Example~\ref{ex:uEven.v1}, we exploited the assumption $\fracz{p}{q} = p$ is an even integer to absorb the quadratic involving $\kappa_j^2$ into $2\pii n$.
Note that this is the only option since all other terms must be linear in $j$, while $n$ is free to be any integer, including an integer-valued function of $j$.
On the other hand, if $p = \fracz{p}{q}$ is an odd integer, the resulting $n$ will not be an integer-valued function of $j$.
Our next example discusses how to navigate this difficulty.

\begin{eg}\label{ex:uOdd.v1}
Now assume that $q=1$ and $p$ is a positive odd integer. 
As above, we aim to find $s\in\RR$, $c\in\CC$, and appropriate functions $\phi(x)$, $\psi(x)$ such that \eq{eic} holds.
Let us set 
\BE
	\phi(x) =
	\begin{cases}
		\ \gamma^{-\frac{1}{2}} u_0\left(x+\frac{1}{2}\:L\right) , &0 \leq x \leq \frac{1}{2}\:L, \\
		\ \gamma^{\frac{1}{2}} u_0\left(x-\frac{1}{2}\:L\right), & \frac{1}{2}\:L < x \leq L,
	\end{cases}
	\qquad
	\psi(x) =
	\begin{cases}
		\ \gamma^{\frac{1}{2}} u_0\left(x+\frac{1}{2}\:L\right), & 0 \leq x \leq \frac{1}{2}\:L, \\
		\ \gamma^{-\frac{1}{2}} u_0\left(x-\frac{1}{2}\:L\right), & \frac{1}{2}\:L < x \leq L.
	\end{cases}
\EE
Let $F(x;s)$ then be defined by equation~\eqref{eqn:Analytic.F.looks.like.q.gen}.
Observe that
\begin{align*}
	\widehat{\phi}(\kappa_j) ={}& \int_0^L e^{-\i \kappa_jx}\phi(x)\ud x = \int_0^\frac{L}{2} e^{-\i \kappa_jx}\gamma^{-\frac{1}{2}}u_0\left(x+\frac{L}{2}\right)\ud x+ \int_\frac{L}{2}^L e^{-\i \kappa_jx}\gamma^{\frac{1}{2}}u_0\left(x-\frac{L}{2}\right)\ud x \\
	={}& \int_\frac{L}{2}^L e^{-\i \kappa_jy}e^{\frac{\i\kappa_0L}{2}+ij\pi L}\gamma^{-\frac{1}{2}}u_0\left(y\right)\ud y+ \int_0^\frac{L}{2} e^{-\i \kappa_j y}e^{-\frac{\i\kappa_0L}{2}-\i j\pi L}\gamma^{\frac{1}{2}}u_0\left(y\right)\ud y \\
	={}&(-1)^j \int_\frac{L}{2}^L e^{-\i \kappa_jy}u_0\left(y\right)\ud y+ (-1)^j\int_0^\frac{L}{2} e^{-\i \kappa_j y}u_0\left(y\right)\ud y  = (-1)^j \widehat{u_0}(\kappa_j).
\end{align*}
A similar calculation reveals 
$$\widehat{\psi}(-\kappa_j)=(-1)^j \,\widehat{u_0}(-\kappa_j).$$
Further, since $p$ is odd it follows that
\begin{equation}\label{eqn:-1fact}
(-1)^j=(-1)^{j^2}=(-1)^{p\:j^2}=e^{-\i \pi p\:j^2}.
\end{equation}
Combining this fact with~\eqref{eqn:Analytic.F.looks.like.q.gen} produces
\begin{equation} \label{eqn:F.uOdd.v1}
\begin{split}
  e^{\i c}F(x;s)=\frac{1}{L} \sum_{j\in\ZZ}e^{\i\kappa_js-\i \pi pj^2}&\left[\frac{ \left((\beta_0+\beta_1)\gamma-2\right)e^{\i \kappa_j x} + (\beta_1-\beta_0)\gamma e^{-\i \kappa_j x}}{(\beta_0+\beta_1)(\gamma-\gamma^{-1})} \right]\,\left[\,\widehat{u_0}(\kappa_j) +\delta \,\widehat{u_0}(-\kappa_j)\,\right].
  \end{split}
\end{equation}
Now, equating~\eqref{eqn:F.uOdd.v1} and~\eqref{eqn:seriessoln} at time $t=\fracp{p\:L^2}{4\pii}$ yields
\Eq{expeq}
$$
  \exp(\i c+\i\kappa_js-\i \pi p j^2) = \exp\left( \frac{-\i p\:L^2\kappa_j^2}{4\pii} \right).
$$
Equating powers of $j$, we find that
\begin{equation}\label{eqn:sc.solution.uOdd.v1}
  s = \frac{-p\:L^2\kappa_0}{2\pii}\in\RR, \qquad c = \frac{p\:L^2\kappa_0^2}{4\pii}\in\CC,
\end{equation}
satisfies \eq{expeq} for all $j\in\ZZ$. 
\end{eg}

\subsection{General rational times}

With these two examples in hand we are now ready to prove Theorem~\ref{thm:ShiftRep.uGenvGen} for a general rational time.
That is, at time given by \eq{eqn:TRelationUV} with denominator $q \geq 1$.
\begin{proof}[Proof of Theorem~\ref{thm:ShiftRep.uGenvGen}]
  By the definition of $\phi$, for all $j\in\ZZ$,
	\BES
		\widehat{\phi}(\kappa_j) = \sum_{\ell=1}^q \int_{L\left(1-\frac{\ell}{q}\right)}^{L\left(1-\frac{\ell-1}{q}\right)} e^{-\i \kappa_jx}\phi(x)\ud x,
	\EES
	so
  \begin{align*}
		e^{\i \frac{\pi p}{q}j^2}\widehat{\phi}(\kappa_j) ={}& \sum_{\ell=1}^q \int_{L\left(1-\frac{\ell}{q}\right)}^{L\left(1-\frac{\ell-1}{q}\right)} e^{-\i \left(\kappa_jx-\pi pj^2/q\right)} \phi(x)\ud x \\
		={}& \frac{1}{q} \sum_{\ell=1}^q \sum_{m=0}^{\ell-1} \gamma^{ -\frac{m}{q}} \sum_{n=0}^{q-1} \alpha^{-2nm-pn^2} \int_{L\left(1-\frac{\ell}{q}\right)}^{L\left(1-\frac{\ell-1}{q}\right)} e^{-\i \left(\kappa_jx-\pi pj^2/q\right)} u_0\left(x+\frac{Lm}{q} \right)\ud x \\
		&+ \frac{1}{q} \sum_{\ell=1}^q \sum_{m=\ell}^{q-1} \gamma^{1-\frac{m}{q}} \sum_{n=0}^{q-1} \alpha^{-2nm-pn^2} \int_{L\left(1-\frac{\ell}{q}\right)}^{L\left(1-\frac{\ell-1}{q}\right)} e^{-\i \left(\kappa_jx-\pi pj^2/q\right)} u_0\left(x+\frac{Lm}{q}-L\right)\ud x.
	\end{align*}
	We make the changes of variables
		$$
		y = x+L \Pa{\frac{m}{q} + \varepsilon },
	$$
	where $\varepsilon =0$ in the first integral and $\varepsilon =-1$ in the second, to obtain
	\begin{align*}
		e^{\i \pi pj^2/q}\widehat{\phi}(\kappa_j) =\frac{1}{q} \sum_{\ell=1}^q {}& \sum_{m=0}^{\ell-1} \sum_{n=0}^{q-1} \alpha^{2(j-n)m+p(j^2-n^2)} \int_{L\left(1+\frac{m-\ell}{q}\right)}^{L\left(1+\frac{m-\ell+1}{q}\right)} e^{-\i \kappa_jy}u_0(y)\ud y \\
		&+ \frac{1}{q} \sum_{\ell=1}^q \sum_{m=\ell}^{q-1} \sum_{n=0}^{q-1} \alpha^{2(j-n)m+p(j^2-n^2)} \int_{L\left(\frac{m-\ell}{q}\right)}^{L\left(\frac{m-\ell+1}{q}\right)} e^{-\i \kappa_jy}u_0(y)\ud y.
	\end{align*}
	We reorder the sums over $\ell,m$ so that the outer sum is over $m$:
	\begin{align*}
		e^{\i \frac{\pi p}{q}j^2}\widehat{\phi}(\kappa_j)
			= \frac{1}{q} \sum_{m=0}^{q-1} {}&\sum_{n=0}^{q-1} \alpha^{2(j-n)m+p(j^2-n^2)} \sum_{\ell=m+1}^q \int_{L\left(1+\frac{m-\ell}{q}\right)}^{L\left(1+\frac{m-\ell+1}{q}\right)} e^{-\i \kappa_jy}u_0(y)\ud y \\
			&+ \frac{1}{q} \sum_{m=0}^{q-1} \sum_{n=0}^{q-1} \alpha^{2(j-n)m+p(j^2-n^2)} \sum_{\ell=1 }^m \int_{L\left(\frac{m-\ell}{q}\right)}^{L\left(\frac{m-\ell+1}{q}\right)} e^{-\i \kappa_jy}u_0(y)\ud y, \\
	\intertext{where, in the second line, we have introduced the additional term $m=0$, justified by the corresponding sum over $\ell$ from $1$ to $m$ being empty. Then,}
	e^{\i \pi pj^2/q}\widehat{\phi}(\kappa_j)	={}& \frac{1}{q} \sum_{m=0}^{q-1} \sum_{n=0}^{q-1} \alpha^{2(j-n)m+p(j^2-n^2)} \left\{ \int_{\frac{Lm}{q}}^L + \int_{0}^{\frac{Lm}{q}} \right\} e^{-\i \kappa_jy}u_0(y)\ud y \\
		={}& \frac{1}{q} \,\widehat{u_0}(\kappa_j) \sum_{n=0}^{q-1} \sum_{m=0}^{q-1} \alpha^{2(j-n)m+p(j^2-n^2)} \\
		={}& \frac{1}{q} \,\widehat{u_0}(\kappa_j) \sum_{m=0}^{q-1} \alpha^0  + \frac{1}{q}\, \widehat{u_0}(\kappa_j) \sum_{\substack{n=0,\\n\neq j}}^{q-1} \alpha^{p(j^2-n^2)} \sum_{m=0}^{q-1} \alpha^{2(j-n)m} = \widehat{u_0}(\kappa_j).
	\end{align*}
	Similarly, 
	$$e^{\i \frac{\pi p}{q}j^2}\widehat{\psi}(-\kappa_j) = \widehat{u_0}(-\kappa_j) \roq{for all} j\in\ZZ.$$
Thus,
	\BES
		F(x;s) = \frac{1}{L} \sum_{j\in\ZZ}e^{\i (\kappa_js-\i \pi pj^2/q)}\left[\frac{ \left((\beta_0+\beta_1)\gamma-2\right)e^{\i \kappa_j x} + (\beta_1-\beta_0)\gamma e^{-\i \kappa_j x}}{(\beta_0+\beta_1)(\gamma-\gamma^{-1})} \right] \left[ \,\widehat{u_0}(\kappa_j)+\delta \, \widehat{u_0}(-\kappa_j)\,\right].
	\EES
Checking that the exponential factors match is a simple direct calculation:
$$\qxeq{\exp\left(\i\frac{p\:L^2\kappa_0^2}{4\pi q}-\i \kappa_j\frac{p\:L^2\kappa_0}{2\pi q}-\i \frac{\pi p}{q}j^2\right) = \exp\left( \frac{-\i p\:L^2\kappa_j^2}{4\pi q} \right),\\ j\in\ZZ,} $$
{which thus establishes Theorem~\ref{thm:ShiftRep.uGenvGen}.} \end{proof}

Specialising the formulas in Theorem~\ref{thm:ShiftRep.uGenvGen}, which provides a representation for the solution to the initial-boundary value problem \eq{eqn:ls_pp} at rational times \eq{rattime} in terms of certain shifts and reflections of the initial datum, to the two cases of odd and even denominator $q$, we deduce the following slightly simplified forms for the key functions \eqref{eqn:phipsi}. 
In the cases $p$ even, $q$ general, and $p$ odd, $q$ odd, we can, for each of the integers $n$ and $m$, write the term $\alpha^{2nm-pn^2}$ using a $q$\textsuperscript{th} root of unity.
Thus, equation~\eqref{eqn:phipsi} holds and there is no branching behaviour.
However, when $p$ and $q$ are both odd, we are nearer in spirit to Example~\ref{ex:uOdd.v1}. In this case, we first define 
\Eq{linebar}
$$\qeq{\overline{u_0}(x)=\left\{\begin{array}{lll} \gamma^{-\frac{1}{2}}u_0\left(x+\frac{1}{2}\:L\right),&0\leq x\leq \frac{1}{2}\:L,\\
  	\gamma^{\frac{1}{2}}u_0\left(x-\frac{1}{2}\:L\right),&\frac{1}{2}\:L< x\leq L,
  	\end{array}\right.
\\	\widetilde{u_0}(x)=\left\{\begin{array}{lll} \gamma^{\frac{1}{2}}u_0\left(x+\frac{1}{2}\:L\right),\hskip-5pt&0\leq x\leq \frac{1}{2}\:L,\\
  	\gamma^{-\frac{1}{2}}u_0\left(x-\frac{1}{2}\:L\right),&\frac{1}{2}\:L< x\leq L.
  	\end{array}\right.
}$$
and then set    
\begin{subequations}\label{phipsiodd}
\begin{equation}
\begin{split}
\phi(x)
			=\frac 1q\sum_{m=0}^{\ell-1} \gamma^{-m/q} {}&\overline{u_0}\left(x+\frac{Lm}{q} \right) \sum_{n=0}^{q-1} \alpha^{-2nm+p(q-1)n^2}\\
			&+ \frac 1q\sum_{m=\ell}^{q-1} \gamma^{1-m/q} \overline{u_0}\left(x+\frac{Lm}{q}-L\right) \sum_{n=0}^{q-1} \alpha^{-2nm+p(q-1)n^2},
			\end{split}
\end{equation}
\begin{equation}
\begin{split}
		\hskip-5pt\psi(x)
			=\frac 1q\sum_{m=0}^{\ell-1} \gamma^{m/q} {}&\widetilde{u_0}\left(x+\frac{Lm}{q} \right) \sum_{n=0}^{q-1} \alpha^{ 2nm+p(q-1)n^2}\\
			&+ \frac 1q\sum_{m=\ell}^{q-1} \gamma^{m/q-1} \widetilde{u_0}\left(x+\frac{Lm}{q}-L\right) \sum_{n=0}^{q-1} \alpha^{ 2nm+p(q-1)n^2}.
			\end{split}
			\end{equation}
			\end{subequations}
In both cases, once $\phi,\psi$ are defined appropriately, i.e.,~\eqref{eqn:phipsi} or \eq{phipsiodd}, we recover the solution at the indicated rational time by use of formula~\eqref{eqn:F=q.withsc} with $F(x;s)$ defined by equation~\eqref{eqn:Analytic.F.looks.like.q.gen}.


\section{The Talbot effect for general linear boundary conditions} \label{sec:genBCs}
In the preceding sections, we concentrated on the effect of the pseudoperiodic boundary conditions~\eqref{eqn:ls_pp}. 
In what follows we will seek to generalise our analysis to generic linear boundary conditions:
\begin{equation}\label{eqn:ls}
\eeq{  {\i u_t+u_{xx}=0, \qquad }  u(0,x)=u_0(x), \\
  \beta_{11}u_x(t,L)+\beta_{12}u(t,L)+\beta_{13}u_x(t,0)+\beta_{14}u(t,0)=f_1(t),\\
   \beta_{22}u(t,L)+\beta_{23}u_x(t,0)+\beta_{24}u(t,0)=f_2(t), 
}\qquad (t,x) \in[0,T]\times [0,L],
\end{equation}
where $\beta_{ij}$ are (possibly complex) constants.

Finding the solution via the UTM proceeds in exactly the same way as in Section~\ref{sec:DeriveUTMRep} except now the \emph{discriminant} of the problem takes the form
\begin{equation}\label{eqn:Delta_genBC}
\begin{split}
\Delta(\kappa )={}2\i\left[\,(\beta_{11}\beta_{22}-\beta_{14}\beta_{23}+\beta_{13}\beta_{24})\kappa\right.&+(\beta_{13}\beta_{22}-\beta_{12}\beta_{23}+\beta_{11}\beta_{24})\kappa\cos(\kappa L)\\
&\left.+(\beta_{12}\beta_{24}-\beta_{14}\beta_{22}+\beta_{11}\beta_{23}\kappa^2)\sin(\kappa L) \,\right].
\end{split}
\end{equation}
We define 
\Eq{ZDgen}
$$\ZD = \{\kappa\in\CC:\Delta(\kappa)=0\} = \{\kappa_j,-\kappa_j:j\in\ZZ\},$$
and decompose the discriminant locus into
\Eq{ZDpm_genlin}
$$\ZD = \ZDp \cup \ZDm \cup Z_{\Delta}^{(0,+)}\cup Z_{\Delta}^{(0,-)} $$
where
\begin{align*}
\ZDp=&\{\kappa\in\ZD:\Re(\kappa)>0\},\qquad Z_\Delta^{(0,+)} =\{\kappa\in\ZD:\Re(\kappa)=0, \ \ \Im(\kappa)>0\},\\
\ZDm=&\{\kappa\in\ZD:\Re(\kappa)<0\},\qquad Z_\Delta^{(0,-)} =\{\kappa\in\ZD:\Re(\kappa)=0, \ \ \Im(\kappa)<0\}.\\
\end{align*}

Since the discriminant equation is transcendental, we cannot find analytic expressions for $\kappa_j$ as was done in~\eqref{eqn:Lambdaj.gen}.

For the problem to be well-posed, there can only be finitely many $ \i \kappa_j^2 $ with $\kappa _j \in \ZD$ which have real part $< 0$; if present, these produce a finite number of exponentially growing modes that eventually dominate the solution, and hence the initial value problem is unstable, but not ill-posed, which requires infinitely many such modes.
For an example of exponential growth see Figure~\ref{fig:gen_complexeval_r}.
On the other hand, if all $ \i \kappa_j^2 $ have real part $> 0$, then the solutions decay to zero in an unexpected and unusual manner; see Figure~\ref{fig:genbc_energyleak_r} and the ensuing discussion.
The observed behavior suggests that a form revival may still arise in which the solution at a rational time is a finite linear combination of appropriately scaled translates and reflections of the initial datum, so that the solution decays to zero as time increases, but, in contrast to dissipative systems like the heat equation, does not smooth out. 
This new phenomenon clearly warrants further investigation and rigorous justification.



We note the following asymptotic relations for $\kappa_j$.
If $\beta_{11}\beta_{23}\neq0$, then
\begin{subequations}\label{eqns:asym_lambda}
\begin{equation}\label{eqn:asym_lambda_nonzero}
\kappa_j\sim \frac{j\:\pi}{L},\qquad \hbox{as \quad $j \to \infty $},
\end{equation}
whereas, if $\beta_{11}\beta_{23}=0$, then
\begin{equation}\label{eqn:alt_asym_zero}
\kappa_j\sim \frac{1}{L}\left(2\:j\:\pi\pm\arccos\left(\frac{\beta_{14}\beta_{23}-\beta_{11}\beta_{22}-\beta_{13}\beta_{24}}{\beta_{13}\beta_{22}-\beta_{12}\beta_{23}+\beta_{11}\beta_{24}}\right) \right),\qquad \hbox{as \quad $j \to \infty $}.
\end{equation}
\end{subequations}

Finding a UTM solution representation proceeds in the same manner as above. We set
\begin{equation}
\begin{split}
	\zeta^+(t,\kappa ) ={}& \left( 2\i\kappa\beta_{11}\beta_{22}+e^{\i\kappa L}\left(\beta_{12}\beta_{24}-\beta_{14}\beta_{22}+\i\kap(\beta_{13}\beta_{22}-\beta_{12}\beta_{23}+\beta_{11}\beta_{24})+\kappa^2\beta_{11}\beta_{23}\right)\right) \widehat{u_0}(\kappa)\\
	&+ \left(\beta_{14}\beta_{22}-\beta_{12}\beta_{24}+\i\kap(\beta_{12}\beta_{23}-\beta_{13}\beta_{22}+\beta_{11}\beta_{24})+\kappa^2\beta_{11}\beta_{23}\right) e^{-\i\kappa L} \widehat{u_0}(-\kappa)\\
	&+2\kappa\left[\left(-e^{-\i\kappa L}\beta_{22}+\i\kappa\beta_{23}-\beta_{24}\right)f_1(t)
	+\left(e^{-\i\kappa L}(\beta_{12}-\i\kappa\beta_{11})+\beta_{14}-\i\kappa\beta_{13}\right)f_2(t)\right].
\end{split}
\end{equation}
The resulting solution formula has the form
\begin{align}
u(t,x)={}&\frac{1}{2\pii} \sum_{\mu \in \ZD}\int_{S(\mu,r)} e^{\i\kappa\:x-\i\kappa ^2t}\>\frac{\zeta^+(\kap,t)}{\Delta(\kappa)}\ud\kappa = \i \sum_{j\in\ZZ} e^{\i \kappa_j x-\i\kappa_j^2t} \>\frac{\zeta^+(\kappa_j,t)}{\Delta'(\kappa_j)},
\end{align}
where the last equality holds if and only if all zeros of $\Delta$ are simple, which, for simplicity of presentation we will assume.
Let us now specialize to general homogeneous boundary conditions, i.e., take $f_0(t) \equiv f_1(t) \equiv 0$ in~\eqref{eqn:ls}.
The boundary value problem can still be solved by a generalised eigenfunction series of the form~\eqref{eqn:series}, except now the eigenfunctions $X_j(x)$ satisfy
\Eq{genefneq}
$$\eeq{X_j''(x)+\kappa_j^2X_j(x)=0,\\
\beta_{11}X_j'(L)+\beta_{12}X_j(L)+\beta_{13}X_j'(0)+\beta_{14}X_j(0)=0,\\
\beta_{22}X_j(L)+\beta_{23}X_j'(0)+\beta_{24}X_j(0)=0.
}$$
For any values of $\beta_{11},\beta_{23}$, and $\kappa_j \in \ZDp\cup Z_\Delta^{(0,+)}$, the corresponding eigenfunctions take the form
\begin{align}\label{genefn}
X_j(x)={}&e^{\i\kappa_j x}+b_1(\kappa_j) e^{-\i\kappa_jx},
\end{align}
where
$$b_1(\kappa_j)=-\frac{e^{\i\kappa_jL}(\i\kappa_j\beta_{11}+\beta_{12})+\i\kappa_j\beta_{13}+\beta_{14}}{e^{-\i\kappa_jL}(\beta_{12}-\i\kappa_j\beta_{11})+\beta_{14}-\i\kappa_j\beta_{13}} = -\frac{e^{\i\kappa_jL}\beta_{22}+\i\kappa_j\beta_{23}+\beta_{24}}{e^{-\i\kappa_jL}\beta_{22}+\beta_{24}-\i\kappa_j\beta_{23}}.
$$
As before, we form a biorthogonal system by pairing the eigenfunctions $X_j$  with the eigenfunctions
\begin{align}
Y_j(x)={}&e^{\i\overline{\kappa_j} x}+\overline{b_2(\kappa_j)}e^{-\i\overline{\kappa_j}x},
\end{align}
 of the adjoint problem, where
\begin{align*}
\overline{b_2(\kappa_j)}={}&\frac{e^{-\i\kappa_jL}(-\i\kappa_j\beta_{11}\beta_{24}-\beta_{14}\beta_{22}+\beta_{12}\beta_{24})-\i\kappa_j\beta_{11}\beta_{22}}{e^{\i\kappa_jL}(-\i\kappa_j\beta_{11}\beta_{24}+\beta_{14}\beta_{22}-\beta_{12}\beta_{24})-\i\kappa_j \beta_{11}\beta_{22}}\\
={}&\frac{e^{-\i\kappa_jL}(-\i\kappa_j\beta_{11}\beta_{23}+\beta_{12}\beta_{23}-\beta_{13}\beta_{22})-\beta_{11}\beta_{22}}{e^{\i\kappa_jL}(-\i\kappa_j\beta_{11}\beta_{23}-\beta_{12}\beta_{23}+\beta_{13}\beta_{22})+ \beta_{11}\beta_{22}}.
\end{align*}
That is,
\Eq{XjYkgen}
$$\langle X_j,Y_k\rangle=\left\{\begin{array}{lll}0,&&j\neq k,\\ \tau _j,&&j=k.\end{array}\right.,$$
where
\Eq{tauj}
$$\tau _j=1+b_1(\kappa_j)\,\overline{b_2(\kappa_j)}+\displaystyle \frac{\sin(\kappa L_j)}{\kappa_jL}\left(b_1(\kappa_j)\,e^{-\i \kappa L_j}+\overline{b_2(\kappa_j)}\,e^{\i \kappa L_j}\right).$$


Evaluating the eigenfunction series~\eqref{eqn:series} at $t=0$ and taking the inner product of both sides with $Y_k$ we find
$$c_j=\frac{1}{\tau _jL}\bbk{\widehat{u_0}(\kappa_j)+\overline{b_2(\kappa_j)}\,\widehat{u_0}(-\kappa_j)}.$$
Hence,
\begin{equation}\label{eqn:seressoln_genBC}
u(t,x)=\frac{1}{L}\sum_{j\in\ZZ}\frac{e^{-\i\kappa_j^2t}\left[e^{\i\kappa_jx}+b_1(\kappa_j)\,e^{-\i\kappa_jx}\right]}{\tau _j}\left[\,\widehat{u_0}(\kappa_j)+\overline{b_2(\kappa_j)}\widehat{u_0}(-\kappa_j)\,\right].
\end{equation}



Although the form of the solution for general boundary conditions~\eqref{eqn:seressoln_genBC} is similar to that associated with pseudoperiodic boundary conditions~\eqref{eqn:seriessoln}, the analysis we complete in Section~\ref{sec:analysis} no longer applies when the coefficients in~\eqref{eqn:Asymptotics.defn.F.gen} are $j$ dependent. 

Let us now present a selection of boundary conditions illustrating different phenomena we have so far observed. 
A complete, rigorous classification and proof of the various possibilities remains to be done.

\subsection{Robin boundary conditions}
We begin with the simple case of Dirichlet boundary data at $x=0$ and Robin data at $x=L$:
\Eq{RD1}
$$-2\,u_x(t,1)+u(t,1)=0 \qquad  u(t,0)=0,$$
so that
$$\qeq{\beta_{11}=-2,\\\beta_{12}=1,\\\beta_{13}=\beta_{14}=\beta_{22}=\beta_{23}=0,\\\beta_{24}=1.}$$
In this case, all the discriminant roots $\kappa_j$ are purely real, as illustrated Figure~\ref{fig:evals_real}.
\begin{figure}[h!]
\centering
    \includegraphics[width=.5\linewidth]{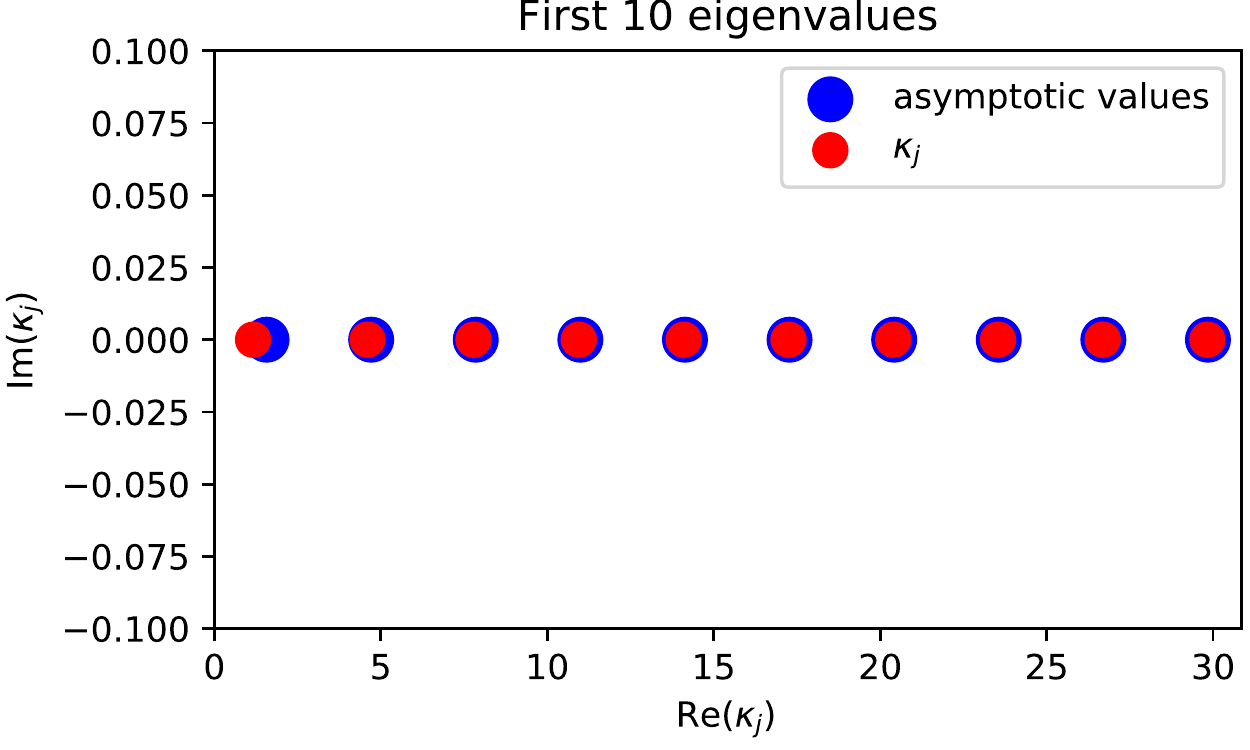}
    \caption{The first ten $\kappa_j$ found numerically in red and the asymptotic values they approach~\eqref{eqn:alt_asym_zero} in blue for $\beta_{11}=-2$, $\beta_{12}=1$, $\beta_{13}=0$, $\beta_{14}=0$, $\beta_{22}=0$, $\beta_{23}=0$, $\beta_{24}=1$, and $L=1$.}\label{fig:evals_real}
\end{figure}
Note that in the case of Robin boundary data at both $x=0$ and $x=L$, i.e., when $\beta_{13}=\beta_{14}=\beta_{22}=0$, if the remaining coefficients are all real then the problem is self-adjoint.
In Figures~\ref{fig:RobinDirichlet_real_r} and~\ref{fig:RobinDirichlet_real_irr} we can see that although the solution is no longer simply a number of shifted copies of the initial condition we do see something similar to revivals.
\begin{figure}[h!]
  \centering
  \begin{minipage}[b]{.32\linewidth}
    \centering
    \includegraphics[width=\linewidth]{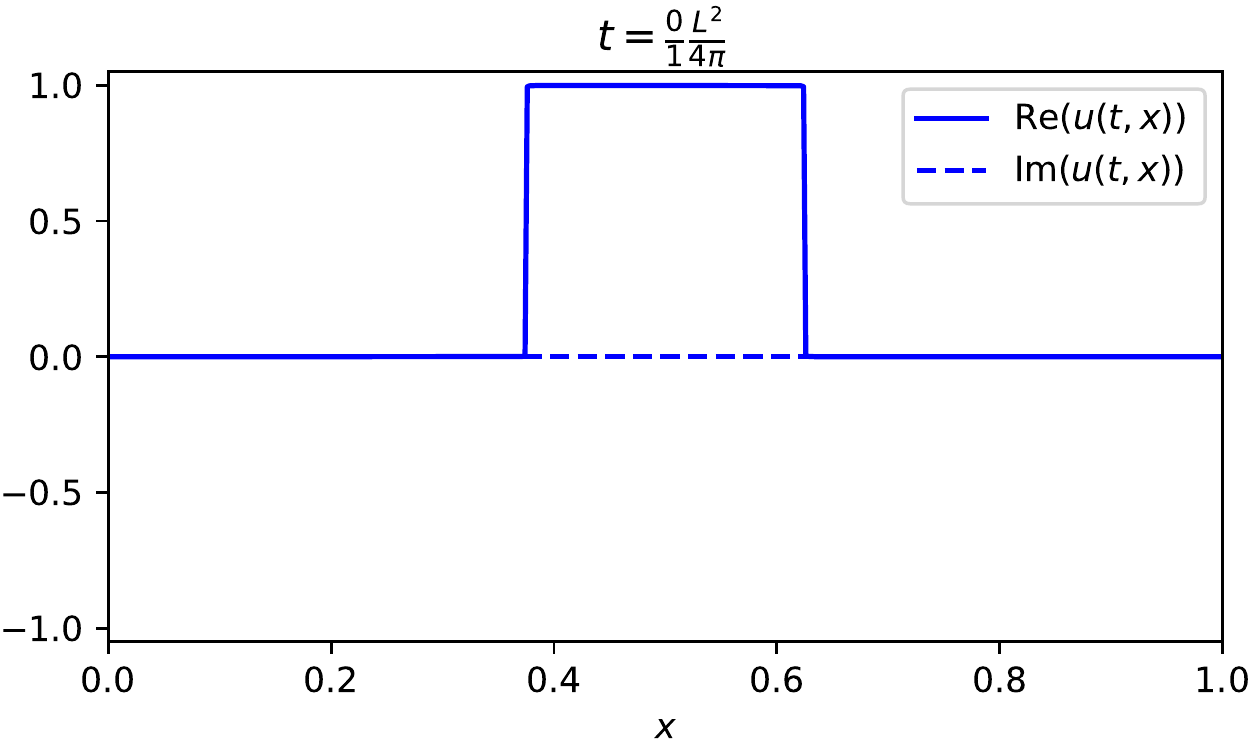}
    \subcaption{$t=0$}
  \end{minipage}
  \hfill
  \begin{minipage}[b]{.32\linewidth}
    \centering
    \includegraphics[width=\linewidth]{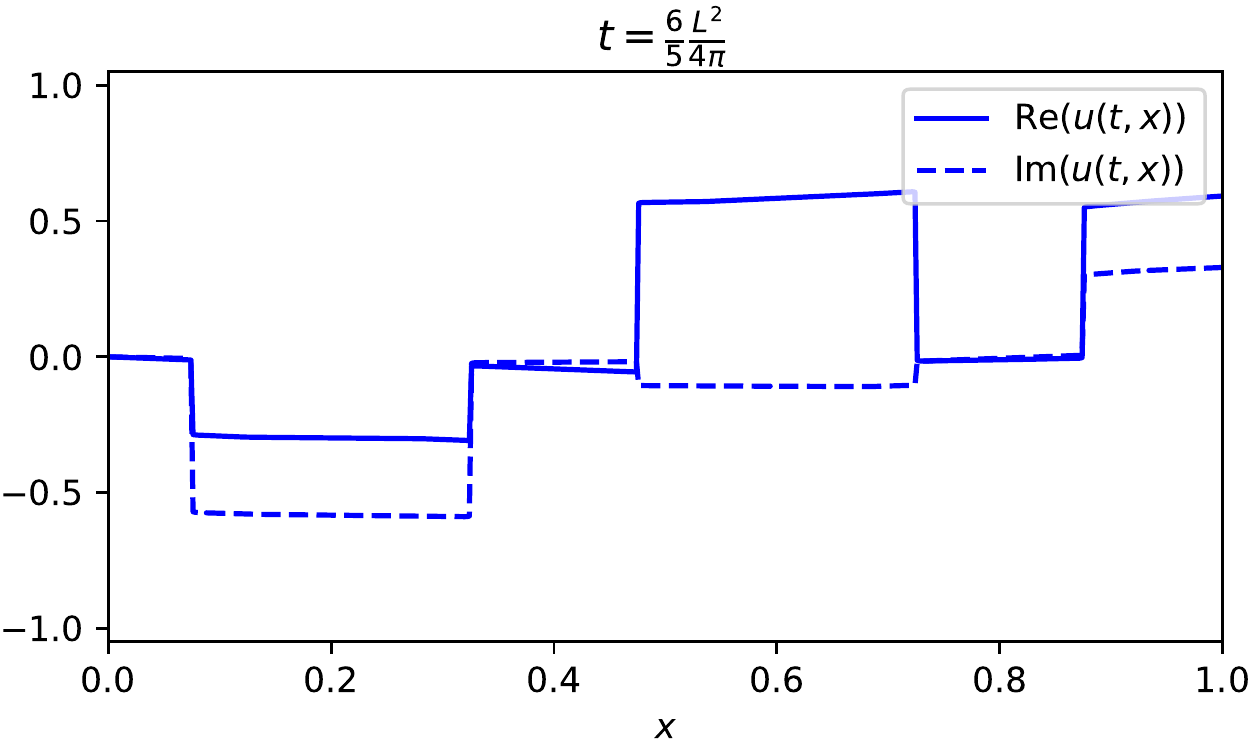}
    \subcaption{$\udsty t=\frac{6}{5}\frac{L^2}{4\pii}\approx0.09$}
  \end{minipage}
  \hfill
  \begin{minipage}[b]{.32\linewidth}
    \centering
    \includegraphics[width=\linewidth]{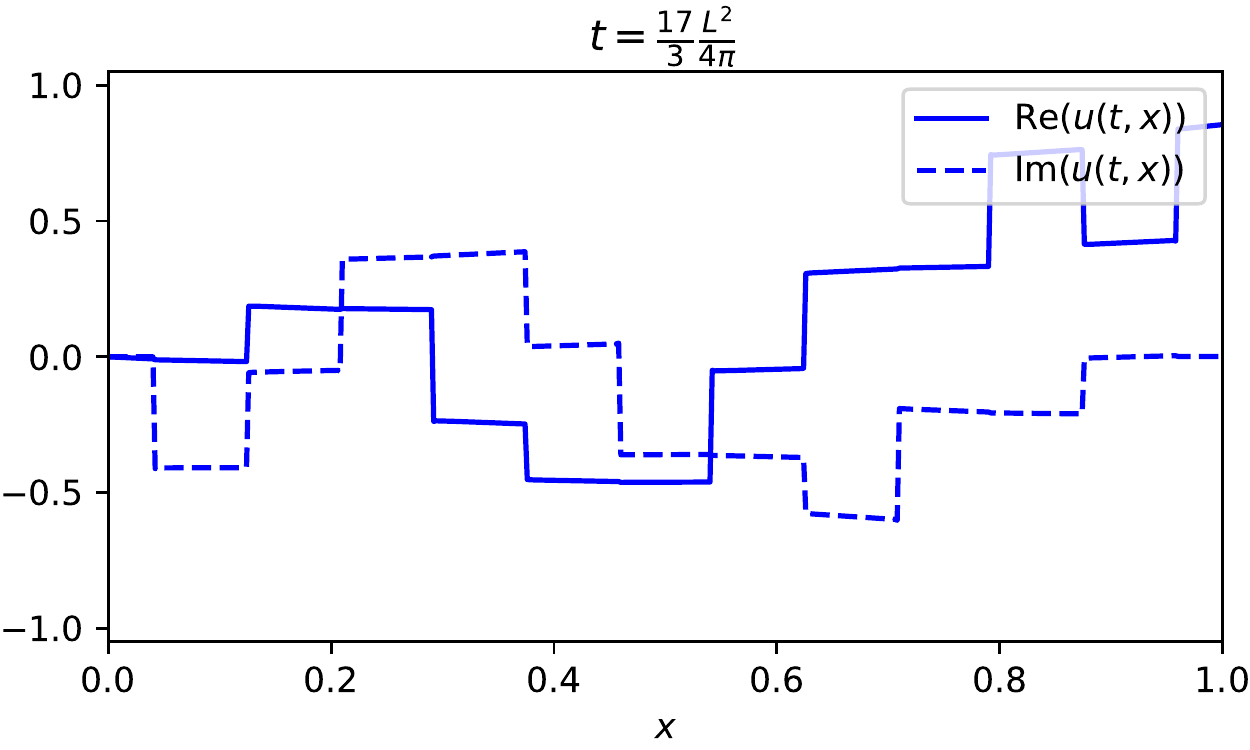}
    \subcaption{$\udsty t=\frac{17}{3}\frac{L^2}{4\pii}\approx0.45$}
  \end{minipage}
  \\
  \vspace{2ex}
  \begin{minipage}[b]{.32\linewidth}
    \centering
    \includegraphics[width=\linewidth]{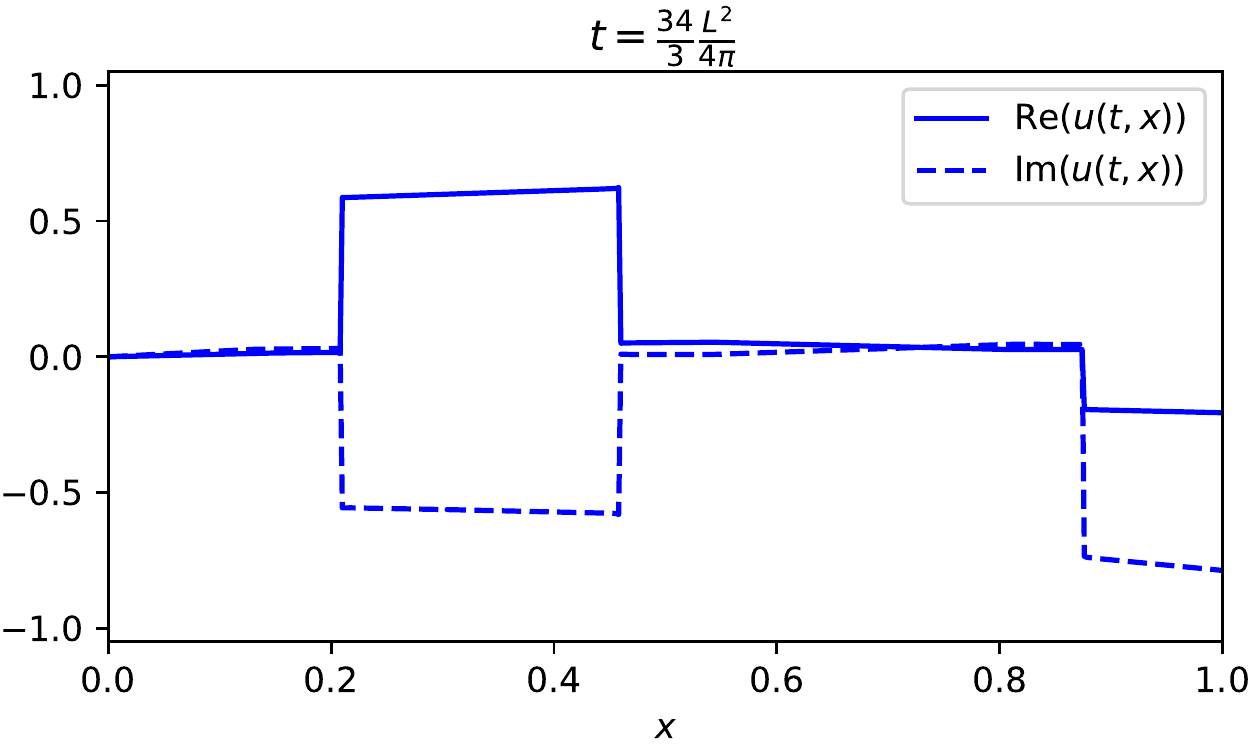}
    \subcaption{$\udsty t=\frac{34}{3}\frac{L^2}{4\pii}\approx0.9$}
      \end{minipage}
  \hfill
  \begin{minipage}[b]{.32\linewidth}
    \centering
    \includegraphics[width=\linewidth]{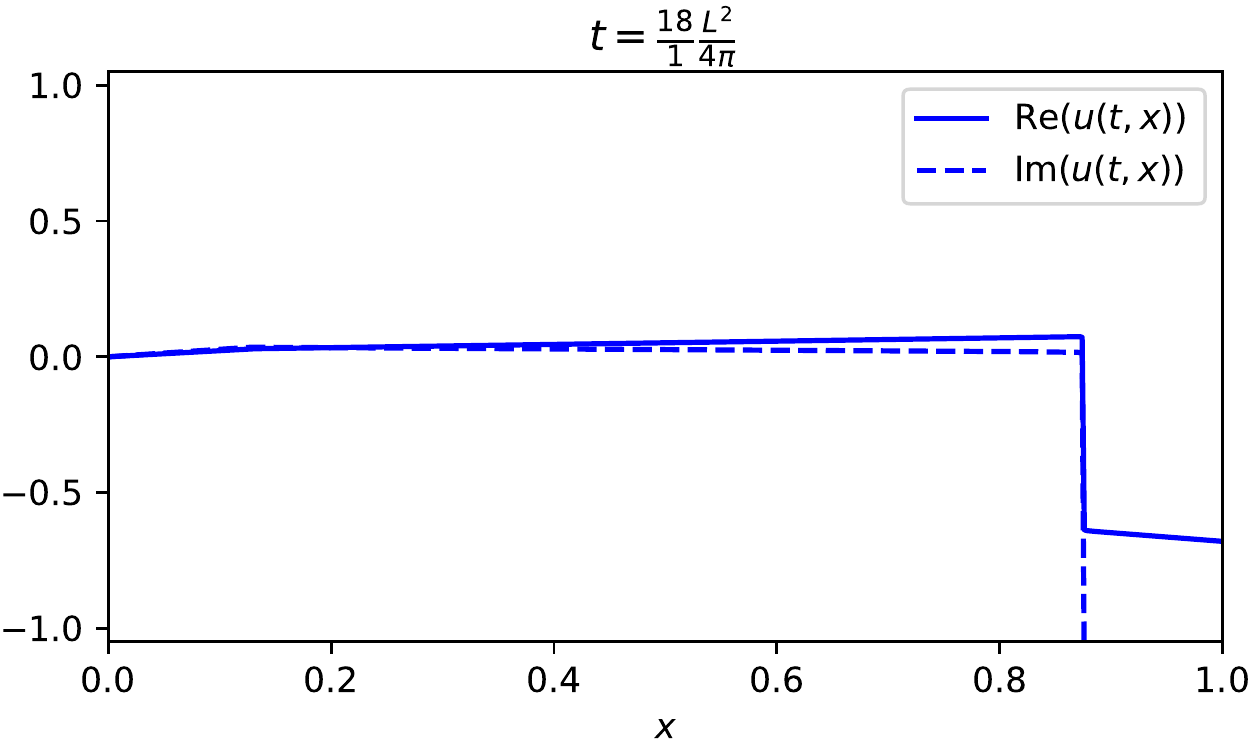}
    \subcaption{$\udsty t=\frac{18}{1}\frac{L^2}{4\pii}\approx1.44$}
  \end{minipage}
  \hfill
  \begin{minipage}[b]{.32\linewidth}
    \centering
    \includegraphics[width=\linewidth]{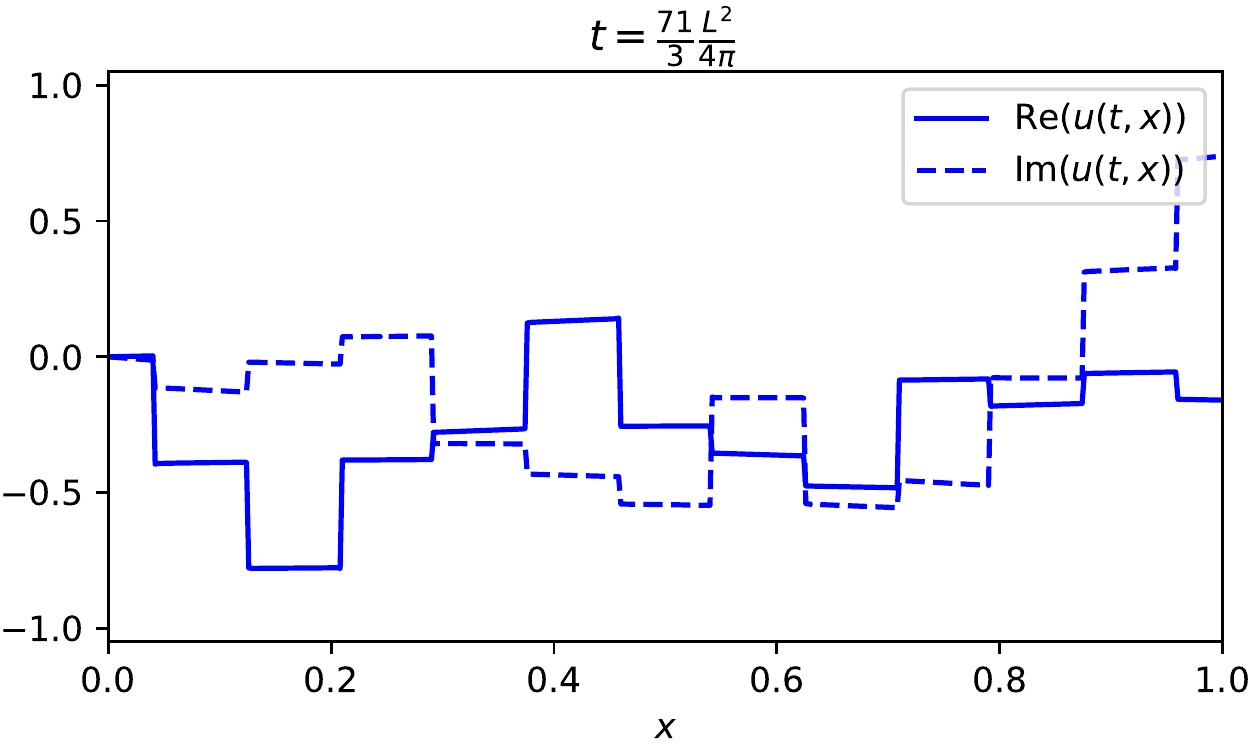}
    \subcaption{$\udsty t=\frac{71}{3}\frac{L^2}{4\pii}\approx1.89$}
  \end{minipage}
  \caption{
    The solution of the linear Schr\"{o}dinger equation with linear boundary conditions $\beta_{11}=-2$, $\beta_{12}=1$, $\beta_{13}=0$, $\beta_{14}=0$, $\beta_{22}=0$, $\beta_{23}=0$, and $\beta_{24}=1$ on $[0,1]$ and box initial datum evaluated at ``rational" times which are commensurate with $L^2/(4\pi)$.
  }
  \label{fig:RobinDirichlet_real_r}
\end{figure}

\begin{figure}[h!]
  \centering
  \begin{minipage}[b]{.32\linewidth}
    \centering
    \includegraphics[width=\linewidth]{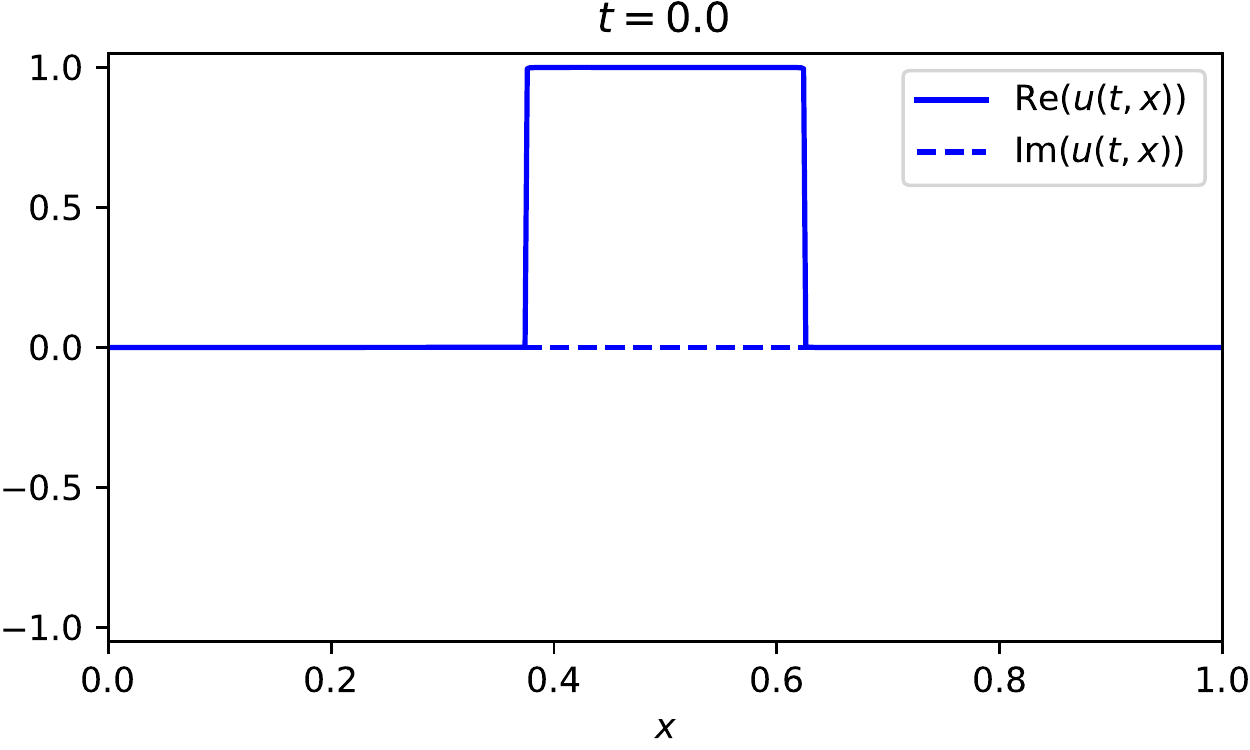}
    \subcaption{$t=0$}
  \end{minipage}
  \hfill
  \begin{minipage}[b]{.32\linewidth}
    \centering
    \includegraphics[width=\linewidth]{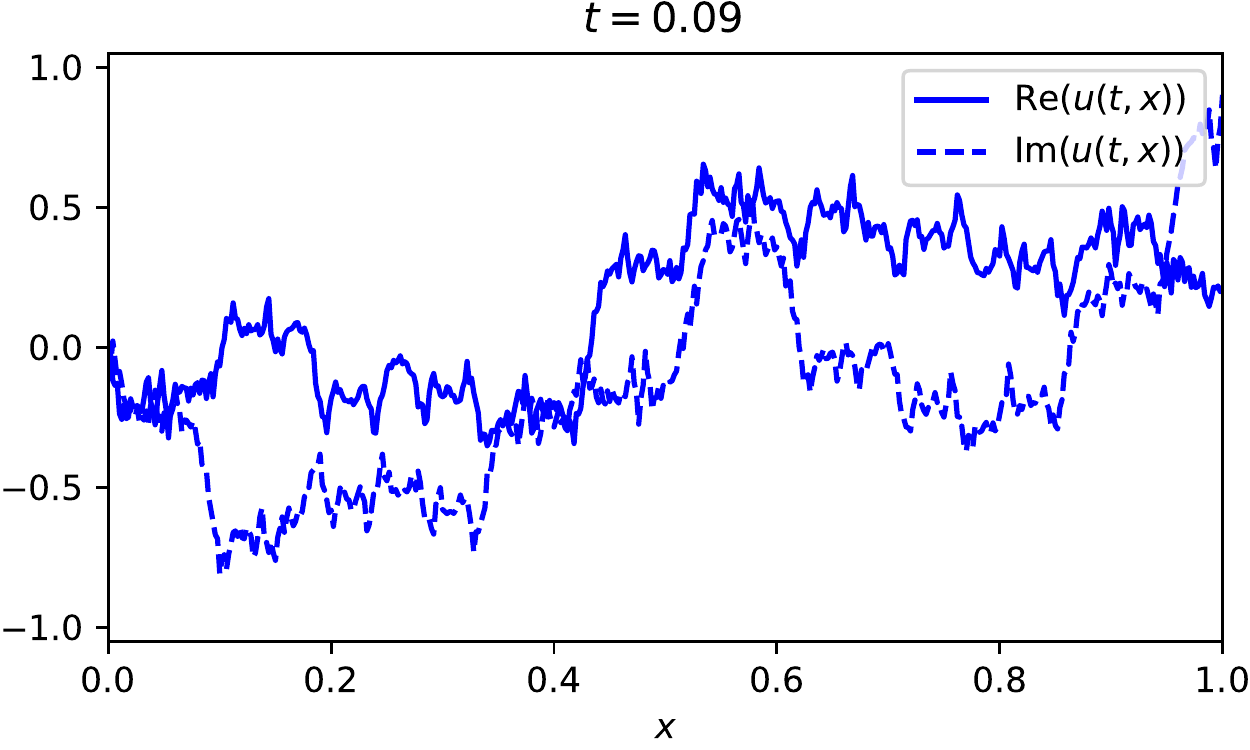}
    \subcaption{$\udsty t=0.09$}
  \end{minipage}
  \hfill
  \begin{minipage}[b]{.32\linewidth}
    \centering
    \includegraphics[width=\linewidth]{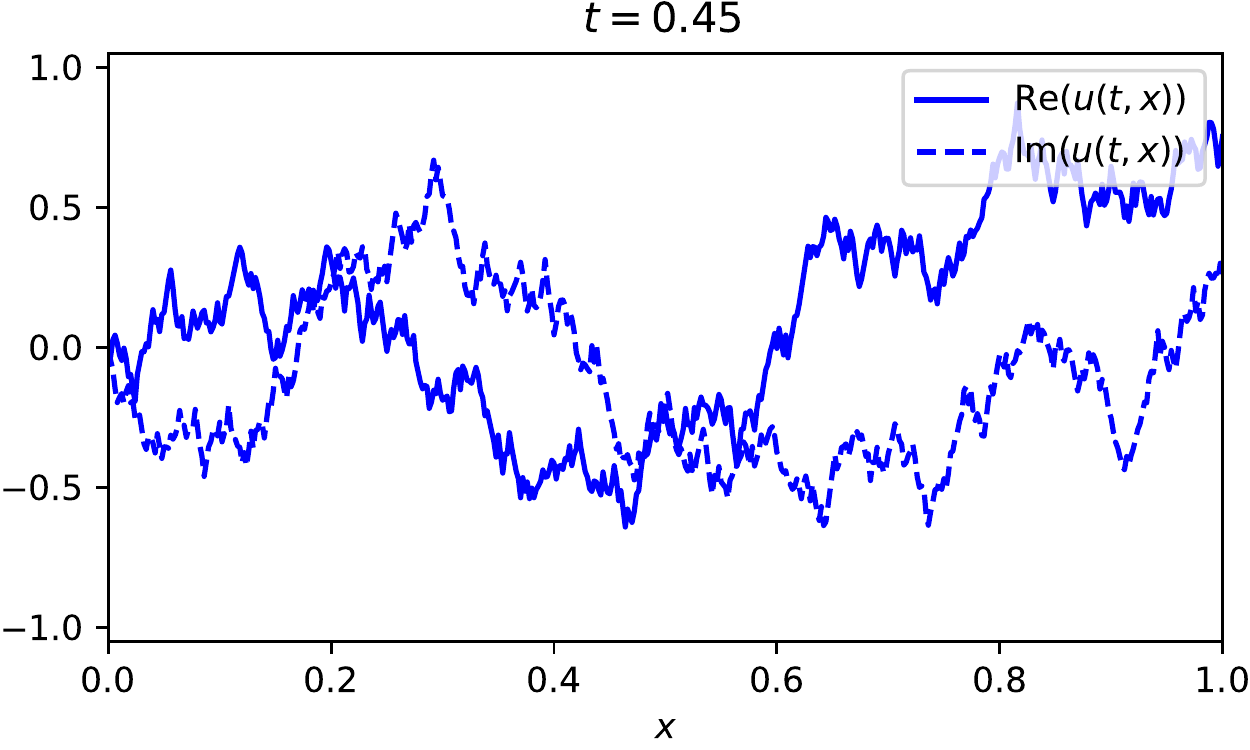}
    \subcaption{$\udsty t=0.45$}
  \end{minipage}
  \\
  \vspace{2ex}
  \begin{minipage}[b]{.32\linewidth}
    \centering
    \includegraphics[width=\linewidth]{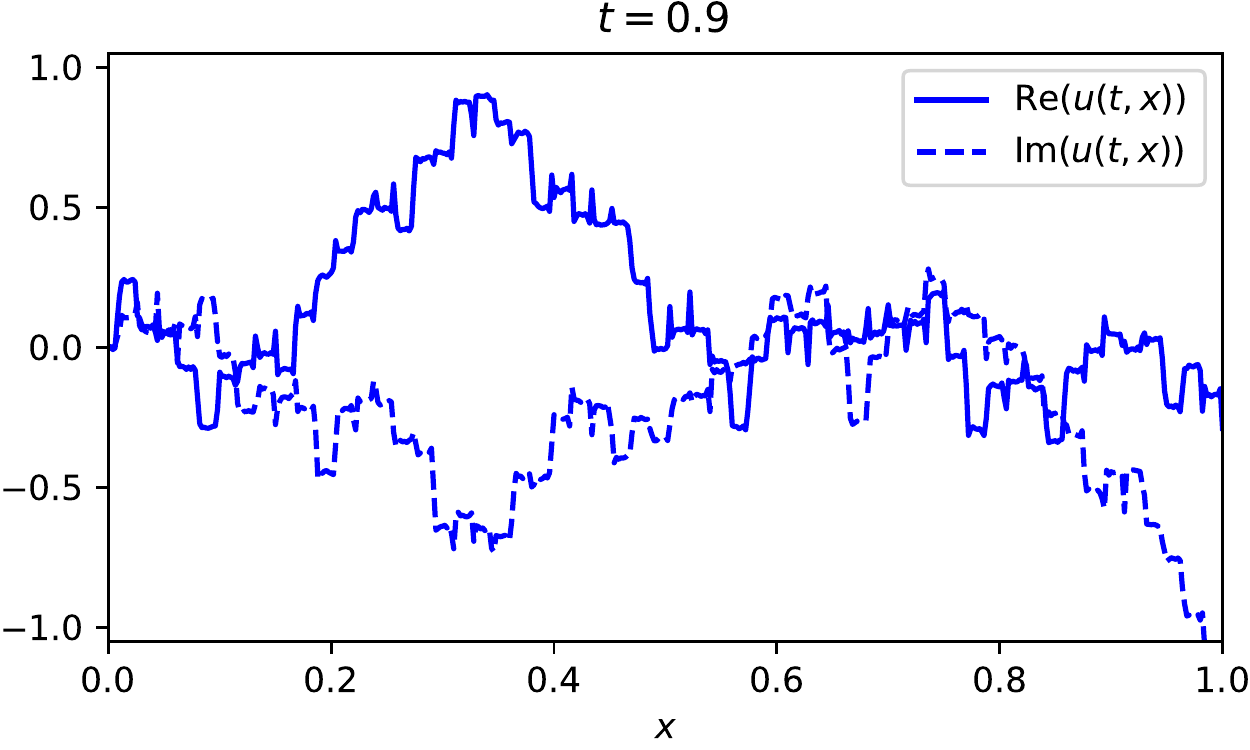}
    \subcaption{$\udsty t=0.9$}
      \end{minipage}
  \hfill
  \begin{minipage}[b]{.32\linewidth}
    \centering
    \includegraphics[width=\linewidth]{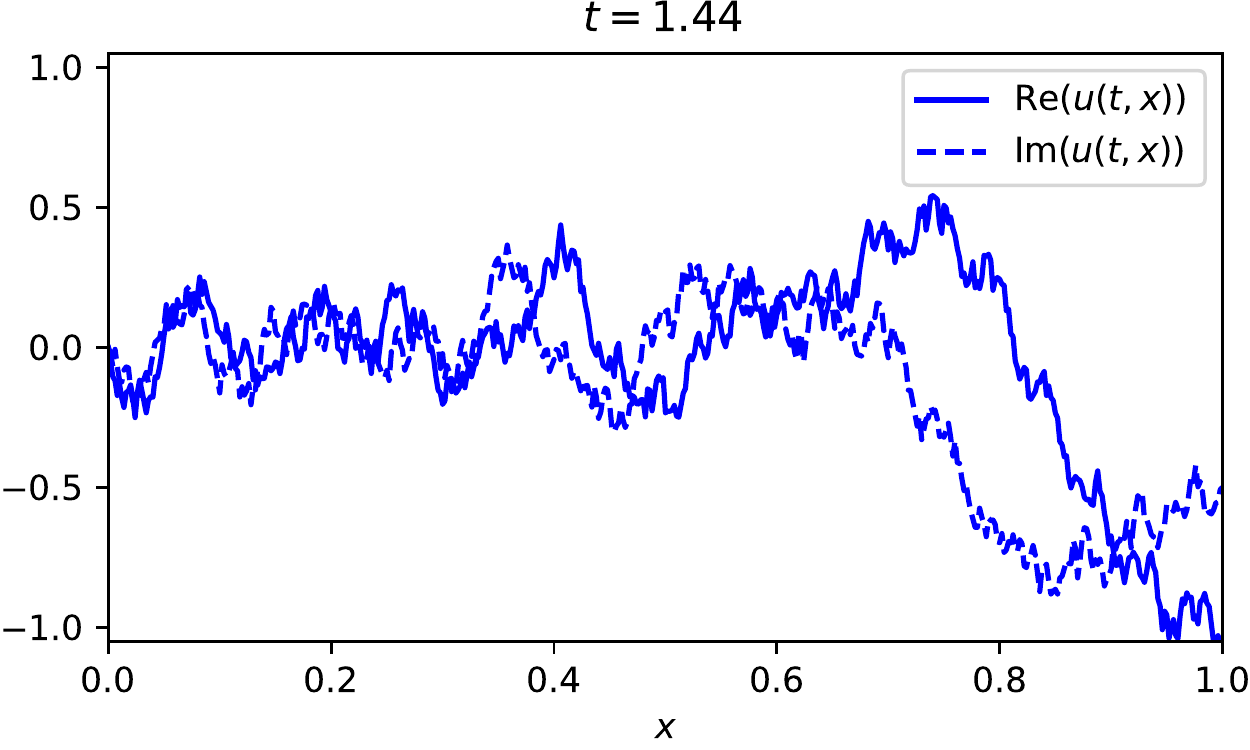}
    \subcaption{$\udsty t=1.44$}
  \end{minipage}
  \hfill
  \begin{minipage}[b]{.32\linewidth}
    \centering
    \includegraphics[width=\linewidth]{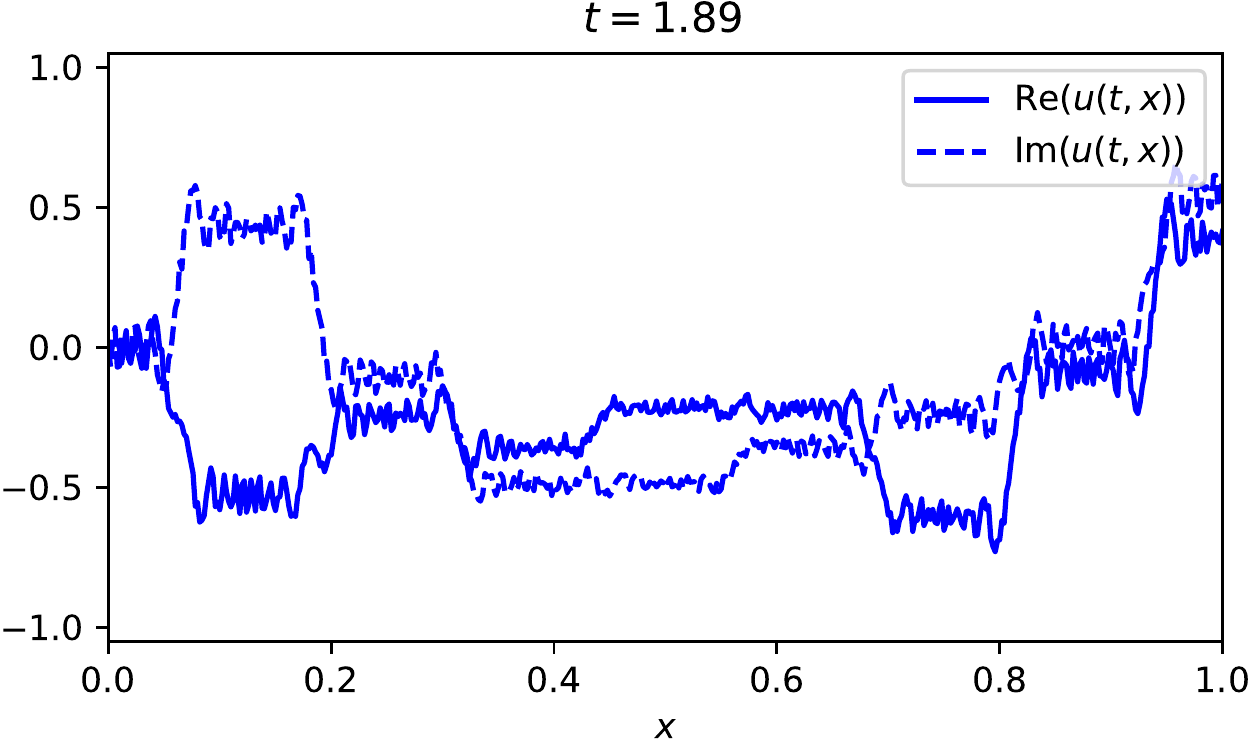}
    \subcaption{$\udsty t=1.89$}
  \end{minipage}
  \caption{
    The solution of the linear Schr\"{o}dinger equation with linear boundary conditions $\beta_{11}=-2$, $\beta_{12}=1$, $\beta_{13}=0$, $\beta_{14}=0$, $\beta_{22}=0$, $\beta_{23}=0$, and $\beta_{24}=1$ on $[0,1]$ and box initial datum evaluated at ``irrational" times which are not commensurate with $L^2/(4\pi)$.
  }
  \label{fig:RobinDirichlet_real_irr}
\end{figure}

A similar set of Robin plus Dirichlet boundary conditions 
\Eq{RD2}
$$-.7\,u_x(t,1)+u(t,1)=0 \qquad  u(t,0)=0,$$
corresponding to
$$\qeq{\beta_{11}=-.7,\\\beta_{12}=1,\\\beta_{13}=\beta_{14}=\beta_{22}=\beta_{23}=0,\\\beta_{24}=1,}$$
produces one purely imaginary $\kappa_j$, as illustrated in Figure~\ref{fig:evals_imag}. Here the solution is qualitatively similar as seen in Figure~\ref{fig:RobinDirichlet_imag_r}.
\begin{figure}[h!]
\centering
    \includegraphics[width=.5\linewidth]{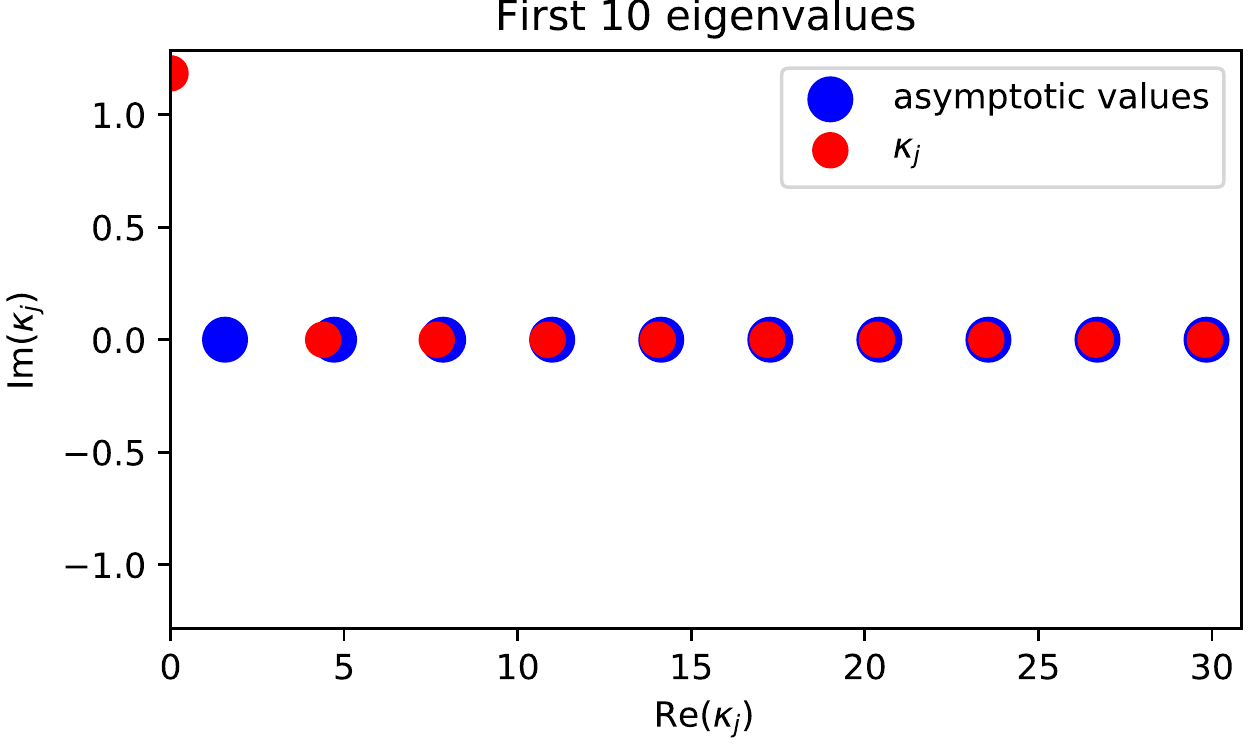}
    \caption{The first ten $\kappa_j$ found numerically in red and the asymptotic values they approach~\eqref{eqn:alt_asym_zero} in blue for $\beta_{11}=-.7$, \ $\beta_{12}=1$, \ $\beta_{13}=0$, \ $\beta_{14}=0$, \ $\beta_{22}=0$, \ $\beta_{23}=0$,  \ $\beta_{24}=1$, and $L=1$.}\label{fig:evals_imag}
\end{figure}

\begin{figure}[h!]
  \centering
  \begin{minipage}[b]{.32\linewidth}
    \centering
    \includegraphics[width=\linewidth]{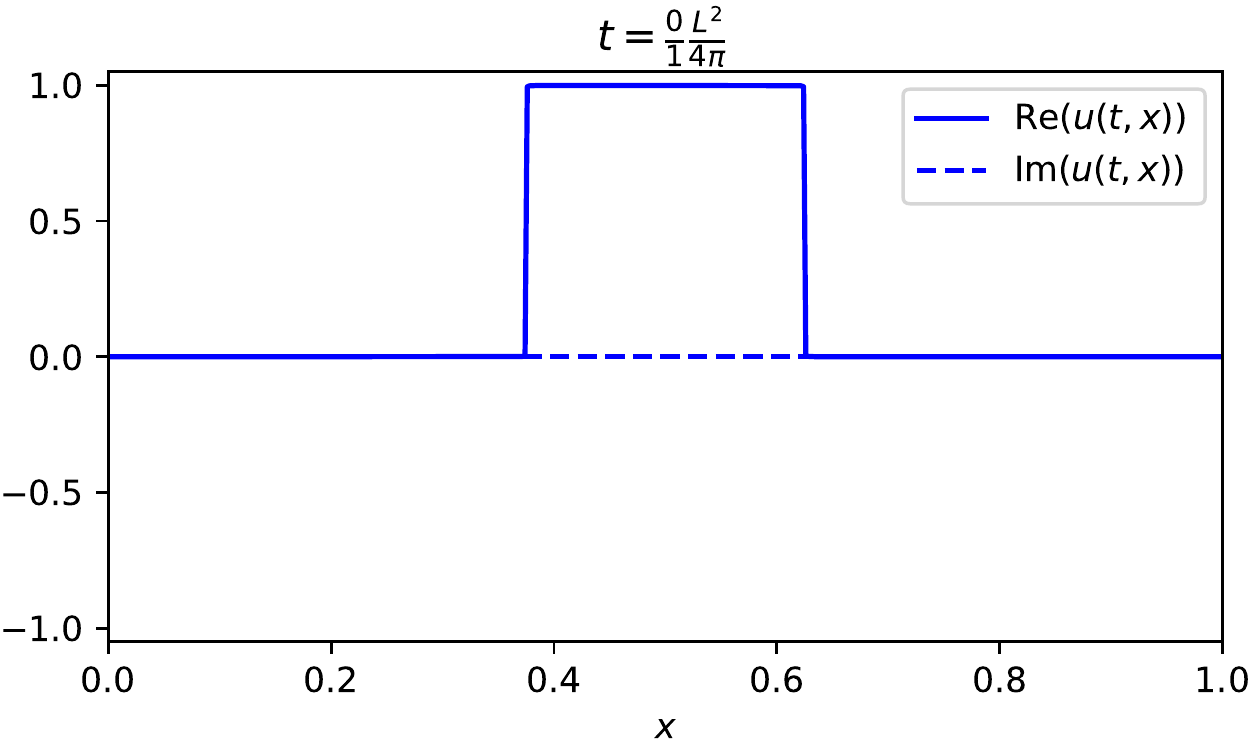}
    \subcaption{$t=0$}
  \end{minipage}
  \hfill
  \begin{minipage}[b]{.32\linewidth}
    \centering
    \includegraphics[width=\linewidth]{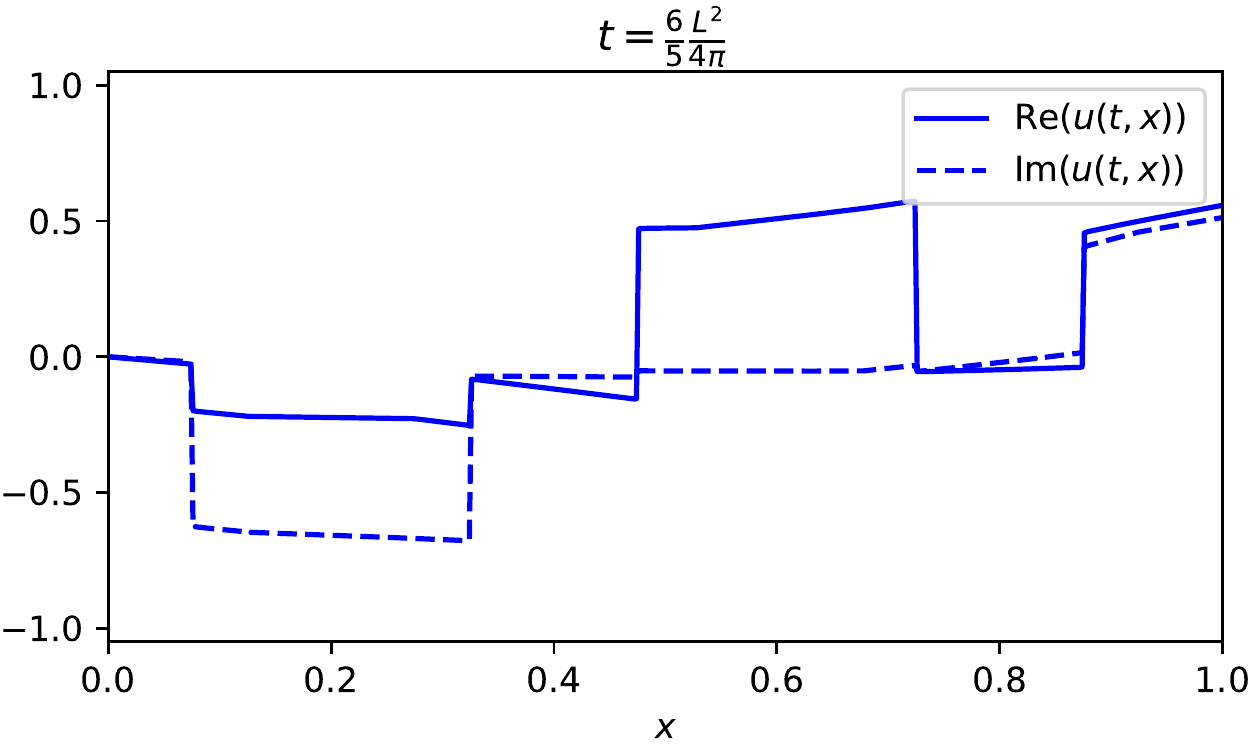}
    \subcaption{$\udsty t=\frac{6}{5}\frac{L^2}{4\pii}\approx0.09$}
  \end{minipage}
  \hfill
  \begin{minipage}[b]{.32\linewidth}
    \centering
    \includegraphics[width=\linewidth]{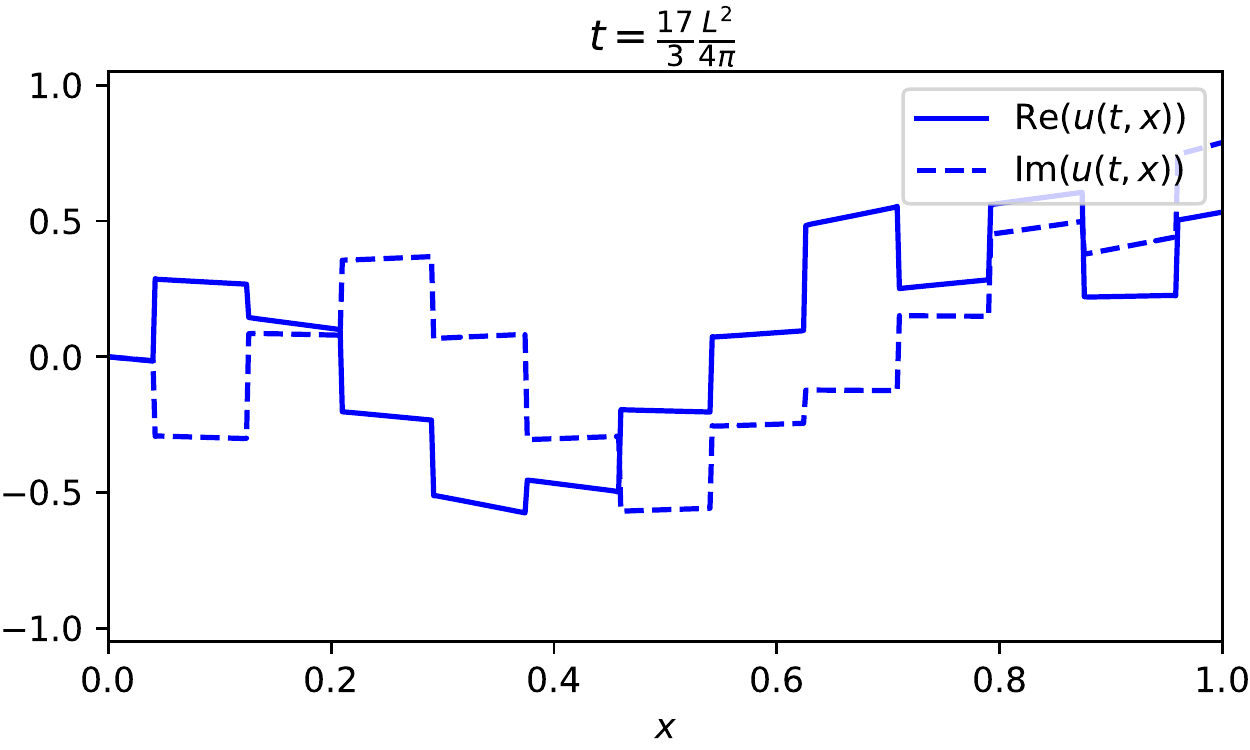}
    \subcaption{$\udsty t=\frac{17}{3}\frac{L^2}{4\pii}\approx0.45$}
  \end{minipage}
  \\
  \vspace{2ex}
  \begin{minipage}[b]{.32\linewidth}
    \centering
    \includegraphics[width=\linewidth]{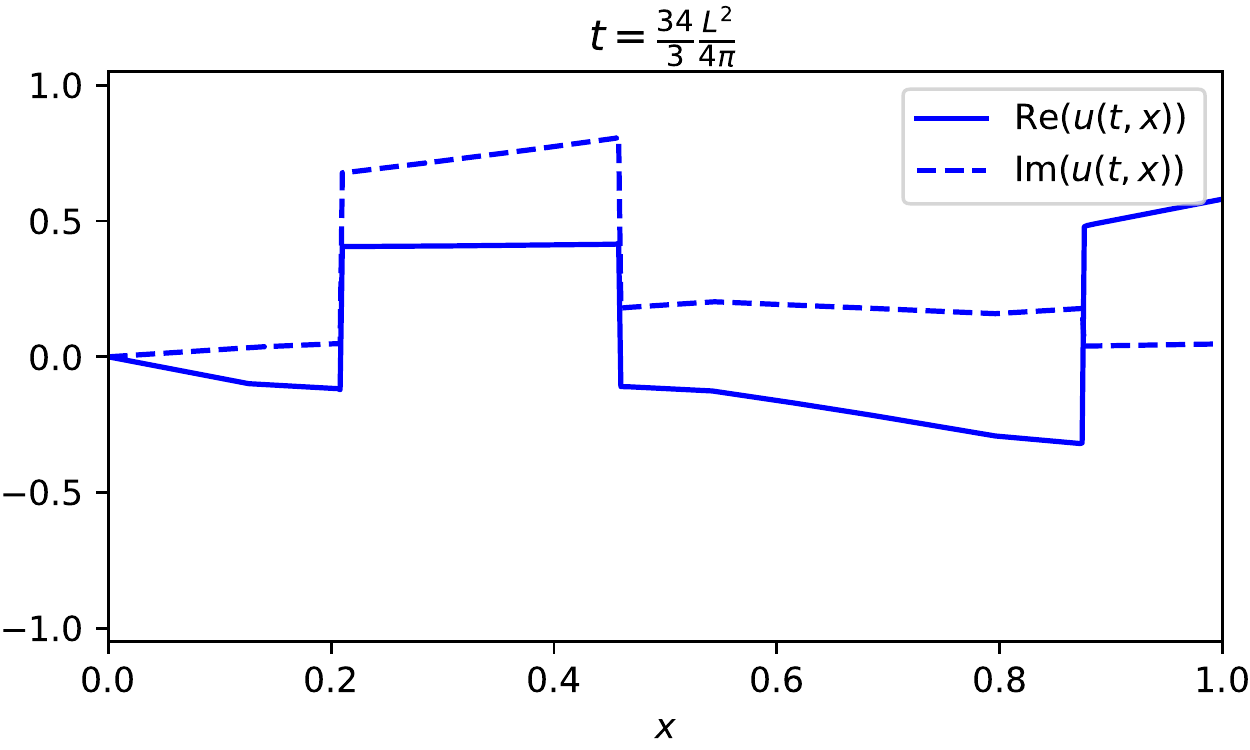}
    \subcaption{$\udsty t=\frac{34}{3}\frac{L^2}{4\pii}\approx0.9$}
      \end{minipage}
  \hfill
  \begin{minipage}[b]{.32\linewidth}
    \centering
    \includegraphics[width=\linewidth]{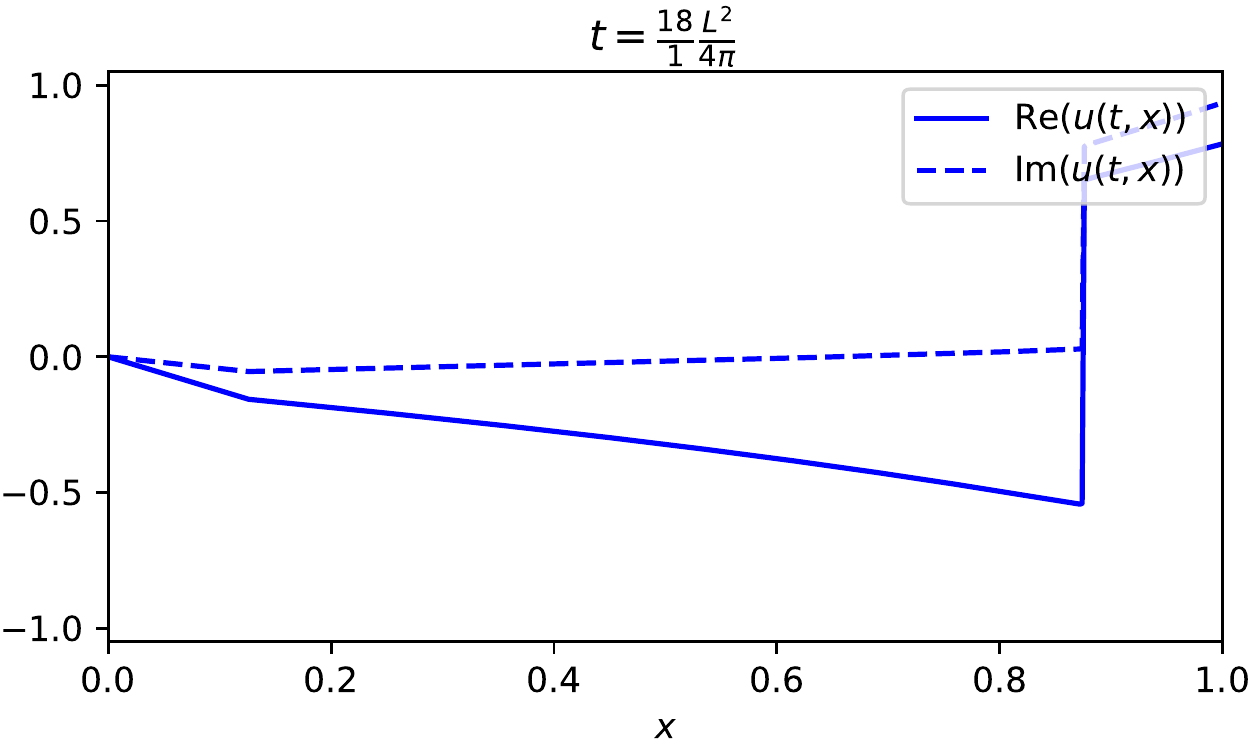}
    \subcaption{$\udsty t=\frac{18}{1}\frac{L^2}{4\pii}\approx1.44$}
  \end{minipage}
  \hfill
  \begin{minipage}[b]{.32\linewidth}
    \centering
    \includegraphics[width=\linewidth]{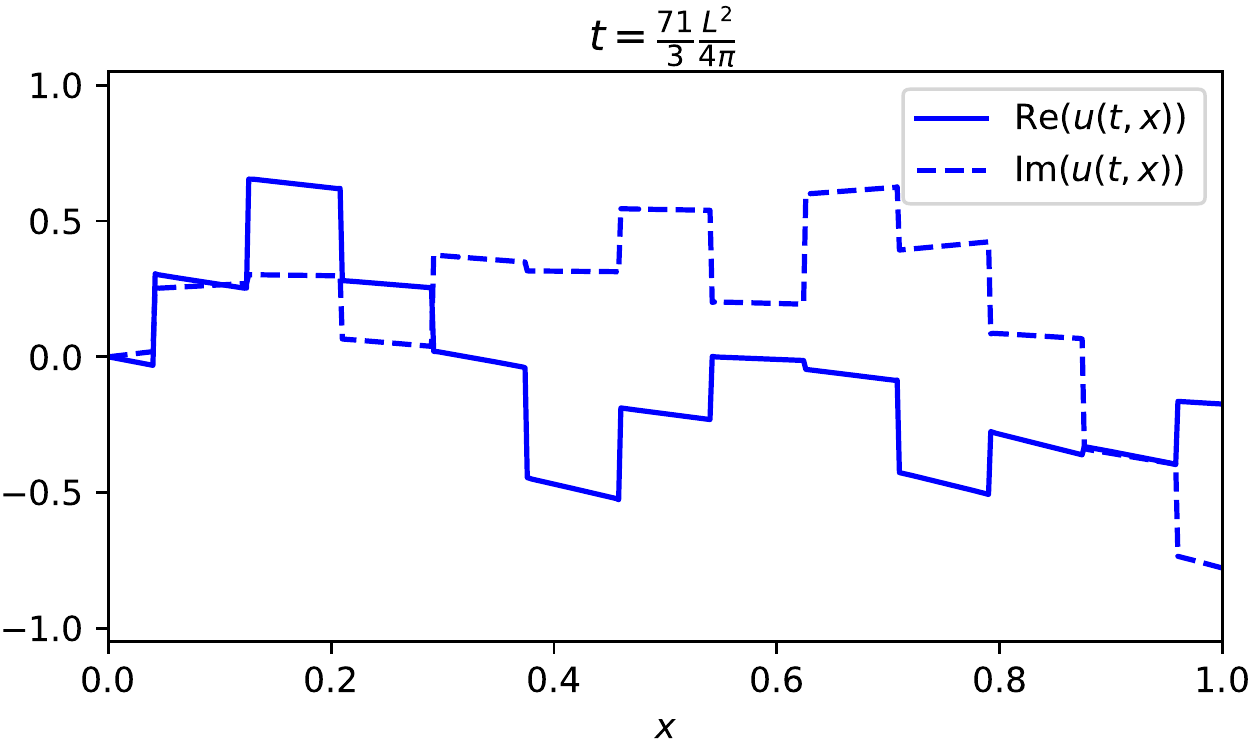}
    \subcaption{$\udsty t=\frac{71}{3}\frac{L^2}{4\pii}\approx1.89$}
  \end{minipage}
  \caption{
    The solution of the linear Schr\"{o}dinger equation with linear boundary conditions $\beta_{11}=-.7$, \ $\beta_{12}=1$, \ $\beta_{13}=0$, \ $\beta_{14}=0$, \ $\beta_{22}=0$, \ $\beta_{23}=0$, \ and \ $\beta_{24}=1$ on $[0,1]$ and box initial datum evaluated at ``rational" times which are commensurate with $L^2/(4\pi)$.
  }
  \label{fig:RobinDirichlet_imag_r}
\end{figure}

\subsection{Unstable boundary conditions}
In other settings when the $\kappa_j$ are farther from their asymptotic values and the problem is no longer self-adjoint, we can see energy growth and the Talbot effect does not seem to appear. An example is the boundary conditions
\Eq{UBC}
$$10\,u_x(t,1)-13\,u(t,1)+2\,u_x(t,0)-.1\,u(t,0)=0, \qquad 
19\,u(t,1)+u_x(t,0)-.1\,u(t,0)=0,$$
corresponding to 
$$\qeq{\beta_{11}=10,\\\beta_{12}=-13,\\\beta_{13}=2,\\\beta_{14}=-.1,\\\beta_{22}=19,\\\beta_{23}=1,\\\beta_{24}=.1.}$$
In this example there are two purely imaginary, and twelve complex $\kappa_j$; the latter come in three complex conjugate pairs and their opposites, $-\kappa_j$, as illustrated in Figure~\ref{fig:evals_gen_complex}. 
\begin{figure}[h!]
\centering
    \includegraphics[width=.5\linewidth]{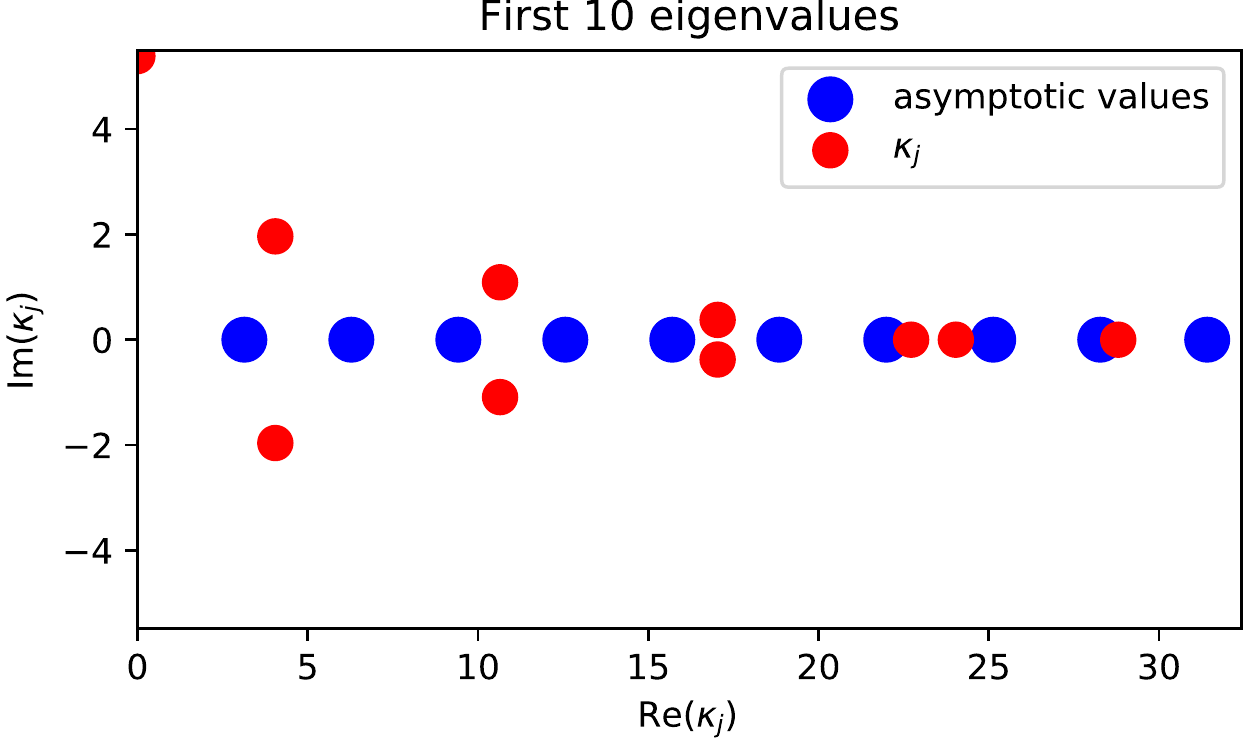}
    \caption{The first ten $\kappa_j$ found numerically in red and the asymptotic values they approach~\eqref{eqn:asym_lambda_nonzero} in blue for $\beta_{11}=10$, $\beta_{12}=-13$, $\beta_{13}=2$, $\beta_{14}=-.1$, $\beta_{22}=19$, $\beta_{23}=1$, $\beta_{24}=.1$, and $L=1$.}\label{fig:evals_gen_complex}
\end{figure}
The corresponding solution is plotted in Figure~\ref{fig:gen_complexeval_r}. Note the rapidly increasing ranges for the $y$-axis as $t$ increases, as necessitated by the dominance of the unstable modes.
\begin{figure}[h!]
  \centering
  \begin{minipage}[b]{.32\linewidth}
    \centering
    \includegraphics[width=\linewidth]{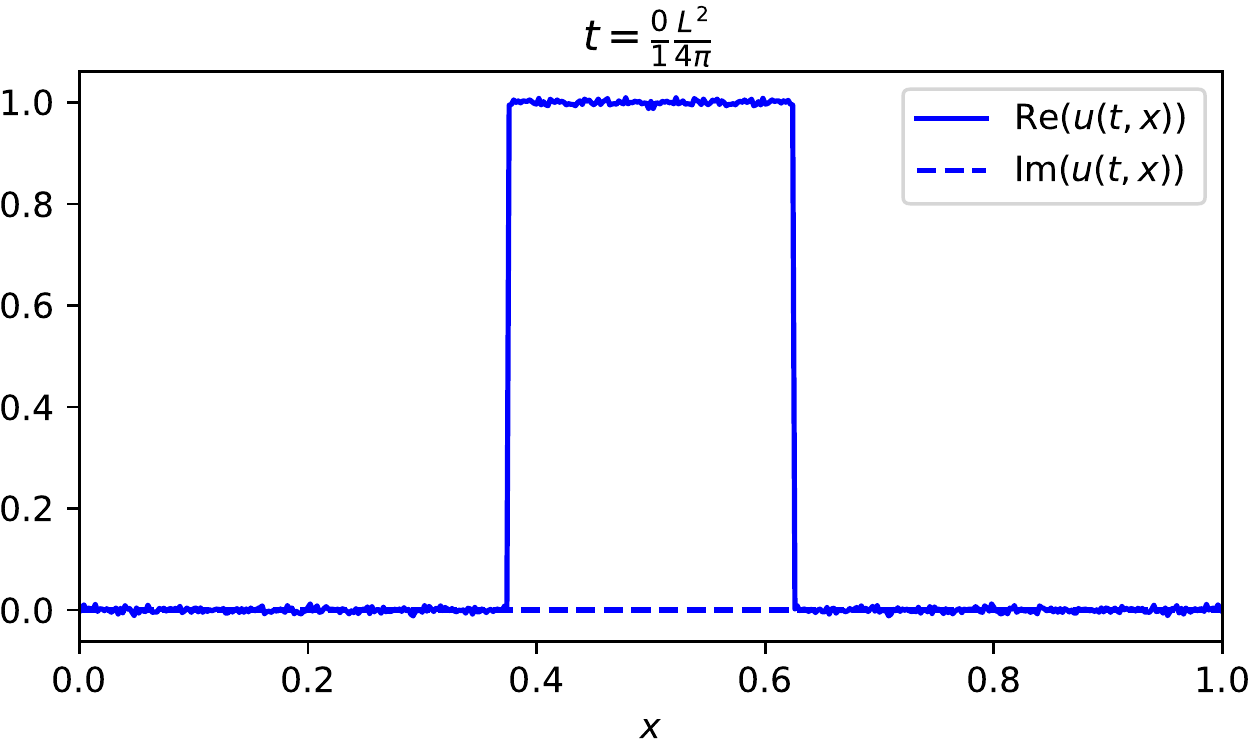}
    \subcaption{$t=0$\textcolor{white}{$\frac{0}{4\pii}$}}
  \end{minipage}
  \hfill
  \begin{minipage}[b]{.32\linewidth}
    \centering
    \includegraphics[width=\linewidth]{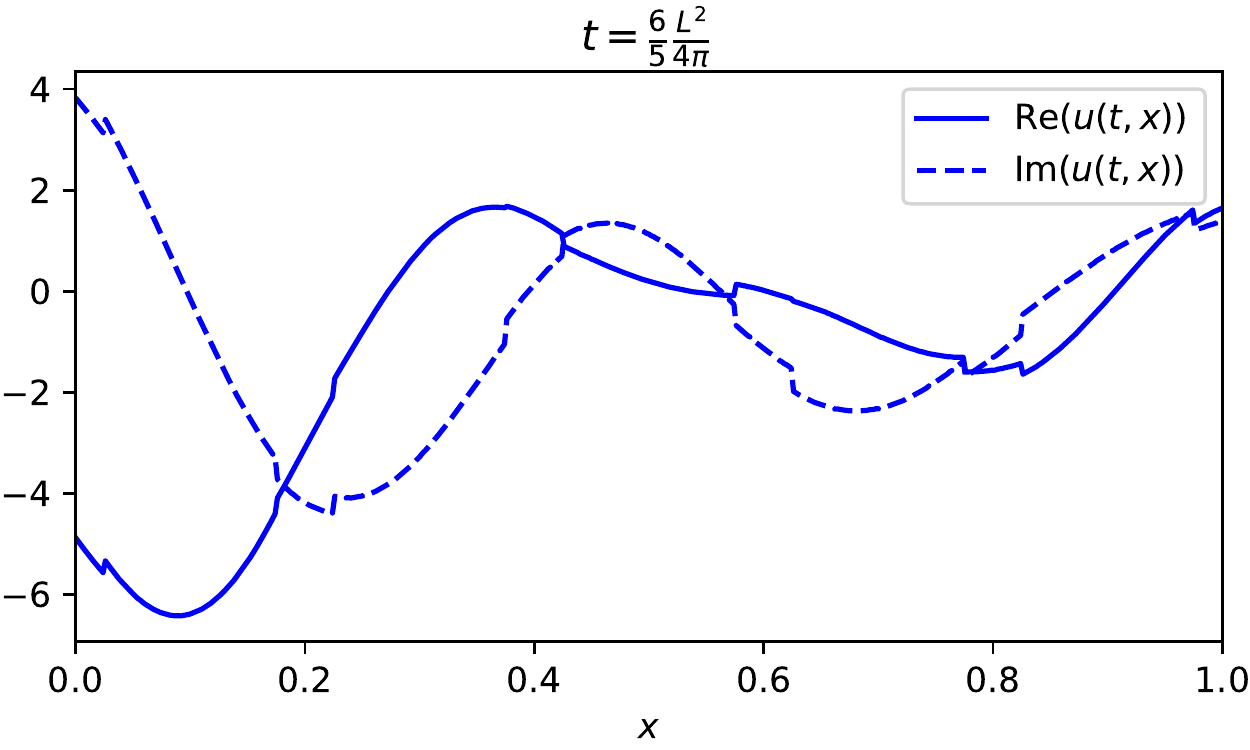}
    \subcaption{$\udsty t=\frac{6}{5}\frac{L^2}{4\pii}\approx0.09$}
  \end{minipage}
  \hfill
  \begin{minipage}[b]{.32\linewidth}
    \centering
    \includegraphics[width=\linewidth]{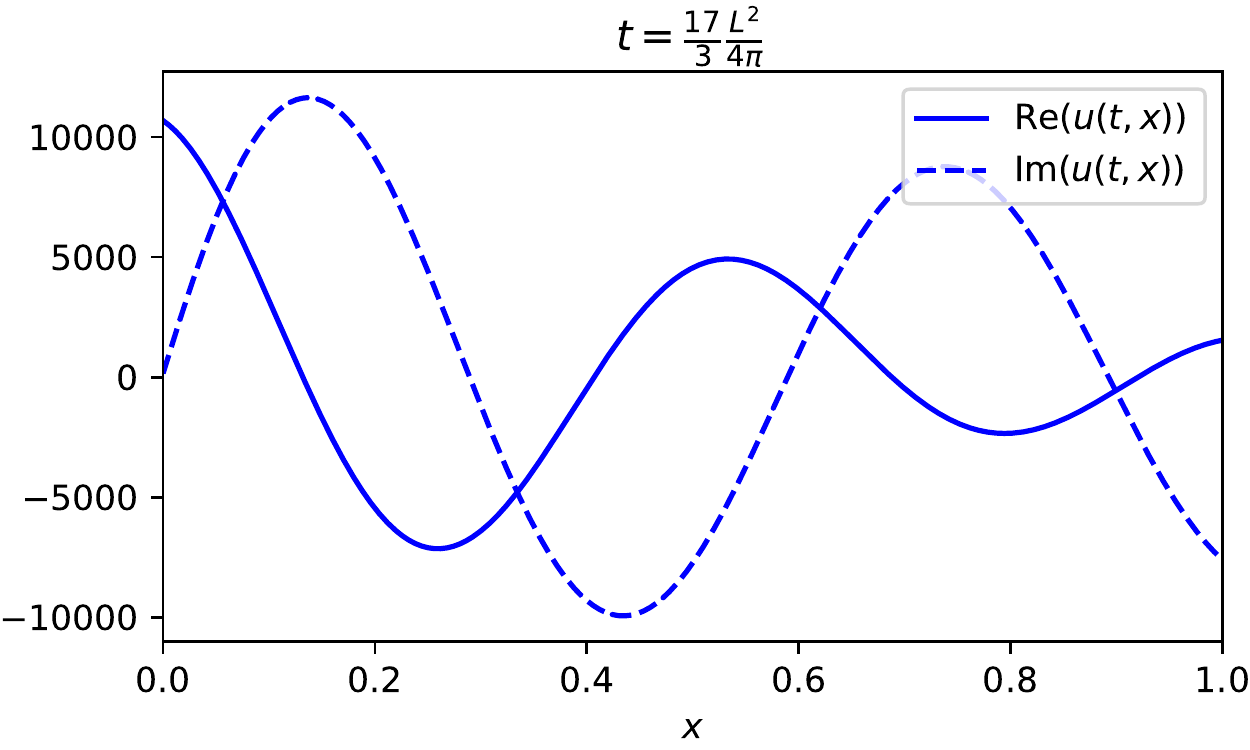}
    \subcaption{$\udsty t=\frac{17}{3}\frac{L^2}{4\pii}\approx0.45$}
  \end{minipage}
  \\
  \vspace{2ex}
  \begin{minipage}[b]{.32\linewidth}
    \centering
    \includegraphics[width=\linewidth]{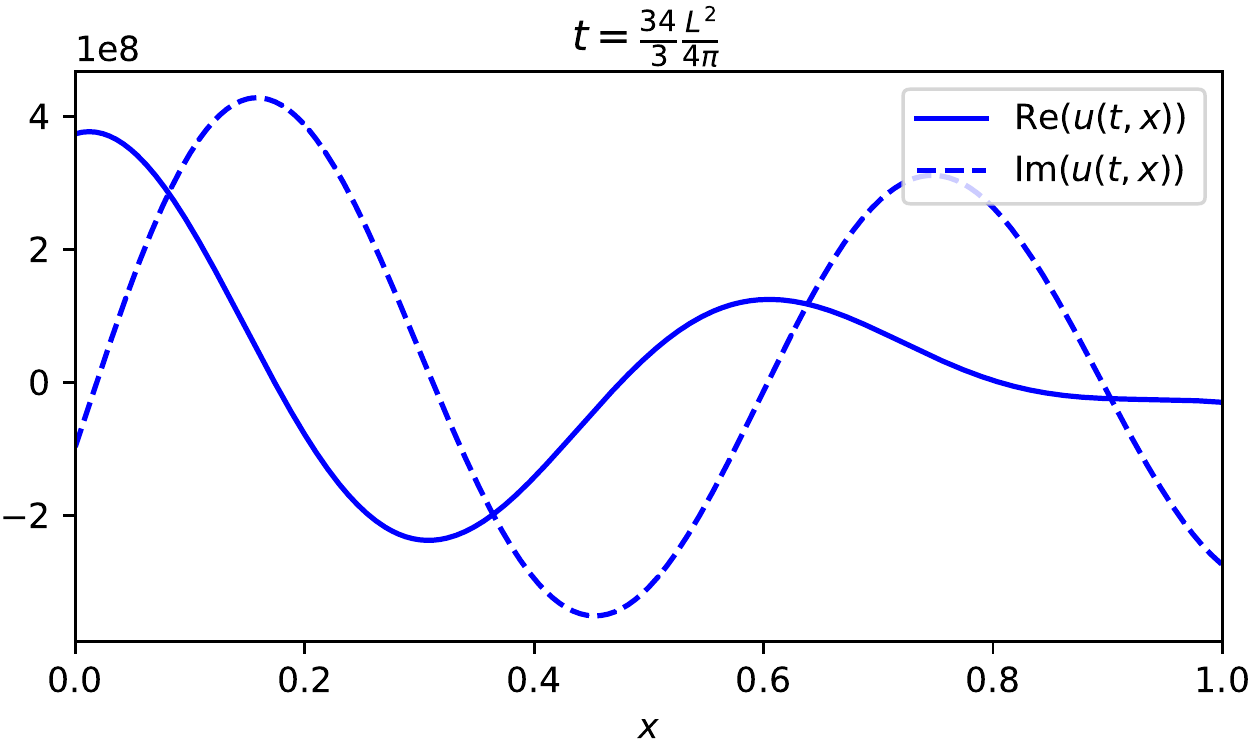}
    \subcaption{$\udsty t=\frac{34}{3}\frac{L^2}{4\pii}\approx0.9$}
      \end{minipage}
  \hfill
  \begin{minipage}[b]{.32\linewidth}
    \centering
    \includegraphics[width=\linewidth]{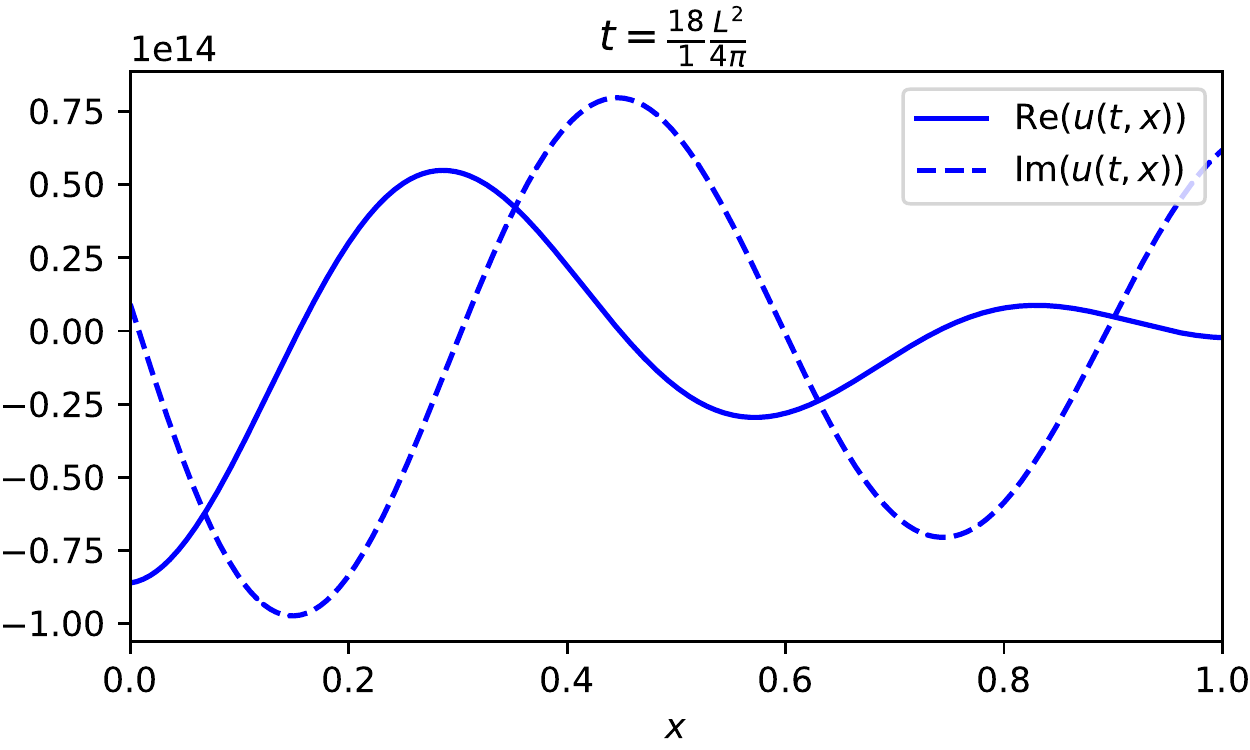}
    \subcaption{$\udsty t=\frac{18}{1}\frac{L^2}{4\pii}\approx1.44$}
  \end{minipage}
  \hfill
  \begin{minipage}[b]{.32\linewidth}
    \centering
    \includegraphics[width=\linewidth]{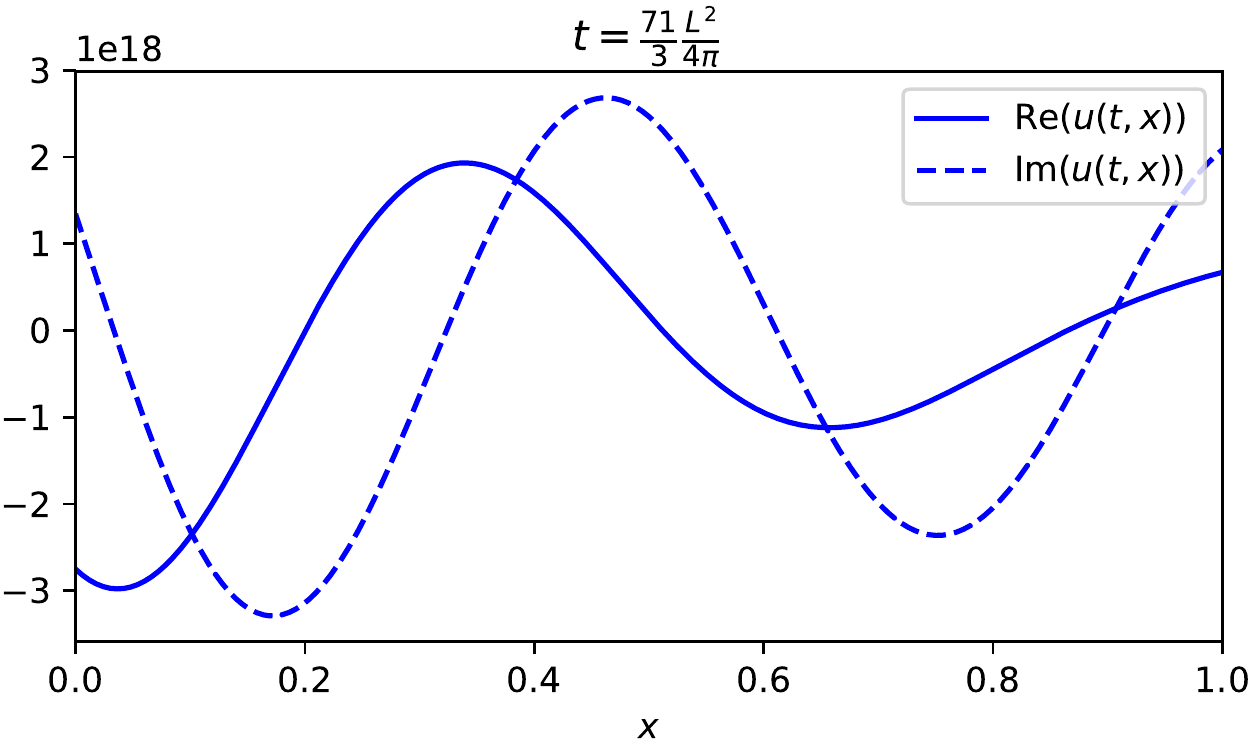}
    \subcaption{$\udsty t=\frac{71}{3}\frac{L^2}{4\pii}\approx1.89$}
  \end{minipage}
  \caption{
    The solution of the linear Schr\"{o}dinger equation with linear boundary conditions $\beta_{11}=10$, $\beta_{12}=-13$, $\beta_{13}=2$, $\beta_{14}=-.1$, $\beta_{22}=19$, $\beta_{23}=1$, and $\beta_{24}=.1$ on $[0,1]$ and box initial datum evaluated at ``rational" times which are commensurate with $L^2/(4\pi)$.
  }
  \label{fig:gen_complexeval_r}
\end{figure}

\subsection{Self-adjoint boundary conditions}


Next, let us look a the following self-adjoint boundary conditions
\Eq{SABC}
$$\qeq{\tfrac15\,u(t,L)+u(t,0)= 0,\\
  5\,u_x(t,1)+\tfrac12\,u(t,1)+u_x(t,0)=0.}$$
corresponding to
$$\qeq{\beta_{11}=5,\\\beta_{12}=1/2,\\\beta_{13}=1,\\\beta_{14}=0,\\\beta_{22}=1/5,\\\beta_{23}=0,\\\beta_{24}=1.}$$
In this case all $\kappa_j$ are real and they approach~\eqref{eqns:asym_lambda} very quickly.
The plots in Figure~\ref{fig:genBC_selfadjoint_r} show behavior that appears to exhibit revivals.
\begin{figure}[h!]
  \centering
  \begin{minipage}[b]{.32\linewidth}
    \centering
    \includegraphics[width=\linewidth]{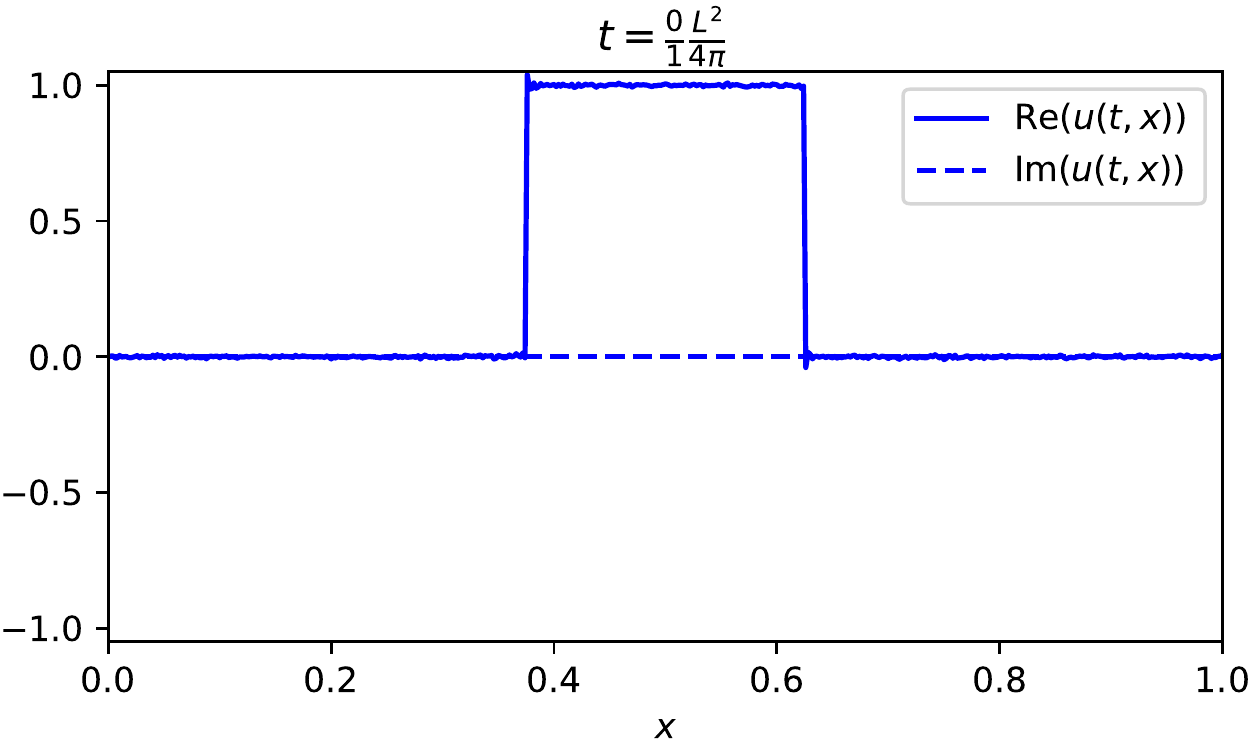}
    \subcaption{$t=0$ \textcolor{white}{$L^2/4\pii$}}
  \end{minipage}
  \hfill
  \begin{minipage}[b]{.32\linewidth}
    \centering
    \includegraphics[width=\linewidth]{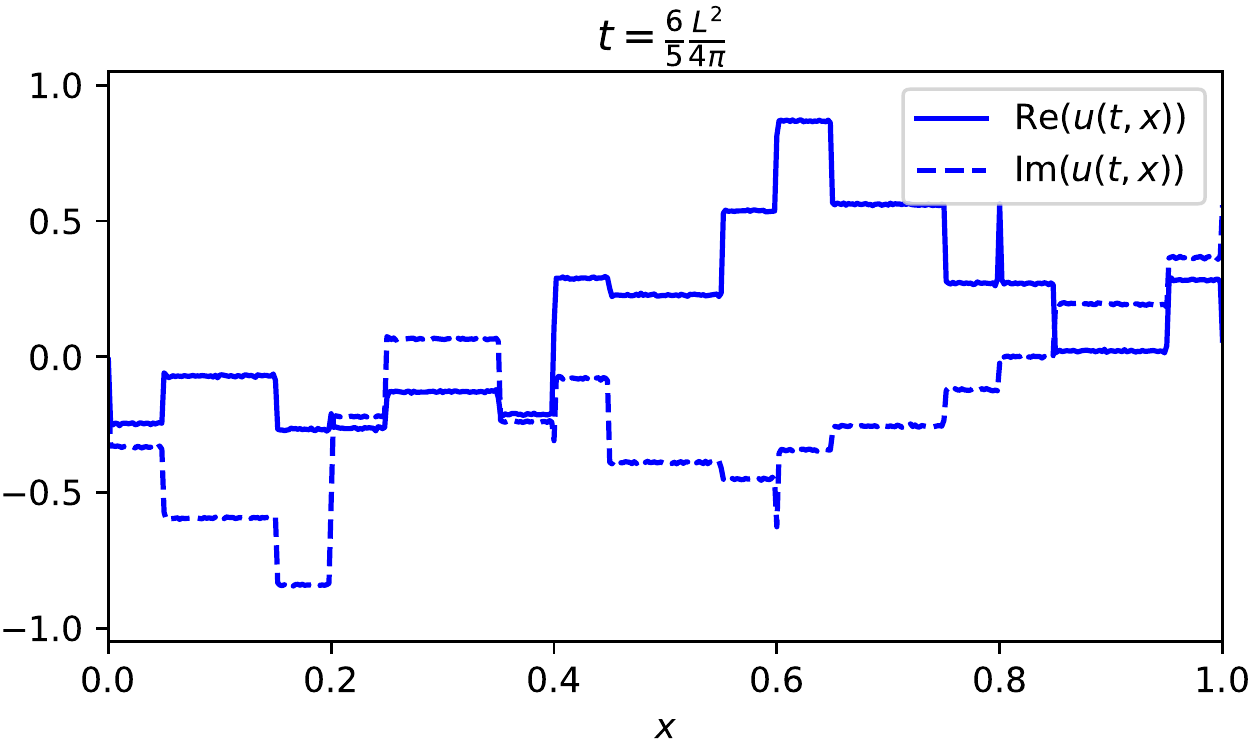}
    \subcaption{$\udsty t=\frac{6}{5}\frac{L^2}{4\pii}\approx0.09$}
  \end{minipage}
  \hfill
  \begin{minipage}[b]{.32\linewidth}
    \centering
    \includegraphics[width=\linewidth]{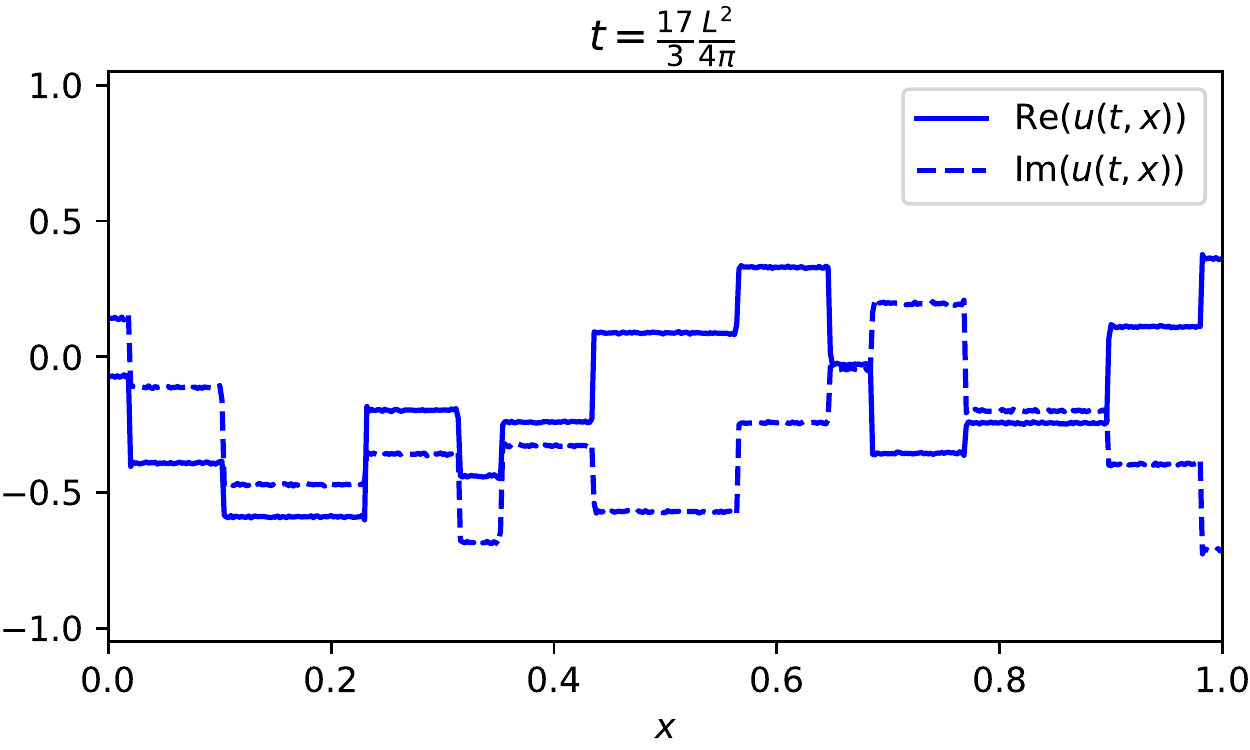}
    \subcaption{$\udsty t=\frac{17}{3}\frac{L^2}{4\pii}\approx0.45$}
  \end{minipage}
  \\
  \vspace{2ex}
  \begin{minipage}[b]{.32\linewidth}
    \centering
    \includegraphics[width=\linewidth]{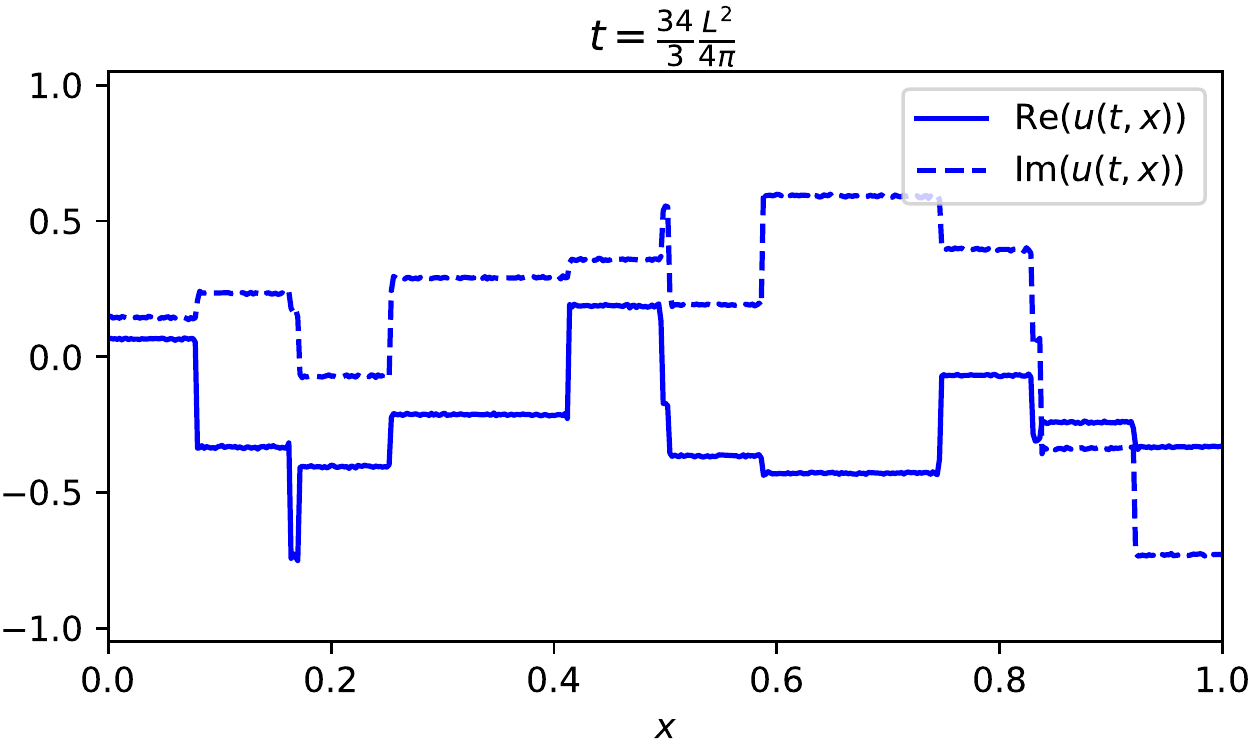}
    \subcaption{$\udsty t=\frac{34}{3}\frac{L^2}{4\pii}\approx0.9$}
      \end{minipage}
  \hfill
  \begin{minipage}[b]{.32\linewidth}
    \centering
    \includegraphics[width=\linewidth]{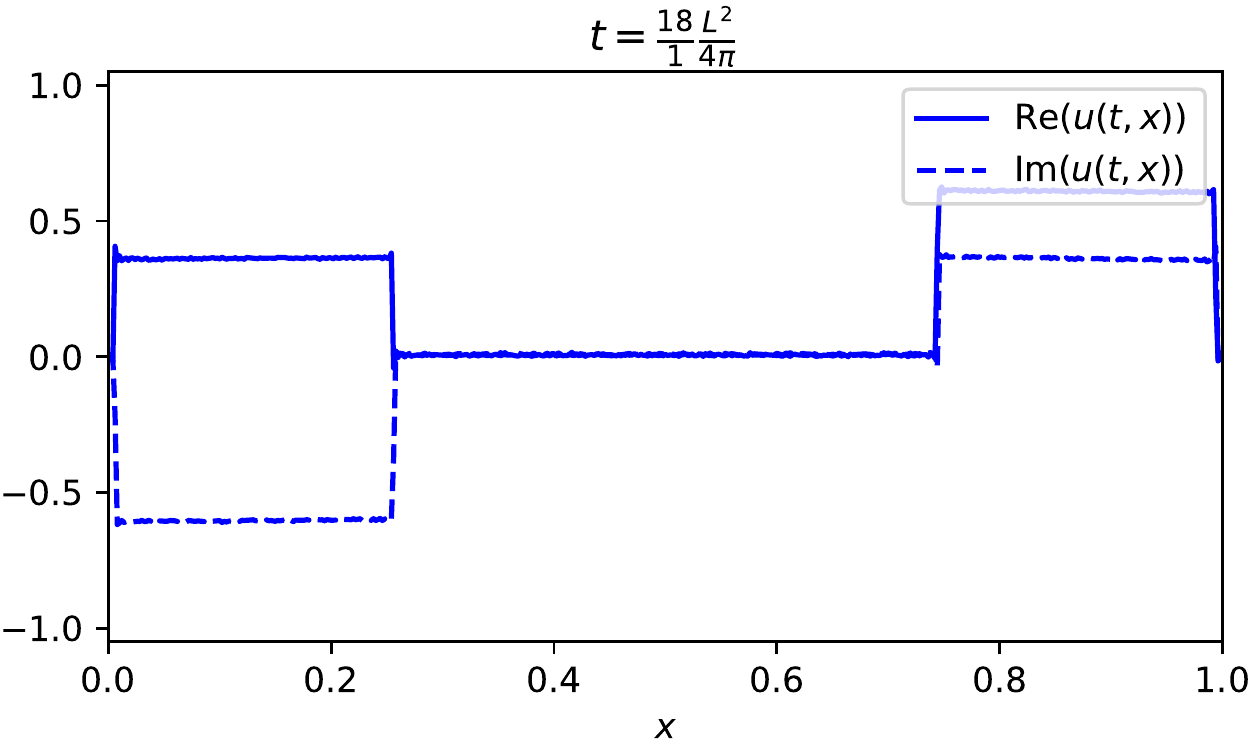}
    \subcaption{$\udsty t=\frac{18}{1}\frac{L^2}{4\pii}\approx1.44$}
  \end{minipage}
  \hfill
  \begin{minipage}[b]{.32\linewidth}
    \centering
    \includegraphics[width=\linewidth]{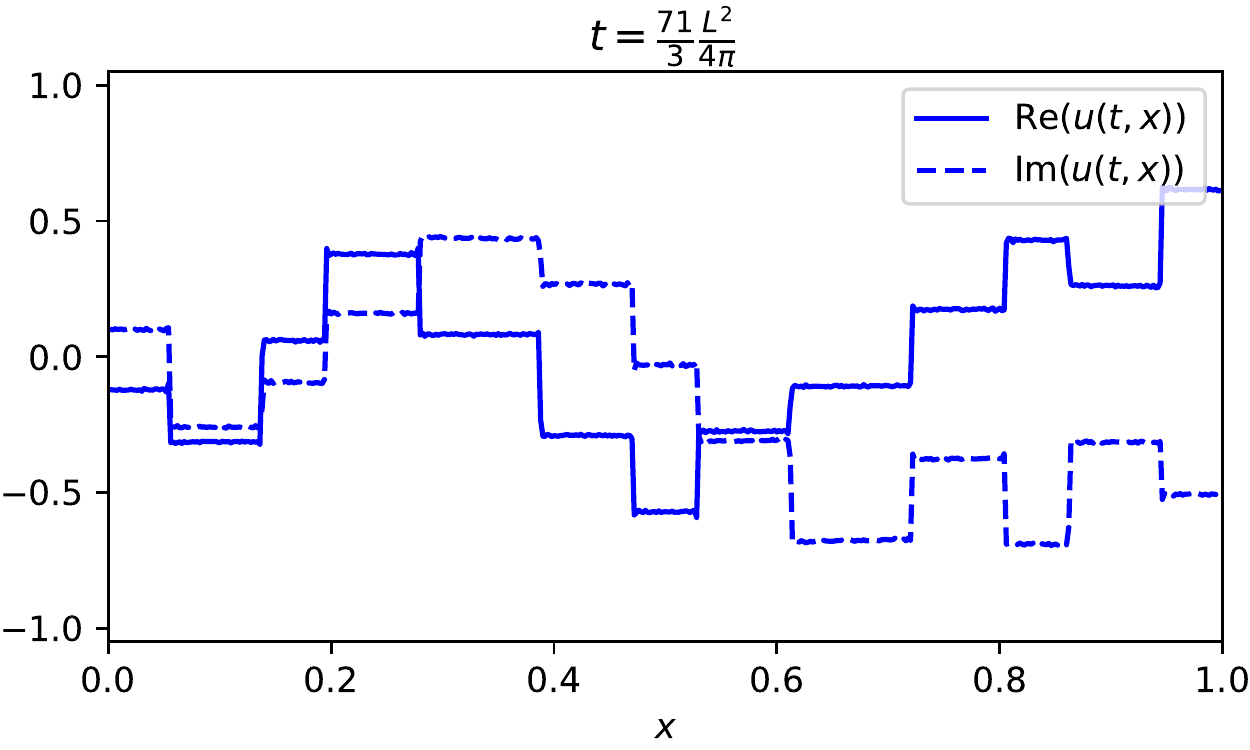}
    \subcaption{$\udsty t=\frac{71}{3}\frac{L^2}{4\pii}\approx1.89$}
  \end{minipage}
  \caption{
    The solution of the linear Schr\"{o}dinger equation with self-adjoint linear boundary conditions $\beta_{11}=5$, $\beta_{12}=1/2$, $\beta_{13}=1$, $\beta_{14}=0$, $\beta_{22}=1/5$, $\beta_{23}=0$, and $\beta_{24}=1$ on $[0,1]$ and box initial datum evaluated at ``rational" times which are commensurate with $L^2/(4\pi)$.
  }
  \label{fig:genBC_selfadjoint_r}
\end{figure}

\subsection{Energy decreasing boundary conditions}
Finally, we consider the case of boundary conditions for a non self-adjoint problem where energy leaks out.
Consider the case
$$-4\,u_x(t,1)+\i u(t,1)=0, \qquad 
u(t,0)=0,$$
whereby
$$\qeq{\beta_{11}=-4,\\\beta_{12}=\i,\\\beta_{13}=\beta_{14}=\beta_{22}=\beta_{23}=0,\\\beta_{24}=1.}$$
In this example all of the $\kappa_j$ satisfy $\Re(\kappa_j)\Im(\kappa_j)<0$ even as $\Im(\kappa_j)$ approaches $0$ for large $j$ as illustrated by Figure~\ref{fig:evals_energydecrease}.
\begin{figure}[h!]
\centering
    \includegraphics[width=.5\linewidth]{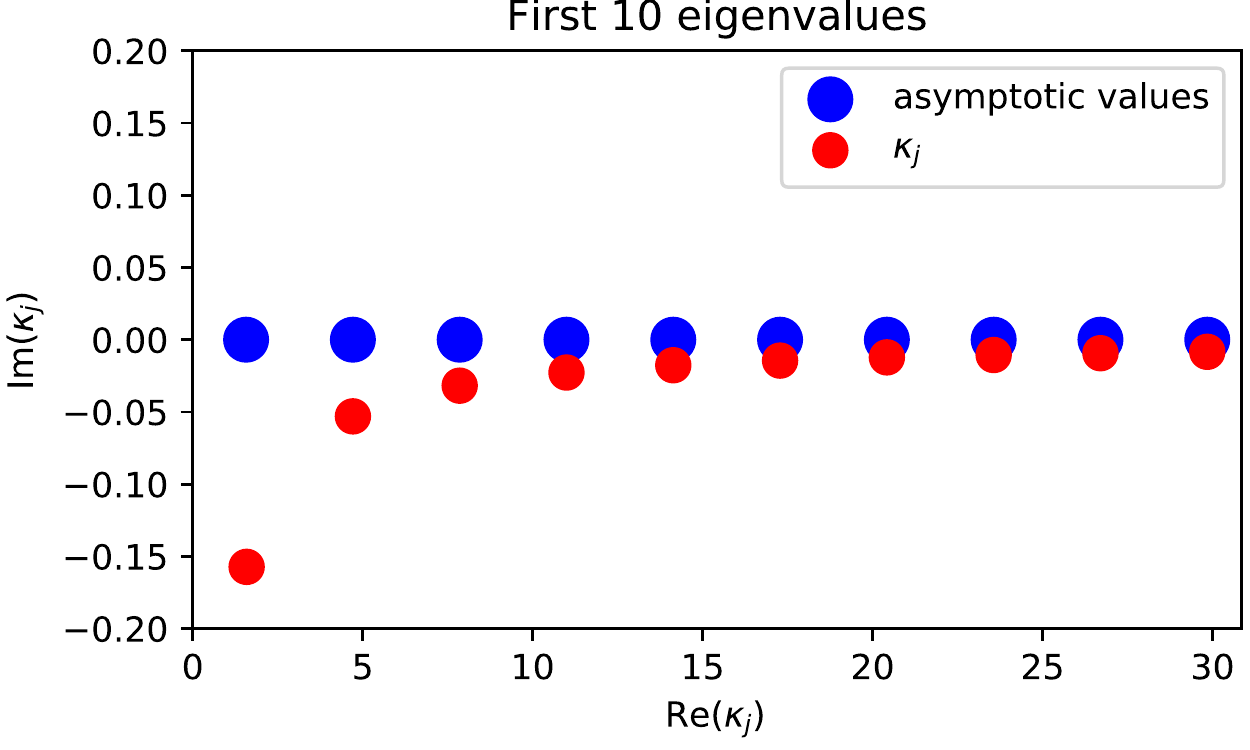}
    \caption{The first ten $\kappa_j$ found numerically in red and the asymptotic values they approach~\eqref{eqn:alt_asym_zero} in blue for $\beta_{11}=-4$, $\beta_{12}=\i$, $\beta_{13}=0$, $\beta_{14}=0$, $\beta_{22}=0$, $\beta_{23}=0$, $\beta_{24}=1$, and $L=1$.}\label{fig:evals_energydecrease}
\end{figure}

\begin{figure}[h!]
  \centering
  \begin{minipage}[b]{.32\linewidth}
    \centering
    \includegraphics[width=\linewidth]{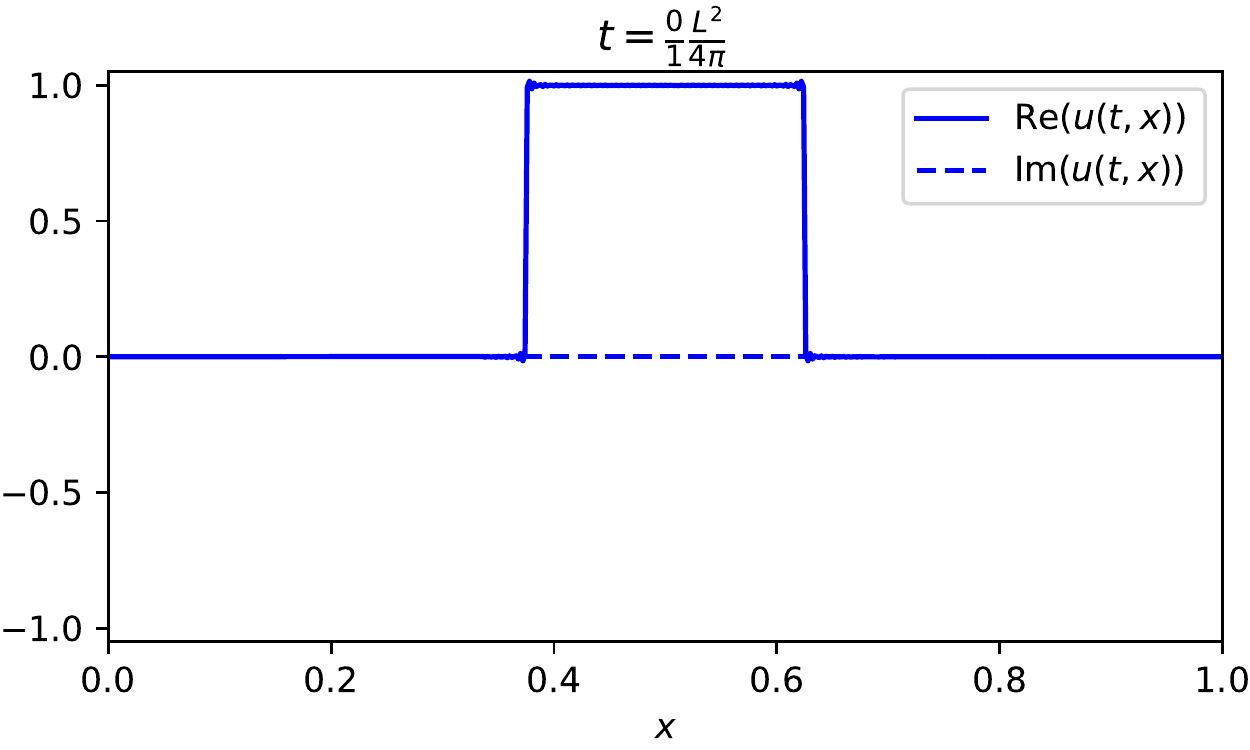}
    \subcaption{$t=0$\textcolor{white}{$\frac{0}{4\pii}$}}
  \end{minipage}
  \hfill
  \begin{minipage}[b]{.32\linewidth}
    \centering
    \includegraphics[width=\linewidth]{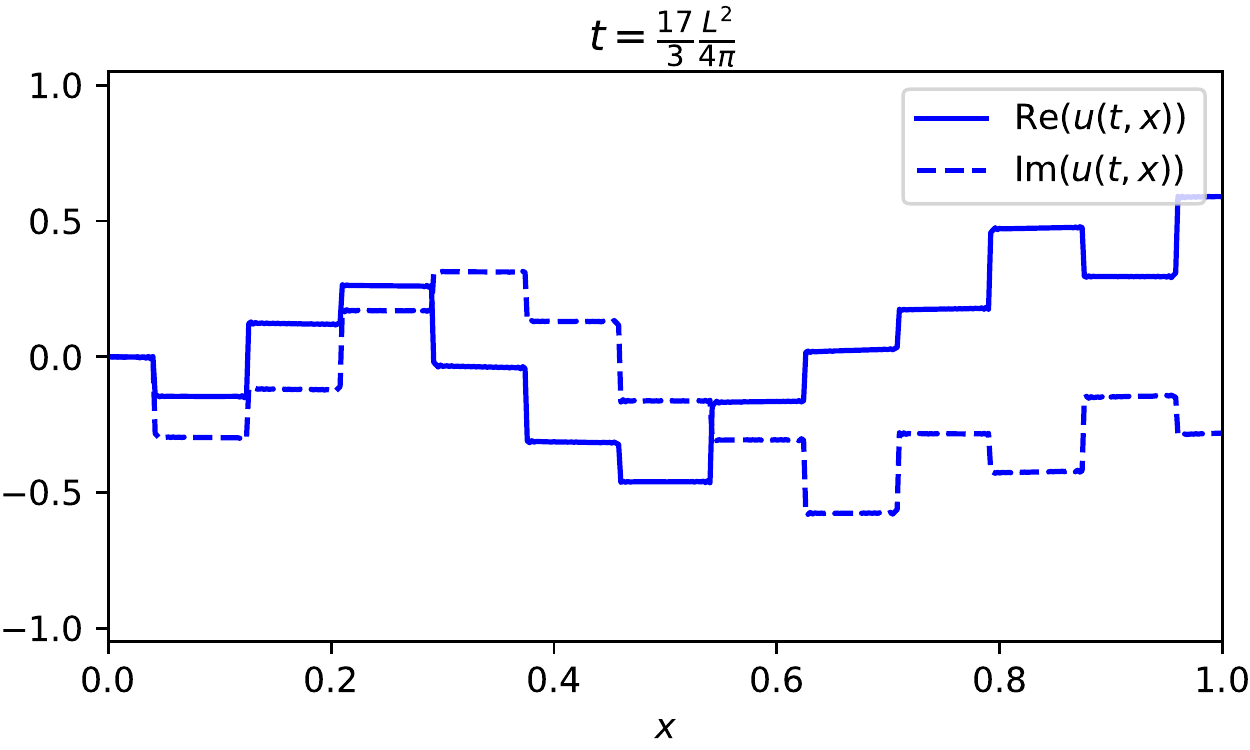}
    \subcaption{$\udsty t=\frac{17}{3}\frac{L^2}{4\pii}\approx0.45$}
  \end{minipage}
  \hfill
  \begin{minipage}[b]{.32\linewidth}
    \centering
    \includegraphics[width=\linewidth]{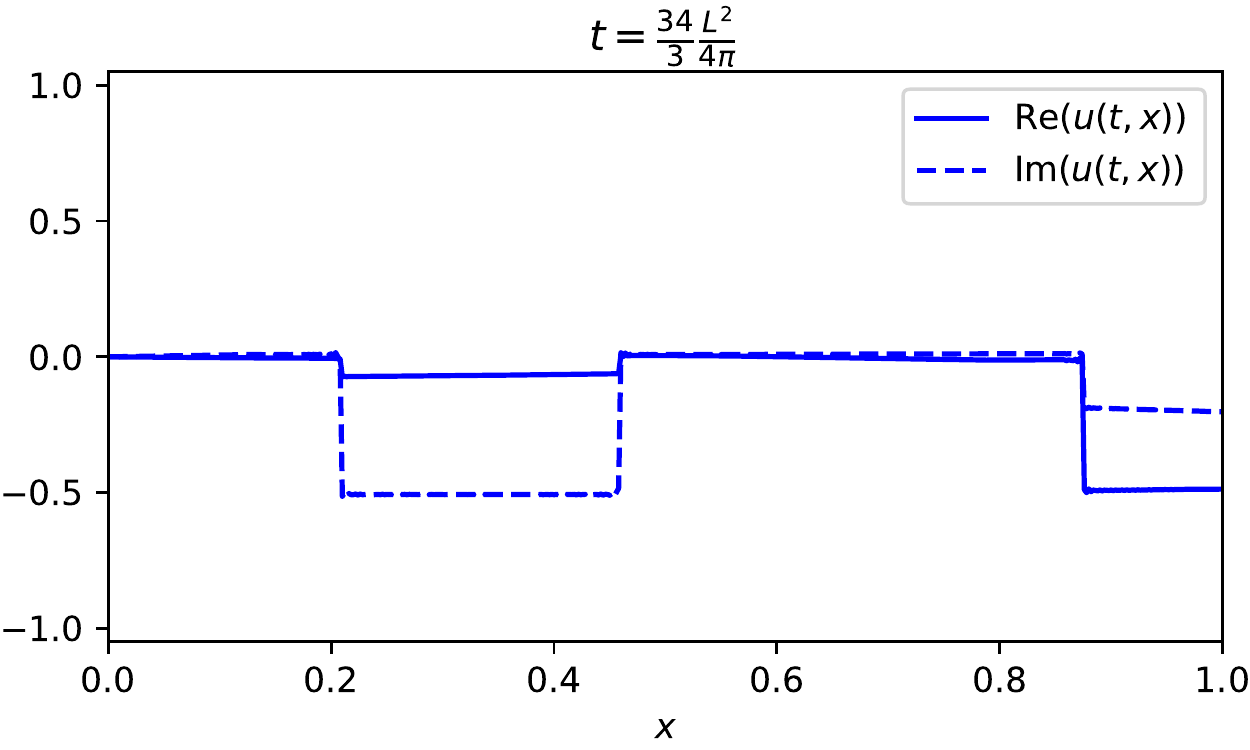}
    \subcaption{$\udsty t=\frac{34}{3}\frac{L^2}{4\pii}\approx0.9$}
      \end{minipage}
        \\
  \vspace{2ex}
  \begin{minipage}[b]{.32\linewidth}
    \centering
    \includegraphics[width=\linewidth]{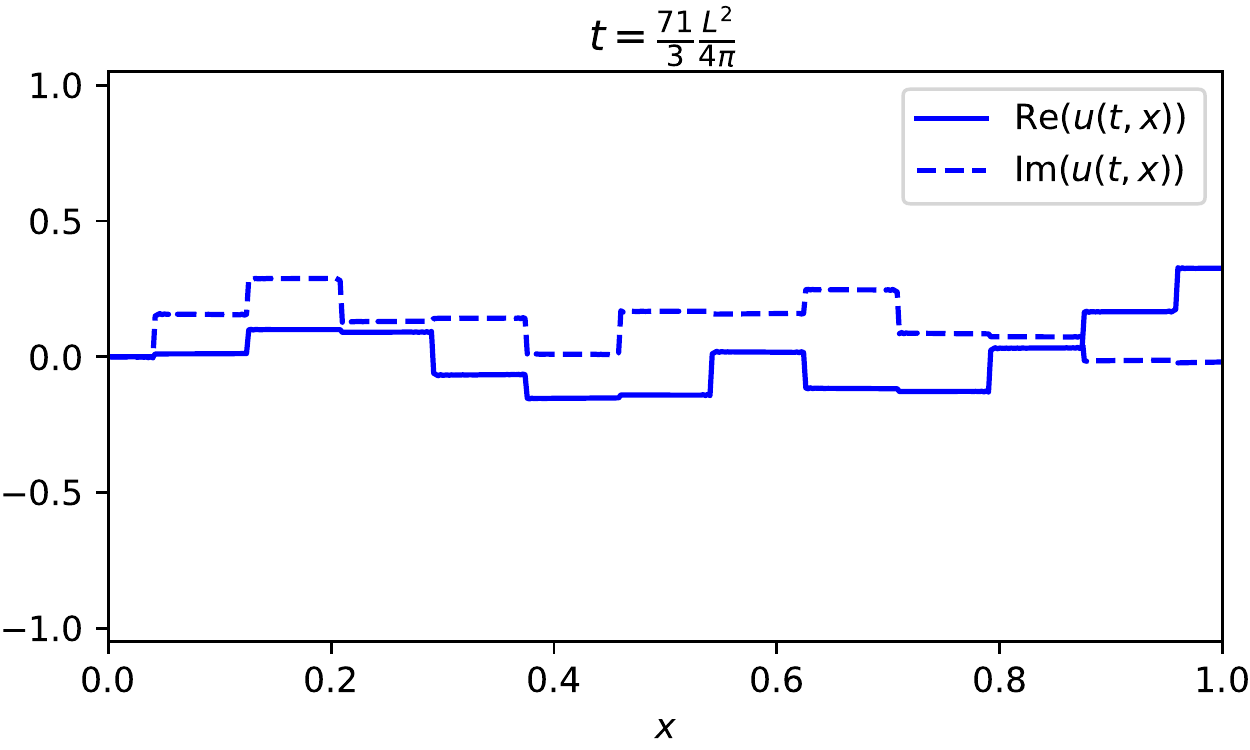}
    \subcaption{$\udsty t=\frac{71}{3}\frac{L^2}{4\pii}\approx1.89$}
  \end{minipage}
  \hfill
  \begin{minipage}[b]{.32\linewidth}
    \centering
    \includegraphics[width=\linewidth]{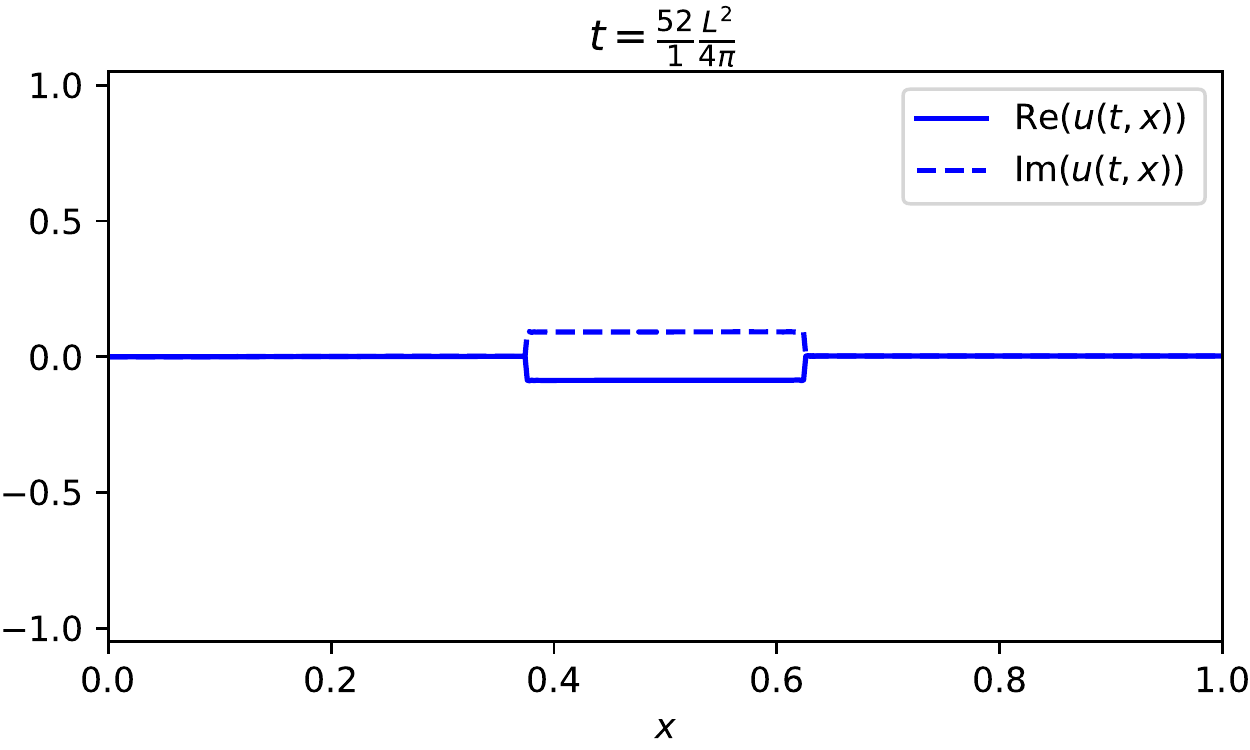}
    \subcaption{$\udsty t=\frac{52}{1}\frac{L^2}{4\pii}\approx4.14$}
  \end{minipage}
    \hfill
  \begin{minipage}[b]{.32\linewidth}
    \centering
    \includegraphics[width=\linewidth]{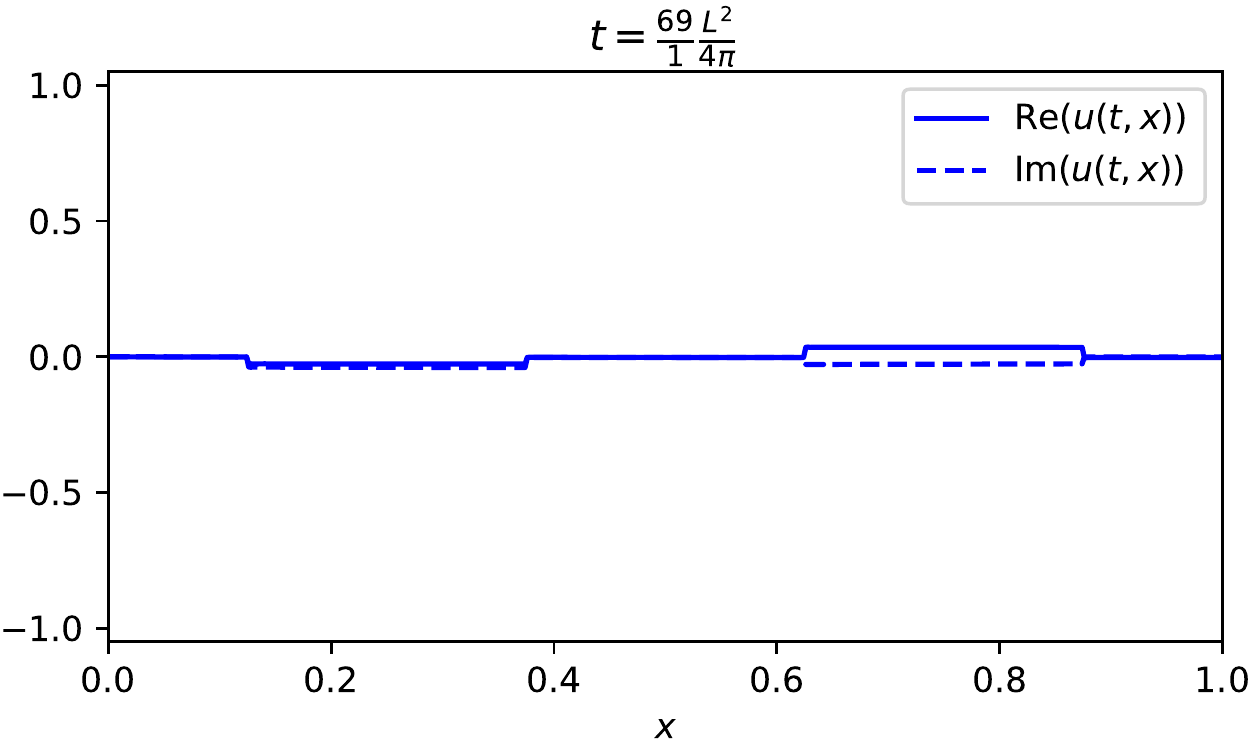}
    \subcaption{$\udsty t=\frac{69}{1}\frac{L^2}{4\pii}\approx5.49$}
  \end{minipage}
  \caption{
    The solution of the linear Schr\"{o}dinger equation with linear boundary conditions $\beta_{11}=-4$, $\beta_{12}=\i$, $\beta_{13}=0$, $\beta_{14}=0$, $\beta_{22}=0$, $\beta_{23}=0$, and $\beta_{24}=1$ on $[0,1]$ and box initial datum evaluated at ``rational" times which are commensurate with $L^2/(4\pi)$.
  }
  \label{fig:genbc_energyleak_r}
\end{figure}
It is rather surprising that although the solution decays to zero as $t \to \infty $, unlike the dissipative heat equation, it does not smooth out. 
Further numerical experiments suggest that the solution remains piecewise linear when the initial data is also piecewise linear. 
We have thus uncovered a new phenomenon that could be named ``dissipative(-dispersive) revival''.

In contrast, suppose that we tried to apply the framework to discern a Talbot type effect for the heat equation with periodic boundary data.
A classical smoothness argument is sufficient to show that the heat equation cannot exhibit dispersive quantisation; moreover, it is not hard to see that there cannot be a Talbot type effect even for smooth initial data.

Given the preceding plots in this section and further numerical experimentation, it seems that revivals are present whenever the problem is self-adjoint, and dissipative revivals appear when the boundary conditions cause the energy to decrease.
At this early stage, we cannot prove this and, indeed, would need to develop a new method to study~\eqref{eqn:seressoln_genBC} since the ideas presented in Section~\ref{sec:analysis} no longer hold.
Another possible characterization of revivals would be the ``nearness" of the $\kappa_j$ to their asymptotic values given in~\eqref{eqns:asym_lambda}.
This is an interesting problem for future study that we plan to consider.

\section{Concluding remarks}


We have investigated how the Talbot effect associated with the linear Schr\"odinger equation, originally detected in the context of periodic boundary conditions, arises in more general situations through the appearance of certain linear combinations of shifts and reflections of the initial datum at rational times.
To achieve this, we used two solution methods: generalised Fourier series and the Unified Transform Method. 
For the cases investigated in detail here, the generalised Fourier  series is cleaner and leads to the main results, but will not generalise as far as the UTM representation. 
Moreover, if the UTM approach produces complex integrals that can be  deformed down to the real line, then one is guaranteed that the eigenfunctions form a complete set. 
While this is not so crucial for second order problems of the type analyzed here, this will be when generalising to higher order equations.

We have shown how the Talbot effect manifests itself in the case of pseudoperiodic boundary conditions. 
More generally, it appears that self-adjointness or energy preservation is required for the effect to continue in its current form; boundary conditions that are not of this type but nevertheless well posed and with no purely imaginary eigenvalues produce solutions that eventually decay to zero due to the leakage of energy through the endpoints of the interval,  but do not smooth out but rather exhibit a form of ``dissipative revival.''
However, the rigorous justification of this observation as well as understanding the ``dissipative Talbot effect" remains an intriguing open question.

\section{Acknowledgements}
The authors wish to thank Beatrice Pelloni for her useful comments and her role in bringing this problem to our attention.


\bibliographystyle{plain}
\bibliography{./FullBib,./dbrefs}

\end{document}